\documentclass[a4paper,11pt]{article}
\usepackage[utf8]{inputenc}
\usepackage[T1]{fontenc}
\usepackage[UKenglish]{isodate}

\usepackage{bm}
\usepackage{bbm}
\usepackage{palatino}

\usepackage[margin=25mm]{geometry}

\usepackage[usenames,dvipsnames]{xcolor}
\usepackage{amsmath}
\usepackage{amssymb}
\usepackage{amsthm}  
\usepackage{ifthen}
\usepackage{mathrsfs}
\usepackage{mathabx}
\usepackage{stmaryrd}
\usepackage{array}
\usepackage{mathtools}
\usepackage{enumitem}
\usepackage{verbatim}
\usepackage{subfigure}
\usepackage{autobreak}

\usepackage{tikz}
\usetikzlibrary{shapes.geometric,plotmarks,backgrounds,fit,calc,circuits.ee.IEC}
\usetikzlibrary{decorations.pathreplacing}
\tikzset{phase/.style = {draw,fill,shape=circle,minimum size=5pt,inner sep=0pt},crossx/.style={path picture={ 
\draw[thick,black,inner sep=0pt]
(path picture bounding box.south east) -- (path picture bounding box.north west) (path picture bounding box.south west) -- (path picture bounding box.north east);
}}, cross/.style={path picture={ 
\draw[thick,black](path picture bounding box.north) -- (path picture bounding box.south) (path picture bounding box.west) -- (path picture bounding box.east);
}}, not/.style={draw,circle,cross,minimum width=0.3 cm}}

\usepackage[colorlinks=true,urlcolor=Mahogany,linkcolor=Mahogany,citecolor=Mahogany,plainpages=false,pdfpagelabels]{hyperref}
\usepackage[backend=bibtex,maxcitenames=4,maxalphanames=100,maxbibnames=100,isbn=false]{biblatex}

\newtheorem{theorem}{Theorem}
\newtheorem{lemma}[theorem]{Lemma}
\newtheorem{corollary}[theorem]{Corollary}
\newtheorem{proposition}[theorem]{Proposition}
\newtheorem{definition}[theorem]{Definition}

\newcommand{\sket}[1]{{\ensuremath{\lvert#1\rangle}}}
\newcommand{\lket}[1]{{\ensuremath{\left\lvert#1\right\rangle}}}
\newcommand{\ket}[1]{\mathchoice{\lket{#1}}{\sket{#1}}{\sket{#1}}{\sket{#1}}}

\newcommand{\sbra}[1]{{\ensuremath{\langle#1\rvert}}}
\newcommand{\lbra}[1]{{\ensuremath{\left\langle#1\right\rvert}}}
\newcommand{\bra}[1]{\mathchoice{\lbra{#1}}{\sbra{#1}}{\sbra{#1}}{\sbra{#1}}}

\newcommand{\ident}{\mathbbm{1}}
\DeclareMathOperator{\tr}{Tr}

\DeclareMathOperator{\argmax}{argmax}

\newcommand{\mbE}{\mathbb{E}}

\renewcommand{\otimes}{\varotimes}
\newcommand{\htwo}{\tilde{H}^{\downarrow}_2}
\newcommand{\hhalf}{H^{\uparrow}_{\frac{1}{2}}}
\newcommand{\eqdef}{:=}

\newcommand{\email}[1]{\href{mailto:#1}{#1}}

\renewcommand\footnotemark{} 

\begin{document}
\cleanlookdateon
\allowdisplaybreaks

\title{Polarization of Quantum Channels using Clifford-based Channel Combining%
\thanks{This work was supported in part by  the ``Investissements d’avenir'' (ANR-15-IDEX-02) program of the French National Research Agency. Ashutosh Goswami acknowledges the European Union’s Horizon 2020 research and innovation programme, under the Marie Skłodowska-Curie grant agreement No 754303.}}

\author{%
Frédéric~Dupuis,\thanks{Frédéric Dupuis was with Université de Lorraine, CNRS, Inria, LORIA, F-54000 Nancy, France. He is currently with Département d'Informatique et de Recherche Opérationnelle, Université de Montréal, Québec, Canada (\email{dupuisf@iro.umontreal.ca}).}
\ \ 
Ashutosh~Goswami,\thanks{Ashutosh Goswami is with Université Grenoble Alpes, Grenoble INP, LIG, F-38000 Grenoble, France (\email{ashutosh-kumar.goswami@univ-grenoble-alpes.fr}).}
\ \ 
Mehdi~Mhalla, 
\thanks{Mehdi Mhalla is with Université Grenoble Alpes, CNRS, Grenoble INP, LIG, F-38000 Grenoble, France (\email{mehdi.mhalla@univ-grenoble-alpes.fr}).}
\ 
Valentin~Savin \thanks{Valentin Savin is with Université Grenoble Alpes, CEA-LETI, F-38054 Grenoble, France (\email{valentin.savin@cea.fr}).}
}

\date{}

%

\maketitle

\begin{abstract}
We provide a purely quantum version of polar codes, achieving the symmetric coherent information of any qubit-input quantum channel.  Our scheme relies on a recursive channel combining and splitting construction, where a two-qubit gate randomly chosen from the Clifford group is used to combine two single-qubit channels. The inputs to the synthesized bad channels are frozen by preshared EPR pairs between the sender and the receiver, so our scheme is entanglement assisted. We further show that quantum polarization can be achieved by choosing the channel combining Clifford operator randomly, from a much smaller subset of only nine two-qubit Clifford gates. Subsequently, we show that a Pauli channel polarizes if and only if a specific classical channel over a four-symbol input set polarizes. 
We exploit this equivalence to prove fast polarization for Pauli channels, and to devise an efficient successive cancellation based decoding algorithm for such channels. 
Finally, we present a code construction based on chaining several quantum polar codes, which is shown to require a rate of preshared entanglement that vanishes asymptotically.
\end{abstract}


\section{Introduction}
Polar codes, proposed by Arikan~\cite{arikan09}, are the first explicit construction of a family of codes that provably achieve the channel capacity for any binary-input, symmetric, memoryless channel. His construction relies on a channel combining and splitting procedure, where a CNOT gate is used to combine two instances of the transmission channel. Applied recursively, this procedure allows synthesizing a set of so-called virtual channels from several instances of the transmission channel. When the code length goes to infinity, the synthesized channels tend to become either noiseless (good channels) or completely noisy (bad channels), a phenomenon which is known as ``channel polarization''. Channel polarization can effectively be exploited by transmitting messages via the good channels, while freezing the inputs to the bad channels to values known to both the encoder and decoder. This construction has been further generalized to classical channels with non-binary input alphabets in \cite{sta09}. Moreover, polar codes have been generalized for the transmission of classical information over classical-quantum channels in \cite{wg13-2}, and for transmitting quantum information in \cite{rdr11,wg13,rw12}. It was shown in \cite{rdr11} that the recursive construction of polar codes using a CNOT polarizes in both amplitude and phase bases for Pauli and erasure channels, and~\cite{rw12} extended this to general quantum channels. Then, a Calderbank-Shor-Steane (CSS)-like construction~\cite{cs96,steane96} was used to generalize polar codes for transmitting quantum information. This construction requires a small number of EPR pairs to be shared between the sender and the receiver, in order to deal with virtual channels that are bad in both amplitude and phase bases, thus making the resulting code entanglement-assisted~\cite{hsieh-devetak-brun2}. This construction was further refined in \cite{srdr13}, where preshared entanglement is completely suppressed at the cost of a more complicated multilevel coding scheme, in which polar coding is employed separately at each level.  However, all of these quantum channel coding schemes essentially exploit classical polarization, in either amplitude or phase basis.

\smallskip In this paper, we give a purely quantum version of polar codes, {\em i.e.},~a family of polar codes where the good channels are good as quantum channels, and not merely in one basis. Our construction uses a two-qubit gate chosen randomly from the Clifford group to combine two single-qubit channels, which bears similarities to the randomized channel combining/splitting operation proposed in~\cite{sta09}, for the polarization of classical channels with non-binary input alphabets. We show that the synthesized quantum channels tend to become either noiseless or completely noisy as quantum channels, meaning that their symmetric coherent information\footnote{Symmetric coherent (respectively, mutual) information refers the coherent (respectively, mutual) information of the quantum channel for a uniformly distributed input, as defined in Section~\ref{sec:preliminaries}.} tend to either $+1$ or $-1$.
Similar to the classical case, information qubits are transmitted through good (almost noiseless) channels, while the inputs to the bad (noisy) channels are ``frozen'' by sharing EPR pairs between the sender and the receiver. Thus, our scheme is entanglement assisted, for which the capacity is established in~\cite{BSST02}.
We show that the proposed scheme achieves  a quantum communication rate equal to half the symmetric mutual information of the quantum channel. The achieved net rate, defined as the quantum communication rate less the  entanglement consumption rate, is equal to the symmetric coherent information of the quantum channel. 
Further, we show that polarization can be achieved while reducing the set of two-qubit Clifford gates, used to randomize the channel combining operation, to a subset of nine Clifford gates only. 
We also present an efficient decoding algorithm for the proposed quantum polar codes for the particular case of Pauli channels. To a Pauli channel, we associate a classical symmetric channel, with both input and output alphabets given by the quotient of the $1$-qubit Pauli group by its centralizer, and show that the former polarizes quantumly if and only if the latter polarizes classically. This equivalence provides an alternative proof of the quantum polarization for a Pauli channel and, more importantly, an effective way to decode the quantum polar code for such channels, by decoding its classical counterpart. Fast polarization properties \cite{sta09,at09} are also proven for Pauli channels, by using techniques similar to those in~\cite{sta09}. Finally, we present a code construction based on chaining several quantum polar codes~\cite{BDH06}, which is shown to require a rate of preshared entanglement that vanishes asymptotically. We conclude by discussing some perspectives opened by the proposed construction,  which  we believe may complement or extend the classical CSS-based viewpoint.

\section{Preliminaries}
\label{sec:preliminaries}

Here are some basic definitions that we will need to prove the quantum polarization. First, we recall standard definitions of EPR pair and Clifford group.
\begin{definition}[Einstein-Podolsky-Rosen (EPR) pair]
An EPR pair on two qubit systems $A, A'$ is the quantum state, $\Phi_{AA'} = \ket{\Phi_{AA'}}\bra{\Phi_{AA'}}$, where,
 $$\ket{\Phi_{AA'}} = \frac{\ket{0_A0_{A'}} + \ket{1_A1_{A'}}}{\sqrt{2}}.$$
\end{definition}
\begin{definition}[Clifford group]
Let $P_n$ be Pauli group on $n$ qubits. The Clifford group $\mathcal{C}_n$ on $n$ qubits is the group of unitary transformations that take $P_n$ to $P_n$ via conjugation. Precisely,
\begin{equation*}
\mathcal{C}_n \eqdef \{ U \in U(2^n) \mid U \sigma U^\dagger \in P_n, \forall \sigma \in P_n \}  \,/\, U(1),
\end{equation*}
where global phase factors are ignored, since $U$ and $e^{i\varphi}U$ act in the same way.
\end{definition}

We will need the conditional sandwiched Rényi entropy of order 2, as defined by Renner~\cite{renner-phd}, and the conditional Petz-Rényi entropy of order $\frac{1}{2}$.
\begin{definition}[Conditional sandwiched Rényi entropy of order 2]
\label{def:renyi-2-entropy}
    Let $\rho_{AB}$ be a quantum state. Then,
    \[ \htwo(A|B)_{\rho} := -\log \tr\left[ \rho_B^{-\frac{1}{2}} \rho_{AB} \rho_B^{-\frac{1}{2}} \rho_{AB} \right], \]
    where $\rho_B := \tr_A(\rho_{AB})$ is the quantum state obtained by tracing out the $A$ system.
\end{definition}
\begin{definition}[Petz-Rényi entropy of order $\frac{1}{2}$]
\label{def:renyi-half-entropy}
    Let $\rho_{AB}$ be a quantum state. Then,
    \[ \hhalf(A|B)_{\rho} := 2 \log \sup_{\sigma_B} \tr\left[ \rho_{AB}^{\frac{1}{2}} \sigma^{\frac{1}{2}}_B \right], \]
    where the supremum is taken over all quantum states $\sigma_B$. 
\end{definition}
As shown in~\cite[Theorem 2]{tbh14}, those two quantities satisfy a duality relation: given a pure tripartite state $\rho_{ABE}$, $\htwo(A|B)_{\rho} = -\hhalf(A|E)_{\rho}$.

\medskip Throughout this work, we shall consider quantum channels $\mathcal{N}_{A' \rightarrow B}$, with  qubit input system $A'$ and output system $B$ of arbitrary dimension. When no confusion is possible, we shall discard the channel input and output systems from the notation. We will also need the symmetric coherent information of a quantum channel, and the concept of complementary channel.
\begin{definition}[Symmetric coherent information]
\label{def:symmetric-coherent-information}
    Let $\mathcal{N}_{A' \rightarrow B}$ be a channel with qubit input $A'$ and output $B$ of arbitrary dimension. The symmetric coherent information of  $\mathcal{N}$ is defined as the coherent information of the channel for a uniformly distributed input, that is
    $$I(\mathcal{N}) := -H(A|B)_{\mathcal{N}(\Phi_{A'A})} \in [-1, 1],$$
where $H(A|B)_{\rho_{AB}} := H(\rho_{AB}) - H(\rho_B)$, with $H(\sigma)$ being the Von Neumann entropy of a density matrix $\sigma$ and $\rho_B := \tr_A(\rho_{AB})$, and $\mathcal{N}(\Phi_{A'A}) := (\mathcal{N}\otimes I_A)(\Phi_{A'A})$ is the quantum state on the $AB$ system obtained by applying $\mathcal{N}$ on the $A'$-half of the EPR pair $\Phi_{A'A}$.
\end{definition}
\begin{definition}[Complementary channel]
\label{def:complementary-channel}
    Let $\mathcal{N}_{A' \rightarrow B}$ be a channel with qubit input $A'$ and output $B$ of arbitrary dimension, and let $U_{A' \rightarrow BE}$ be a Stinespring dilation of $\mathcal{N}$ (i.e.~a partial isometry such that $\mathcal{N}(\cdot) = \tr_E[U (\cdot) U^{\dagger}]$). The complementary channel of $\mathcal{N}$ is then $\mathcal{N}^c_{A' \rightarrow E}$, which is given by $\mathcal{N}^c(\cdot) := \tr_B[U (\cdot) U^{\dagger}]$. 
\end{definition}
Technically, Definition~\ref{def:complementary-channel} depends on the choice of the Stinespring dilation, so the complementary channel is only unique up to an isometry on the output system. However, this will not matter for any of what we do here.

\medskip Finally, we need the following lemma, providing necessary conditions for the convergence of a stochastic process. The lemma below is a slightly modified version of \cite[Lemma 2]{sta09}, so as to meet our specific needs. The proof is omitted, since it is essentially the same as the one in {\em loc. cit.} (see also \cite[Remark 1]{sta09}).
\begin{lemma}[{\cite[Lemma 2]{sta09}}]\label{lemma:stochastic_proc_convergence}
Suppose $B_i$, $i=1,2,\dots$ are independent and identically distributed (i.i.d.), $\{0,1\}$-valued random variables with $P(B_1 = 0) = P(B_1 = 1) = 1/2$, defined on a probability space $(\Omega, {\cal F}, P)$. Set ${\cal F}_0 = \{ \phi, \Omega\}$ as the trivial $\sigma$-algebra and set ${\cal F}_n$, $n\geq 1$, to be the $\sigma$-field generated by $(B_1, \dots, B_n)$. Suppose further that two stochastic processes $\{I_n : n\geq 0\}$ and $\{T_n : n\geq 0\}$ are defined on this probability space with the following properties:
\begin{description}
\item[{\rm\em(i.1)}] $I_n$ takes values in $[\iota_0, \iota_1]$ and is measurable with respect to ${\cal F}_n$. That is, $I_0$ is a constant, and $I_n$ is a
function of $B_1, \dots, B_n$.
\item[{\rm\em(i.2)}] $\{ (I_n, {\cal F}_n) : n \geq 0\}$ is a martingale.
\item[{\rm\em(t.1)}] $T_n$ takes values in the interval $[\theta_0, \theta_1]$ and is measurable with respect to ${\cal F}_n$.
\item[{\rm\em(t.2)}] $T_{n+1} \leq f(T_n)$ when $B_{n+1} = 1$,  for some continuous function $f: [\theta_0, \theta_1] \rightarrow [\theta_0, \theta_1]$, such that $f(\theta) < \theta, \forall \theta \in (\theta_0, \theta_1)$.
\item[{\rm\em(i\&t.1)}] For any $\epsilon > 0$ there exists $\delta > 0$, such that $I_n \in (\iota_0+\epsilon, \iota_1-\epsilon)$ implies $T_n \in (\theta_0+\delta, \theta_1-\delta)$. 
\end{description}

\smallskip \noindent Then, $\displaystyle I_\infty \eqdef \lim_{n\rightarrow\infty} I_n$ exists with probability 1, $I_\infty$ takes values in $\{\iota_0, \iota_1\}$, and $\mbE(I_\infty)\eqdef \iota_0 P(I_\infty = \iota_0) + \iota_1 P(I_\infty = \iota_1) = I_0$.
\end{lemma}

\section{Purely Quantum Polarization}\label{sec:purely_quantum_polarization}
In this section, we introduce our purely quantum version of polar codes, which is based on the channel combining and splitting operations depicted in Fig.~\ref{fig:channel-combine-split}. For the channel combining operation (Fig.~\ref{fig:combined-channel}), we consider a randomly chosen two-qubit Clifford unitary, to combine two independent copies of a quantum channel ${\cal W}$. The combined channel is then split, with the corresponding bad and good channels shown in Fig.~\ref{fig:bad-channel} and Fig.~\ref{fig:good-channel}, respectively. 
In other words, the bad channel $\mathcal{W} \boxast_C \mathcal{W}$ is a channel from $U_1$ to $Y_1 Y_2$ that acts as 
$$(\mathcal{W} \boxast_C \mathcal{W})(\rho) = \mathcal{W}^{\otimes 2}\left( C (\rho \otimes \frac{\ident}{2}) C^{\dagger} \right).$$
Likewise, the good channel $\mathcal{W} \varoast_C \mathcal{W}$ is a channel from $U_2$ to $R_1 Y_1 Y_2$ that acts as 
$$(\mathcal{W} \varoast_C \mathcal{W})(\rho) = \mathcal{W}^{\otimes 2}\left( C(\Phi_{R_1 U_1} \otimes \rho) C^{\dagger} \right),$$ where $\Phi_{R_1 U_1}$ is an EPR pair. 

\begin{figure}[!b]
    \centering
\subfigure[Combined channel]{%
\begin{tikzpicture}[scale=0.9]
        \draw 
            (0,0) node[draw] (canal1) {$\mathcal{W}$}
            (canal1) ++(0, -1) node[draw] (canal2) {$\mathcal{W}$}
            ($.5*(canal1)+.5*(canal2)$) ++(-1.5, 0) node[draw, minimum height=1.8cm] (C) {$C$}
            ;

        \draw
            (canal1 -| C.east) to node[above] {$X_1$} (canal1)
            (canal2 -| C.east) to node[above] {$X_2$} (canal2)
            (canal1 -| C.west) to ++(-0.4, 0) node[left] {$U_1$}
            (canal2 -| C.west) to ++(-0.4, 0) node[left] {$U_2$}
            (canal1.east) to ++(0.4, 0) node[right] {$Y_1$}
            (canal2.east) to ++(0.4, 0) node[right] {$Y_2$}
            ;
            
        \node[draw, dashed,fit=(C) (canal2) (canal1)] { };
    \end{tikzpicture}
    \label{fig:combined-channel}
}\hfill%
\subfigure[Bad channel, $\mathcal{W} \boxast_C \mathcal{W}$]{
    \begin{tikzpicture}[scale=0.9]
        \draw 
            (0,0) node[draw] (canal1) {$\mathcal{W}$}
            (canal1) ++(0, -1) node[draw] (canal2) {$\mathcal{W}$}
            ($.5*(canal1)+.5*(canal2)$) ++(-1.5, 0) node[draw, minimum height=1.8cm] (C) {$C$}
            ;

        \draw
            (canal1 -| C.east) to node[above] {$X_1$} (canal1)
            (canal2 -| C.east) to node[above] {$X_2$} (canal2)
            (canal1 -| C.west) to ++(-1.5, 0) node[left] {$U_1$}
            (canal2 -| C.west) to ++(-0.5, 0) node[left] (ground) {$\frac{\ident}{2}$}
            (canal1.east) to ++(0.4, 0) node[right] (y1) {$Y_1$}
            (canal2.east) to ++(0.4, 0) node[right] (y2) {$Y_2$}
            ;

        \node[draw, dashed,fit=(C) (canal2) (ground)] { };
    \end{tikzpicture}
    \label{fig:bad-channel}
}\hfill%
\subfigure[Good channel, $\mathcal{W} \varoast_C \mathcal{W}$]{
    \begin{tikzpicture}[scale=0.9]
        \draw 
            (0,0) node[draw] (canal1) {$\mathcal{W}$}
            (canal1) ++(0, -1) node[draw] (canal2) {$\mathcal{W}$}
            ($.5*(canal1)+.5*(canal2)$) ++(-1.5, 0) node[draw, minimum height=1.8cm] (C) {$C$}
            ;

        \draw
            (canal1 -| C.east) to node[above] {$X_1$} (canal1)
            (canal2 -| C.east) to node[above] {$X_2$} (canal2)
            (canal2 -| C.west) to ++(-2.8, 0) node[left] {$U_2$} 
            (canal1.east) to ++(0.4, 0) node[right] (y1) {$Y_1$}
            (canal2.east) to ++(0.4, 0) node[right] {$Y_2$}
            (canal1 -| C.west) to ++(-0.5, 0) to ++(-.5, .5) node[left] (EPR) {$\Phi_{R_1U_1}$} coordinate (coude) to ++(.5, .5) coordinate (topline) to (topline -| y1.west) node[right] (R) {$R_1$}
            ;

        \node[draw, dashed, fit=(EPR) (C) (coude) (topline) (canal2)] { };
    \end{tikzpicture}
    \label{fig:good-channel}
}
    \caption{Channel combining and splitting. (a) Combined channel: $C$ is a two-qubit Clifford unitary chosen at random. (b) Bad channel: we input a totally mixed
state into the second input. (c) Good channel: we input half of an EPR pair into the first input, and the other half becomes the output $R_1$.}
    \label{fig:channel-combine-split}
\end{figure}

\medskip The polarization construction is obtained by recursively applying the above channel combining and splitting operations. Let us denote ${\cal W}_C^{(0)} \eqdef {\cal W}\boxast_C{\cal W}$, ${\cal W}_C^{(1)} \eqdef {\cal W}\varoast_C{\cal W}$, where index $C$ in the above notation indicates the Clifford unitary used for the channel combining operation. To accommodate a random choice of $C$,  a classical description of $C$ must be included as part of the output of the bad/good channels at each step of the transformation.  To do so, for $i_1 = 0, 1$, we define
\begin{equation}\label{eq:recursion_first_step}
{\cal W}^{(i_1)}(\rho) = \frac{1}{|\mathcal{C}_2|} \sum_{C \in \mathcal{C}_2} \ket{C}\bra{C} \otimes {\cal W}_C^{(i_1)}(\rho), 
\end{equation}
where $\mathcal{C}_2$ denotes the Clifford group on two qubits,  $|\mathcal{C}_2|$ is the number of elements of $\mathcal{C}_2$ and is given by $|\mathcal{C}_2|=11520$ \cite{ozols2008clifford}, and $\{ \ket{C} \}_{C \in \mathcal{C}_2}$ denotes an orthogonal basis of some auxiliary system. 
Put differently,  ${\cal W}^{(i_1)}$ is the classical equiprobable mixture of quantum channels $\left\{ {\cal W}_C^{(i_1)} \mid C \in \mathcal{C}_2 \right\}$. It has qubit input system (namely, $U_1$ if $i_1 = 0$, or  $U_2$ if $i_1 = 1$), and composite output system consisting of the above auxiliary system and the output system of the ${\cal W}_C^{(i_1)}$ channels (namely, $Y_1Y_2$ if $i_1 = 0$, or $R_1Y_1Y_2$ if $i_1 = 1$).

\medskip Now, applying twice the  operation ${\cal W} \mapsto \left({\cal W}^{(0)},{\cal W}^{(1)}\right)$, we get channels ${\cal W}^{(i_1 i_2)} \eqdef \left( {\cal W}^{(i_1)} \right)\,\!^{(i_{2})}$,  where  $(i_1i_{2}) \in \{00, 01, 10, 11\}$.
In general, after $n$ levels of recursion, we obtain $2^n$ channels\footnote{Throughout this paper, we shall use a string-like notation for binary vectors, {\em e.g.}, $(i_1\cdots i_n)\in \{0,1\}^n$. In Section~\ref{sec:quantum_polar_coding}, binary strings  $(i_1\cdots i_n)\in \{0,1\}^n$ will be further identified to integers $i\in\{0,1,\dots,N-1\}$, where $N = 2^n$. This will be explicitly stated in the text, so as to avoid any possible confusion that might arise.}:
\begin{equation}\label{eq:recursive_construction}
{\cal W}^{(i_1\cdots i_{n})} \eqdef \left( {\cal W}^{(i_1\cdots i_{n-1})} \right)\,\!^{(i_{n})},  \ \ (i_1\cdots i_{n}) \in \{0, 1\}^n.
\end{equation}

Since channels ${\cal W}^{(i_1\cdots i_{n})}$ have qubit input system (and composite output system), it follows from Definition~\ref{def:symmetric-coherent-information} that their symmetric coherent information  $I\left( {\cal W}^{(i_1\cdots i_{n})}\right) \in [-1, +1]$ (see also Lemma~\ref{lem:i-and-r-for-mixture} below). Our main theorem below states that as $n$ goes to infinity, the symmetric coherent information of the synthesized channels ${\cal W}^{(i_1\cdots i_{n})}$ polarizes, meaning that it goes to either $-1$ or $+1$, except possibly for a vanishing fraction of channels.
 To prove the polarization theorem, we will utilize Lemma~\ref{lemma:stochastic_proc_convergence}. This basically requires us to find two quantities $I$ and $T$ that respectively play the roles of the symmetric mutual information of the channel and of the Bhattacharyya parameter from the classical case.

\smallskip   As mentioned above, for $I$ we shall consider the symmetric coherent information of the quantum channel.  
  For $T$, we will need to be slightly more creative.  Any choice (not necessarily unique) of $T$ would be convenient, as long as $I$ and $T$ satisfy the conditions in Lemma~\ref{lemma:stochastic_proc_convergence}. For any channel $\mathcal{N}_{A' \rightarrow B}$, let us define $R(\mathcal{N}) \in \left[\tfrac{1}{2}, 2\right]$ as
\begin{equation} \label{def:renyi_bhat}
R(\mathcal{N}) := 2^{\hhalf(A|B)_{\mathcal{N}(\Phi_{AA'})}} = 2^{-\htwo(A|E)_{\mathcal{N}^c(\Phi_{AA'})}}. 
\end{equation}
This quantity will be our $T$ and we will call it the ``Rényi-Bhattacharyya'' parameter. We can see  from the expression of $\hhalf$ that this indeed looks vaguely like the Bhattacharyya parameter; however we will work mostly with the second form involving the complementary channel as this will be more mathematically convenient for us.

\smallskip Before stating the main theorem, we first provide the following lemma on the symmetric coherent information $I$ and the  Rényi-Bhattacharyya parameter $R$ of a classical mixture of quantum channels.  It will allow us to derive the main steps in the proof of the polarization theorem, by conveniently working with the ${\cal W}_C^{(0)}(\rho)$/ ${\cal W}_C^{(1)}(\rho)$ construction, rather than the ${\cal W}^{(0)}(\rho)$/ ${\cal W}^{(1)}(\rho)$ mixture (in which a classical description of $C$ is included in the output). The proof is omitted, since part (a) is trivial, and  (b) follows easily from~\cite[Section B.2]{mdsft13}.

\begin{lemma}\label{lem:i-and-r-for-mixture}
Let ${\cal N}(\rho) = \sum_{x\in X} \lambda_x \ket{x}\bra{x} \otimes {\cal N}_x(\rho)$, be a classical mixture of quantum channels ${\cal N}_x$, where $\{\ket{x}\}_{x\in X}$ is some orthonormal basis of an auxiliary system, and $\lambda_x\geq 0,\forall x\in X$, with $\sum_{x\in X} \lambda_x = 1$. 
Then,
\begin{itemize}
\item[$(a)$] $I({\cal N}) = \mbE_X I({\cal N}_x) \eqdef \sum_{x\in X} \lambda_x I({\cal N}_x)$.
\item[$(b)$] $R({\cal N}) = \mbE_X R({\cal N}_x) \eqdef \sum_{x\in X} \lambda_x R({\cal N}_x)$.
\end{itemize}
\end{lemma}

\smallskip We can now state the polarization theorem.
\begin{theorem}\label{thm:quantum_polarization}
For any qubit-input quantum channel $\mathcal{W}$, let $I(\mathcal{W})$ be its symmetric coherent information, and  $\left\{\mathcal{W}^{(i_1 \cdots i_n)} : (i_1 \cdots i_n) \in \{0, 1\}^n \right\}$  be the set of virtual channels defined in (\ref{eq:recursive_construction}). Then, for any $\delta > 0$,
\begin{equation*}
\lim_{n\rightarrow\infty} \frac{\#\{(i_1\cdots i_{n}) \in \{0,1\}^{n} : I\left( {\cal W}^{(i_1\cdots i_{n})} \right) \in (-1+\delta, 1-\delta)  \}}{2^n} = 0,
\end{equation*}
and furthermore, 
\begin{equation*}
\lim_{n \rightarrow \infty} \frac{\#\left\{ (i_1 \cdots i_n) \in \{0,1\}^n : I(\mathcal{W}^{(i_1 \cdots i_n)}) \geqslant 1-\delta\right\} }{2^n} = \frac{I(\mathcal{W}) + 1}{2}.
\end{equation*}
\end{theorem}

\begin{proof}
Let $\{B_n : n\geq 1\}$ be a sequence of i.i.d., $\{0,1\}$-valued random variables with $P(B_n = 0) = P(B_n = 1) = 1/2$, as in Lemma~\ref{lemma:stochastic_proc_convergence}. Let $\{I_n : n\geq 0\}$ and $\{R_n : n\geq 0\}$ be the stochastic processes defined by  $I_n \eqdef I\left( {\cal W}^{(B_1\cdots B_{n})} \right)$ and $R_n\eqdef R\left( {\cal W}^{(B_1\cdots B_{n})} \right)$. By convention, ${\cal W}^{(\varnothing)} \eqdef {\cal W}$, thus $I_0 = I({\cal W})$ and $R_0 = R({\cal W})$. We prove that all the conditions of Lemma~\ref{lemma:stochastic_proc_convergence} hold for $I_n$ and $T_n\eqdef R_n$.

\begin{description}
    \item[(i.1)] Straightforward (with $[\iota_0, \iota_1] = [-1,1]$).
    
    \item[(i.2)] We must show that $I_n$ forms a martingale. In other words, that the channel combining and splitting transformation does not change the total coherent information, {\em i.e.}, $I\left({\cal W}^{(0)}\right) + I\left({\cal W}^{(1)}\right) = 2I\left({\cal W}\right)$. This follows from Lemma~\ref{lem:i-is-a-martingale} below, and Lemma~\ref{lem:i-and-r-for-mixture} (a).

    \item[(t.1)] Straightforward (with $[\theta_0, \theta_1] = [\frac{1}{2}, 2]$).
    
    \item[(t.2)] Here, we will show that $R_{n+1} = \frac{2}{5} + \frac{2}{5} R_n^2$, when $B_{n+1}=1$. It is enough to prove it for $n=0$ ({\em i.e.}, the first step of recursion), since in the general case the proof is obtained simply by replacing ${\cal W}$ with ${\cal W}^{(B_1\cdots B_n)}$. First, by using Lemma~\ref{lem:i-and-r-for-mixture} (b), and assuming $B_{1}=1$, we get $R_{1} \eqdef R\left( {\cal W}^{(1)} \right) =  \mbE_C R\left( {\cal W}_C^{(1)} \right) = \mbE_C R\left( \mathcal{W} \varoast_C \mathcal{W} \right)$, 
    where the last equality is simply a reminder of our notation ${\cal W}_C^{(1)} \eqdef \mathcal{W} \varoast_C \mathcal{W}$. We then prove that $ \mbE_C R\left(\mathcal{W} \varoast_C \mathcal{W}\right) = \frac{2}{5} + \frac{2}{5} R(\mathcal{W})^2$. This is where most of the action happens, and the proof is in Lemma~\ref{lem:d-good-channel}.

\item[(i\&t.1)]  For any $\varepsilon > 0$, there exists a $\delta > 0$ such that $I_n \in (-1+\varepsilon, 1-\varepsilon)$ implies that $R_n \in (\frac{1}{2}+\delta, 2-\delta)$. In other words, we need to show that if $R$ polarizes, then so does $I$. This holds for any choice of the Clifford unitary in the channel combining operation, and is proven in Lemma~\ref{lem:i-and-t-1}.
\end{description}

\noindent  By using Lemma~\ref{lemma:stochastic_proc_convergence}, we conclude that $I_\infty := \lim_{n\rightarrow\infty}I_n$ exists with probability $1$ and takes values in $\{-1, +1\}$, which implies the first limit, and $\mbE(I_\infty) := P(I_\infty = +1) - P(I_\infty = -1) = I(W)$, which implies the second limit (since also, $P(I_\infty = +1) + P(I_\infty = -1) = 1$).
\end{proof}

We now proceed with the lemmas. The following lemmas are stated in slightly more general settings, with the channel combining and splitting construction applied to two quantum channels ${\cal N}$ and ${\cal M}$, rather than to two copies of the same quantum channel ${\cal W}$. More precisely, the channels $\mathcal{N} \boxast_C \mathcal{M}$ and $\mathcal{N} \varoast_C \mathcal{M}$ are defined by using a similar construction to the one in Fig.~\ref{fig:channel-combine-split}, while applying ${\cal N}$ on the first  input (top), and ${\cal M}$ on the second  input (bottom). Thus, $(\mathcal{N} \boxast_C \mathcal{M})(\rho) = (\mathcal{N}\otimes \mathcal{M}) \left( C (\rho \otimes \frac{\ident}{2}) C^{\dagger} \right)$ and $(\mathcal{N} \varoast_C \mathcal{M})(\rho) = (\mathcal{N}\otimes \mathcal{M})\left( C(\Phi_{R_1 U_1} \otimes \rho) C^{\dagger} \right)$, where $\Phi_{R_1 U_1}$ is an EPR pair.

\begin{lemma} \label{lem:i-is-a-martingale}
    Given two channels $\mathcal{N}_{A'_1 \rightarrow B_1}$ and $\mathcal{M}_{A'_2 \rightarrow B_2}$ with qubit inputs, then 
    $$I(\mathcal{N} \varoast_C \mathcal{M}) + I(\mathcal{N} \boxast_C \mathcal{M}) = I(\mathcal{N}) + I(\mathcal{M}),$$ 
    and this holds for all choices of $C$.
\end{lemma}
\begin{proof}
    Consider the state 
    $\rho = (\mathcal{N} \otimes \mathcal{M})( C( \Phi_{A_1 A'_1} \otimes \Phi_{A_2 A'_2} ) C^{\dagger} )$ on systems $A_1 A_2 B_1 B_2$. We have that $I(\mathcal{N} \boxast_C \mathcal{M}) = -H(A_1|B_1 B_2)_{\rho}$ and $I(\mathcal{N} \varoast_C \mathcal{M}) = -H(A_2|A_1 B_1 B_2)_{\rho}$. Therefore, by the chain rule, 
\begin{align*}
I(\mathcal{N} \boxast_C \mathcal{M}) + I(\mathcal{N} \varoast_C \mathcal{M})  
      &= -H(A_1|B_1 B_2)_{\rho} - H(A_2|A_1 B_1 B_2)_{\rho}\\
      &= -H(A_1 A_2|B_1 B_2)_{\rho}. 
\end{align*}
Now, recall that the EPR pair has the property that $(Z \otimes \ident) \ket{\Phi} = (\ident \otimes Z^{\top}) \ket{\Phi}$ for any matrix $Z$,  where $Z^{\top}$ is the transpose of $Z$. Using this, we can move $C$ from the input systems $A'_1$ and $A'_2$ to the purifying systems $A_1 A_2$: $\rho = C^{\top}(\mathcal{N} \otimes \mathcal{M})(\Phi_{A_1 A'_1} \otimes \Phi_{A_2 A'_2}) \bar{C}$,  where $\bar{C}$ is the complex conjugate of $C$. Hence, we have 
\begin{align*}
    -H(A_1 A_2 | B_1 B_2)_{\rho} &= -H(A_1 A_2 | B_1 B_2)_{(\mathcal{N} \otimes \mathcal{M})(\Phi)}\\
    &= -H(A_1|B_1)_{\mathcal{N}(\Phi)} - H(A_2 | B_2)_{\mathcal{M}(\Phi)}\\
    &= I(\mathcal{N}) + I(\mathcal{M}),
\end{align*}
which completes the proof.
\end{proof}

\begin{lemma}\label{lem:d-good-channel}
Given two channels $\mathcal{N}_{A'_1 \rightarrow B_1}$ and $\mathcal{M}_{A'_2 \rightarrow B_2}$ with qubit inputs, then 
\[ \mbE_C R(\mathcal{N} \varoast_C \mathcal{M}) = \frac{2}{5} + \frac{2}{5} R(\mathcal{N}) R(\mathcal{M}), \]
where $C$ is the channel combining Clifford operator  and is chosen uniformly at random over the Clifford group.
\end{lemma}
\begin{proof}
 Let $\mathcal{N}^c_{A'_1 \rightarrow E_1}$ and $\mathcal{M}^c_{A'_2 \rightarrow E_2}$ be the complementary channels of $\mathcal{N}$ and $\mathcal{M}$ respectively. It can be shown that (see Lemma~\ref{lem:complementary_gd_bad} in Appendix~\ref{sec:proof-quantum-red-clifford}) 
 $$(\mathcal{N} \varoast_C \mathcal{M})^c(\rho) = (\mathcal{N}^c \otimes \mathcal{M}^c)\left( C \left( \frac{\ident_{A'_1}}{2} \otimes \rho \right) C^{\dagger} \right),$$
    and therefore $R(\mathcal{N} \varoast_C \mathcal{M}) = 2^{-\htwo(A_2|E_1 E_2)_{\rho}}$,  where $\rho_{A_2 E_1 E_2} = (\mathcal{N} \varoast \mathcal{M})^c(\Phi_{A_2 A'_2})$. Note that $\rho_{E_1 E_2} = \mathcal{N}^c\left(\frac{\ident}{2} \right)_{E_1} \otimes \mathcal{M}^c\left(\frac{\ident}{2} \right)_{E_2}$, which is independent of $C$. Now, to compute the expected value of this for a random choice of $C$, we proceed as follows:
\begin{align*} 
\mbE_C 2^{-\htwo(A_2|E_1 E_2)_{\rho}} &= \mbE_C \tr\left[ \left( \rho_{E_1 E_2}^{-\frac{1}{4}} \rho_{A_2 E_1 E_2} \rho_{E_1 E_2}^{-\frac{1}{4}} \right)^2 \right] \\
&= \mbE_C \tr\left[ \left( \rho_{E_1 E_2}^{-\frac{1}{4}} (\mathcal{N}^c \otimes \mathcal{M}^c)\left( C \left( \frac{\ident_{A'_1}}{2} \otimes \Phi_{A_2 A'_2} \right) C^{\dagger} \right) \rho_{E_1 E_2}^{-\frac{1}{4}} \right)^2 \right]
\end{align*}
    Now, note that this is basically the same calculation as in~\cite{fred-these}, at Equation (3.32) (there, $U$ is chosen according to the Haar measure over the full unitary group, but all that is required is a 2-design, and hence choosing a random Clifford yields the same result). However, since here we are dealing with small systems, we will not make the simplifications after (3.44) and (3.45) in \cite{fred-these} but will instead keep all the terms. We therefore get 
    $\mbE_C 2^{-\htwo(A_2|E_1 E_2)_{\rho}} = \alpha \tr\left[ \pi_{A_2}^2 \right] + \beta \tr\left[ \pi_{A'_1}^2 \otimes \Phi_{A_2 A'_2} \right]
        = \frac{1}{2} \alpha + \frac{1}{2} \beta$,
    where
    $\alpha = \frac{16}{15} - \frac{4}{15} 2^{-\htwo(A_1 A_2|E_1 E_2)_{\omega}}$, 
    $\beta = \frac{16}{15} 2^{-\htwo(A_1 A_2|E_1 E_2)_{\omega}} - \frac{4}{15}$, 
    and $\omega_{A_1 A_2 E_1 E_2} := (\mathcal{N}^c \otimes \mathcal{M}^c)(\Phi_{A_1 A'_1} \otimes \Phi_{A_2 A'_2})$.  Hence,
    \begin{align*}
        \mbE_C 2^{-\htwo(A_2|E_1 E_2)_{\rho}} &= \frac{6}{15} + \frac{6}{15} 2^{-\htwo(A_1 A_2|E_1 E_2)_{\omega}}\\
        &= \frac{6}{15} + \frac{6}{15}\, 2^{-\htwo(A_1 |E_1)_{\mathcal{N}^c\left(\Phi_{A_1 A'_1}\right)}} \,2^{-\htwo(A_2 |E_2)_{\mathcal{M}^c\left(\Phi_{A_2 A'_2}\right)}}\\
        &= \frac{2}{5} + \frac{2}{5} R(\mathcal{N}) R(\mathcal{M}).
    \end{align*}
    where we have used that the conditional sandwiched Rényi entropy of order 2 is additive with respect to tensor-product states,  which follows easily from Definition~\ref{def:renyi-2-entropy}.
\end{proof}

\begin{lemma}
    Let $\mathcal{N}_{A' \rightarrow B}$ be a channel with qubit input. Then,
    \begin{enumerate}
        \item $R(\mathcal{N}) \leqslant \frac{1}{2}+\delta \Rightarrow I(\mathcal{N}) \geqslant 1 - \log(1+2\delta)$.
        \item  $R(\mathcal{N}) \geqslant 2 - \delta \Rightarrow  I(\mathcal{N}) \leqslant -1 + \sqrt{2 \delta} + \left(1 + \sqrt{\delta/2}\right) h\left( \frac{\sqrt{\delta/2}}{1 + \sqrt{\delta/2}}\right)$,
    \end{enumerate}
    where $h(\cdot)$ denotes the binary entropy function.
    \label{lem:i-and-t-1}
\end{lemma}
\begin{proof}
    We first prove point 1). Observe that for any state $\sigma_{AB}$, the inequality $H(A|B)_{\sigma} \leqslant \hhalf(A|B)_{\sigma}$ holds. Now, for $\rho_{AB} = \mathcal{N}(\Phi_{AA'})$, we have that
    \begin{equation*}
    \frac{1}{2} + \delta \geqslant R(\mathcal{N})
        = 2^{\hhalf(A|B)_{\rho}}
        \geqslant 2^{H(A|B)_{\rho}}
        = 2^{-I(\mathcal{N})},
    \end{equation*}
where we have used $\hhalf(A|B)_{\rho} \geqslant H(A|B)_{\rho}$, which follows from the monotonically decreasing property of the conditional Petz-Rényi entropy with respect to its order \cite[Theorem 7]{mdsft13}. It follows that $I(\mathcal{N}) \geqslant 1 - \log(1+2\delta)$.

\medskip We now turn to the second point. We have that
\begin{align*}
        2 - \delta \leqslant R(\mathcal{N}) &= \max_{\sigma_B} \tr\left[ \rho^{\frac{1}{2}}_{AB} \sigma^{\frac{1}{2}}_B \right]^2\\
        &= 2 \max_{\sigma_B} \tr\left[ \sqrt{\rho_{AB}} \sqrt{\frac{\ident_A}{2} \otimes \sigma_B} \right]^2\\
        &\leqslant 2 \max_{\sigma_B} \left\| \sqrt{\rho_{AB}} \sqrt{\frac{\ident_A}{2} \otimes \sigma_B} \right\|_1^2\\
        &= 2 \max_{\sigma_B} F\left( \rho_{AB}, \frac{\ident_A}{2} \otimes \sigma_B \right)^2.
    \end{align*}
    Now, using the Fuchs-van de Graaf inequalities~\cite{fg99}, we get that there exists a $\sigma_B$ such that
    \[ \left\| \rho_{AB} - \frac{\ident_A}{2} \otimes \sigma_B \right\|_1 \leqslant \sqrt{2\delta}. \]
    We are now in a position to use the  Alicki-Fannes-Winter \cite[Lemma 2]{winter2016tight} inequality, which states that
    \begin{align*}
        \left| H(A|B)_{\rho} - 1 \right| \leqslant \sqrt{2 \delta} + \left(1 + \sqrt{\delta/2}\right) h\left( \frac{\sqrt{\delta/2}}{1 + \sqrt{\delta/2}}\right).
    \end{align*}
    This concludes the proof of the lemma.
\end{proof}

\section{Quantum Polarization Using Only Nine Clifford Gates}
\label{sec:polarization_9_cliffords}

In this section, we prove that quantum polarization can be achieved while reducing the set of two-qubit Clifford gates used to randomize the channel combining operation, to a subset  of nine Clifford gates only. To do so, we need to find a subset of Clifford gates such that the condition $(t.2)$ from Lemma~\ref{lemma:stochastic_proc_convergence} is still fulfilled. 

\medskip Let $\mathcal{C}_n$ be the Clifford group on $n$ qubits. Clearly $\mathcal{C}_1 \otimes \mathcal{C}_1$ is a subgroup of $\mathcal{C}_2$, and we may define an equivalence relation on $\mathcal{C}_2$, whose equivalence classes are the left cosets of $\mathcal{C}_1 \otimes \mathcal{C}_1$. 
\begin{definition}
 We say that $C'$ and $C'' \in \mathcal{C}_2$ are equivalent, and denote it by $C'\sim C''$, if there exist $C_1, C_2  \in \mathcal{C}_1$ such that $C'' = C'(C_1 \otimes C_2)$ (see also Fig.~\ref{fig:equiv_cliffords}).
\end{definition}

\begin{figure}[!h]
\centering
\begin{tikzpicture}
\draw (0,0) node[left] {$U_2$} -- +(0.75,0);
\draw (0,1) node[left] {$U_1$} -- +(0.75,0);
\draw (0.75,-0.25) rectangle (1.25,0.25);
\draw (0.75,0.75) rectangle (1.25,1.25);
\draw (1,0) node {$C_2$};
\draw (1,1) node {$C_1$};
\draw (1.25,0) -- +(0.5,0);
\draw (1.25,1) -- +(0.5,0);
\draw (1.75,-0.25) rectangle (2.25,1.25);
\draw (2,0.5) node {$C'$};
\draw (2.25,0) -- +(0.75,0) node[right] {$X_2$};
\draw (2.25,1) -- +(0.75,0) node[right] {$X_1$};
\draw[dashed] (0.5,-0.5) rectangle (2.5,1.5);
\draw (1.5,-0.5) node[below] {$C''$};
\end{tikzpicture}
\caption{Equivalent two-qubit Clifford gates $C' \sim C''$}
\label{fig:equiv_cliffords}
\end{figure}
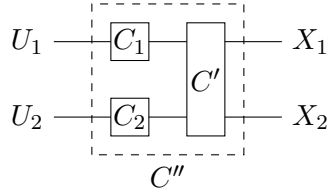

Now, the main observation is that two equivalent Clifford gates used to combine any two quantum channels with qubit inputs, yield the same Rényi-Bhattacharyya parameter of the bad/good channels.
This is stated in the following lemma, whose proof is provided in Appendix~\ref{sec:proof-quantum-red-clifford}.

\begin{lemma}\label{lemma:equiv_conditions}
Let $C', C'' \in \mathcal{C}_2$.  If $C'\sim C''$, then for any two quantum channels $\mathcal{M}$ and $\mathcal{N}$ with qubit inputs, we have:
\begin{equation*}
R(\mathcal{M} \boxast_{C'} \mathcal{N}) = R( \mathcal{M} \boxast_{C''} \mathcal{N}) 
\ \ \text{ and } \ \ 
R(\mathcal{M} \varoast_{C'}\mathcal{N}) =  R( \mathcal{M} \varoast_{C''} \mathcal{N}).
\end{equation*}
\end{lemma}

As a consequence, one may ensure polarization while restricting the set of Clifford gates to any set of representatives of the equivalence classes of the above equivalence relation (since such a restriction will not affect the $\mbE_C R(\mathcal{N} \varoast_C \mathcal{M})$ value, for any two quantum channels $\mathcal{M}$ and $\mathcal{N}$ with qubit inputs). Since $|\mathcal{C}_1| = 24$ and $|\mathcal{C}_2| = 11520$, it follows that there are exactly $11520 / (24\times 24) = 20$ equivalence classes.  A set of representatives of these $20$ equivalence classes can be chosen as follows\footnote{We used a computer program to determine such a set of representatives.}:

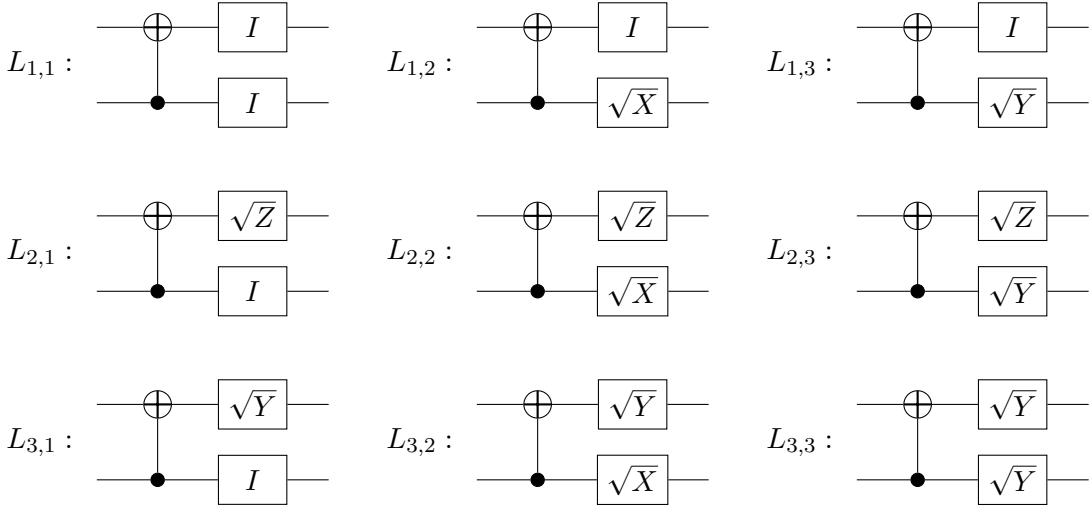
\begin{figure}[!t]
    \centering
    \begin{tikzpicture}
    \def\nodewidth{9mm}     
    \def\nodeheight{6.5mm}  
    \begin{scope}
    \draw
    (0, 0) node[draw, minimum width=\nodewidth, minimum height=\nodeheight](C1){$I$}
    (C1) ++(0, -1) node[draw, minimum width=\nodewidth, minimum height=\nodeheight](C2){$I$}
    (C1) ++(-1.25, 0) node[not](C3){}
    (C2) ++(-1.25, 0) node[phase](C4){}
    ;
    \draw
    (C3) to  (C1)
    (C4) to  (C2)
    (C3) to  (C4)
    (C3) to  ++(-0.8, 0)
    (C4) to  ++(-0.8, 0)
    (C2) to  ++(+1, 0)
    (C1) to  ++(+1, 0)
    (C1) ++(-2.8, -.5) node[](){$L_{1,1}:$}
    ;
    \draw
    (0, -2.5) node[draw, minimum width=\nodewidth, minimum height=\nodeheight](C1){$\sqrt{Z}$}
    (C1) ++(0, -1) node[draw, minimum width=\nodewidth, minimum height=\nodeheight](C2){$I$}
    (C1) ++(-1.25, 0) node[not](C3){}
    (C2) ++(-1.25, 0) node[phase](C4){}
    ;
    \draw
    (C3) to  (C1)
    (C4) to  (C2)
    (C3) to  (C4)
    (C3) to  ++(-0.8, 0)
    (C4) to  ++(-0.8, 0)
    (C2) to ++(+1, 0)
    (C1) to ++(+1, 0)
    (C1) ++(-2.8, -.5) node[](){$L_{2,1}:$}
    ;
    \draw
    (0, -5.0) node[draw, minimum width=\nodewidth, minimum height=\nodeheight](C1){$\sqrt{Y}$}
    (C1) ++(0, -1) node[draw, minimum width=\nodewidth, minimum height=\nodeheight](C2){$I$}
    (C1) ++(-1.25, 0) node[not](C3){}
    (C2) ++(-1.25, 0) node[phase](C4){}
    ;
    \draw
    (C3) to  (C1)
    (C4) to  (C2)
    (C3) to  (C4)
    (C3) to  ++(-0.8, 0)
    (C4) to  ++(-0.8, 0)
    (C2) to ++(+1, 0)
    (C1) to ++(+1, 0)
    (C1) ++(-2.8, -0.5) node[](){$L_{3,1}:$}
    ;
    \end{scope}
    \begin{scope}[xshift =5cm]
    \draw
    (0, 0) node[draw, minimum width=\nodewidth, minimum height=\nodeheight](C1){$I$}
    (C1) ++(0, -1) node[draw, minimum width=\nodewidth, minimum height=\nodeheight](C2){$\sqrt{X}$}
    (C1) ++(-1.25, 0) node[not](C3){}
    (C2) ++(-1.25, 0) node[phase](C4){}
    ;
    \draw
    (C3) to  (C1)
    (C4) to  (C2)
    (C3) to  (C4)
    (C3) to  ++(-0.8, 0)
    (C4) to  ++(-0.8, 0)
    (C2) to ++(+1, 0)
    (C1) to ++(+1, 0)
    (C1) ++(-2.8, -0.5) node[](){$L_{1,2}:$}
    ;
    \draw
    (0, -2.5) node[draw, minimum width=\nodewidth, minimum height=\nodeheight](C1){$\sqrt{Z}$}
    (C1) ++(0, -1) node[draw, minimum width=\nodewidth, minimum height=\nodeheight](C2){$\sqrt{X}$}
    (C1) ++(-1.25, 0) node[not](C3){}
    (C2) ++(-1.25, 0) node[phase](C4){}
    ;
    \draw
    (C3) to  (C1)
    (C4) to  (C2)
    (C3) to  (C4)
    (C3) to  ++(-0.8, 0)
    (C4) to  ++(-0.8, 0)
    (C2) to ++(+1, 0)
    (C1) to ++(+1, 0)
    (C1) ++(-2.8, -0.5) node[](){$L_{2,2}:$}
    ;
     \draw
    (0, -5) node[draw, minimum width=\nodewidth, minimum height=\nodeheight](C1){$\sqrt{Y}$}
    (C1) ++(0, -1) node[draw, minimum width=\nodewidth, minimum height=\nodeheight](C2){$\sqrt{X}$}
    (C1) ++(-1.25, 0) node[not](C3){}
    (C2) ++(-1.25, 0) node[phase](C4){}
    ;
    \draw
    (C3) to  (C1)
    (C4) to  (C2)
    (C3) to  (C4)
    (C3) to  ++(-0.8, 0)
    (C4) to  ++(-0.8, 0)
    (C2) to ++(+1, 0)
    (C1) to ++(+1, 0)
    (C1) ++(-2.8, -.5) node[](){$L_{3,2}:$}
    ;
    \end{scope}
    \begin{scope}[xshift=10cm]

     \draw
    (0, 0) node[draw, minimum width=\nodewidth, minimum height=\nodeheight](C1){$I$}
    (C1) ++(0, -1) node[draw, minimum width=\nodewidth, minimum height=\nodeheight](C2){$\sqrt{Y}$}
    (C1) ++(-1.25, 0) node[not](C3){}
    (C2) ++(-1.25, 0) node[phase](C4){}
    ;
    \draw
    (C3) to  (C1)
    (C4) to  (C2)
    (C3) to  (C4)
    (C3) to  ++(-0.8, 0)
    (C4) to  ++(-0.8, 0)
    (C2) to ++(+1, 0)
    (C1) to ++(+1, 0)
    (C1) ++(-2.8, -0.5) node[](){$L_{1,3}:$}
    ;
     \draw
    (0, -2.5) node[draw, minimum width=\nodewidth, minimum height=\nodeheight](C1){$\sqrt{Z}$}
    (C1) ++(0, -1) node[draw, minimum width=\nodewidth, minimum height=\nodeheight](C2){$\sqrt{Y}$}
    (C1) ++(-1.25, 0) node[not](C3){}
    (C2) ++(-1.25, 0) node[phase](C4){}
    ;
    \draw
    (C3) to  (C1)
    (C4) to  (C2)
    (C3) to  (C4)
    (C3) to  ++(-0.8, 0)
    (C4) to  ++(-0.8, 0)
    (C2) to ++(+1, 0)
    (C1) to ++(+1, 0)
    (C1) ++(-2.8, -.5) node[](){$L_{2,3}:$}
    ;
    \draw
    (0, -5) node[draw, minimum width=\nodewidth, minimum height=\nodeheight](C1){$\sqrt{Y}$}
    (C1) ++(0, -1) node[draw, minimum width=\nodewidth, minimum height=\nodeheight](C2){$\sqrt{Y}$}
    (C1) ++(-1.25, 0) node[not](C3){}
    (C2) ++(-1.25, 0) node[phase](C4){}
    ;
    \draw
    (C3) to  (C1)
    (C4) to  (C2)
    (C3) to  (C4)
    (C3) to  ++(-0.8, 0)
    (C4) to  ++(-0.8, 0)
    (C2) to ++(+1, 0)
    (C1) to ++(+1, 0)
    (C1) ++(-2.8, -.5) node[](){$L_{3,3}:$}
    ;
    \end{scope}
    \end{tikzpicture}
    \caption{The set $\mathcal{L} := \{L_{i,j} \mid 1 \leq i,j \leq 3\}$ containing nine Cliffords.}
    \label{fig:cliff_generators}
\end{figure}

\begin{itemize}[leftmargin=1.25\parindent]
\item For two of these equivalence classes, one may choose the identity gate $I$ and swap gate $S$, as representatives.

 \item For nine out of the remaining 18 equivalence classes, one may find representatives of the form $(C_1 \otimes C_2) \text{\sc cnot}_{21}$, where $\text{\sc cnot}_{21}$ denotes the controlled-NOT gate with  control on the second qubit and  target on the first qubit,  $C_1 \in \{I, \sqrt{Z}, \sqrt{Y}\}$,  $C_2 \in \{I, \sqrt{X}, \sqrt{Y}\}$, and $\sqrt{P} = \frac{(1-i)(\ident + i P)}{2}$, for any Pauli matrix $P \in \{X, Y, Z\}$. 
 We denote this set by $\mathcal{L}$,  which is further depicted in Fig.~\ref{fig:cliff_generators}.
  \begin{equation*}
\mathcal{L} := \left\{ (C_1 \otimes C_2) \text{\sc cnot}_{21} \mid C_1 \in \{I, \sqrt{Z}, \sqrt{Y}\},   C_2 \in \{I, \sqrt{X}, \sqrt{Y}\} \right\}.
 \end{equation*}
\end{itemize}
\begin{itemize}[leftmargin=1.25\parindent]
\item For the remaining nine equivalence classes, one may find representatives of the form $S L$, where $S$ is the swap gate and $L\in \mathcal{L}$. We denote this set by $ \mathcal{R}$,
\begin{equation*}
 \mathcal{R} := \left\{ SL \mid L\in \mathcal{L} \right\}.
 \end{equation*}
\end{itemize}

Now, we prove that two Clifford gates $C'$ and $C''$, such that $C'' = SC'$, used to combine two copies of a quantum channel $\mathcal{W}$ with qubit input, yield the same Rényi-Bhattacharyya parameter of the bad/good channels. Although this property is weaker than the one in Lemma~\ref{lemma:equiv_conditions}, which holds for any two quantum channels $\mathcal{M}$ and $\mathcal{N}$, it is sufficient for whatever we need here.

\begin{lemma} \label{lemma:swap_eq}
Let $C', C'' \in \mathcal{C}_2$, such that $C'' = SC'$, where $S$ is the swap gate. Then, for two copies of a quantum channel $\mathcal{W}$ with qubit input,
\begin{equation*}
R(\mathcal{W} \boxast_{C'} \mathcal{W}) = R( \mathcal{W} \boxast_{C''} \mathcal{W})
\ \ \text{ and } \ \ 
R(\mathcal{W} \varoast_{C'}\mathcal{W}) =  R( \mathcal{W} \varoast_{C''} \mathcal{W}).
\end{equation*}
\end{lemma}

\begin{proof}
First, we note that by applying a unitary on the output of any quantum channel does not change the Rényi-Bhattacharyya parameter. Precisely, let $\mathcal{N}_{A \rightarrow B}$ be any quantum channel, and $U\mathcal{N}_{A \rightarrow B} U^\dagger$ be the quantum channel\footnote{To see that $U\mathcal{N}_{A \rightarrow B} U^\dagger$ is a quantum channel, it is enough to notice that if $\mathcal{N}_{A \rightarrow B}$ is defined by Kraus operators $\{E_k\}$, then $U\mathcal{N}_{A \rightarrow B} U^\dagger$ is defined by Kraus operators $\{ U E_k U^\dagger\}$.} obtained by applying the unitary $U$ on the output system $B$, that is, $\left(U\mathcal{N}_{A \rightarrow B} U^\dagger\right)(\rho_A) := U\mathcal{N}_{A \rightarrow B}(\rho_A) U^\dagger$. Then,
\begin{equation}\label{eq:u_same_r}
R\left( U\mathcal{N}_{A \rightarrow B} U^\dagger \right) = R\left( \mathcal{N}_{A \rightarrow B}  \right).
\end{equation}

Going back to the proof of our Lemma, by the definition of $ \mathcal{W} \boxast_{C} \mathcal{W}$  and using that $S^\dagger = S$, we may write:
\begin{align*}
(\mathcal{W} \boxast_{C''} \mathcal{W}) (\rho) &= (\mathcal{W} \otimes \mathcal{W}) \left(C''(\rho \otimes \frac{\ident}{2})C''^\dagger\right) \\
     &= (\mathcal{W} \otimes \mathcal{W}) \left(SC'(\rho \otimes \frac{\ident}{2})C'^\dagger S\right). 
\end{align*}
Now, it is easily seen that the $\mathcal{W} \otimes \mathcal{W}$ channel is covariant with respect to the swap gate, {\em i.e.}, the swap gate commutes with the action of the channel. 
Hence we may further write:
\begin{align*}
(\mathcal{W} \boxast_{C''} \mathcal{W}) (\rho) &= S(\mathcal{W} \otimes \mathcal{W}) \left(C'(\rho \otimes \frac{\ident}{2})C'^\dagger \right) S \\
   &= S (\mathcal{W} \boxast_{C'} \mathcal{W}) (\rho) S \\
   &= \left(S (\mathcal{W} \boxast_{C'} \mathcal{W}) S\right) (\rho). 
\end{align*}
Hence, $\mathcal{W} \boxast_{C''} \mathcal{W} = S (\mathcal{W} \boxast_{C'} \mathcal{W}) S$, and using (\ref{eq:u_same_r}), with $\mathcal{N}:= \mathcal{W} \boxast_{C'} \mathcal{W}$ and $U:=S$, we get
\begin{equation*}
R\left( \mathcal{W} \boxast_{C'} \mathcal{W} \right) = R\left( \mathcal{W} \boxast_{C''} \mathcal{W} \right),
\end{equation*}
as desired. The equality $R(\mathcal{W} \varoast_{C'} \mathcal{W}) = R(\mathcal{W} \varoast_{C''} \mathcal{W})$ may be proven in a similar way. 
\end{proof}

The following lemma implies that polarization can be achieved by choosing the channel combining Clifford operator randomly from either $\mathcal{L}$ or $\mathcal{R}$. It is the analogue of the Lemma~\ref{lem:d-good-channel} used to check the $(t.2)$ condition in the proof of the polarization Theorem~\ref{thm:quantum_polarization}. 

\begin{lemma} \label{lemma:quantum_pol9}
Given two copies of a quantum channel $\mathcal{W}_{A'_1 \rightarrow B_1}$ with qubit input, we have
\begin{equation*}
\mbE_{C \in \mathcal{L}} R(\mathcal{W} \varoast_C \mathcal{W}) =  \mbE_{C \in \mathcal{R}} R(\mathcal{W} \varoast_C \mathcal{W})
= \frac{4}{9} - \frac{1}{9}R(\mathcal{W}) + \frac{4}{9}R(\mathcal{W})^2,
\end{equation*}
where $C$ is the channel combining Clifford operator and is chosen uniformly either from the set $\mathcal{L}$ or from the set $\mathcal{R}$,  each containing nine Clifford gates. 
\end{lemma}

\begin{proof} 
Since $\mathcal{S}:=\{I, S\} \cup \mathcal{L} \cup \mathcal{R}$ is a set of representatives of the $20$ equivalence classes partitioning the Clifford group $\mathcal{C}_2$, we have:
\begin{equation*}
\mbE_{C \in \mathcal{S}} R(\mathcal{W} \varoast_C \mathcal{W}) = \mbE_{C \in \mathcal{C}_2} R(\mathcal{W} \varoast_C \mathcal{W}) = \frac{2}{5} + \frac{2}{5} R(\mathcal{W})^2,
\end{equation*}
where the first equality follows from Lemma~\ref{lemma:equiv_conditions}, and the second from Lemma~\ref{lem:d-good-channel}.
Now, using Lemma~\ref{lemma:swap_eq}, we have $R(\mathcal{W} \varoast_S \mathcal{W}) = R(\mathcal{W} \varoast_I \mathcal{W}) = R(\mathcal{W})$ and $\mbE_{C \in \mathcal{L}} R(\mathcal{W} \varoast_C \mathcal{W}) = \mbE_{C \in \mathcal{R}} R(\mathcal{W} \varoast_C \mathcal{W})$. Hence,
\begin{equation*}
\mbE_{C \in \mathcal{S}} R(\mathcal{W} \varoast_C \mathcal{W}) = \frac{2R(\mathcal{W}) + 9\mbE_{C \in \mathcal{L}} R(\mathcal{W} \varoast_C \mathcal{W})  + 9\mbE_{C \in \mathcal{R}} R(\mathcal{W} \varoast_C \mathcal{W})}{20},
\end{equation*}
and therefore
\begin{equation*}
\mbE_{C \in \mathcal{L}} R(\mathcal{W} \varoast_C \mathcal{W}) = \mbE_{C \in \mathcal{R}} R(\mathcal{W} \varoast_C \mathcal{W}) 
 = \frac{4}{9} - \frac{1}{9}R(\mathcal{W}) + \frac{4}{9}R(\mathcal{W})^2.
\end{equation*}
Finally, we also note that the above expected value is less than the one in Lemma~\ref{lem:d-good-channel}, namely $\mbE_{C \in \mathcal{C}_2} R(\mathcal{W} \varoast_C \mathcal{W}) = \frac{2}{5} + \frac{2}{5} R(\mathcal{W})^2$, since the expected value can only decrease by taking out the identity and swap gate from the set of representatives.
\end{proof}

\section{Quantum Polar Coding}
\label{sec:quantum_polar_coding}

\subsection{Quantum Polar Codes}

\begin{figure}[!thb]
\centering
\includegraphics[width=.6\linewidth]{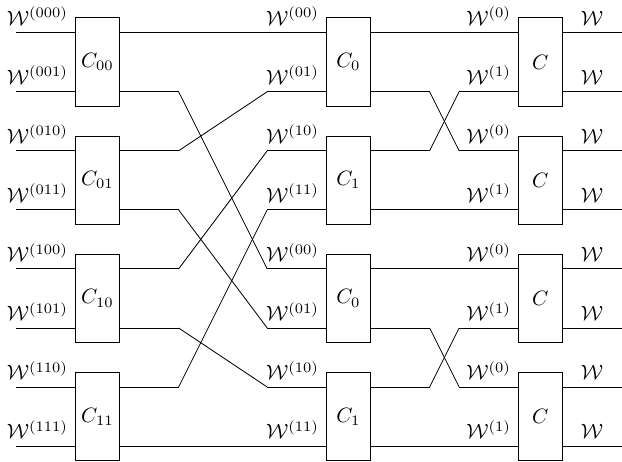}
\caption{Quantum polar code of length $N=8$}
\label{fig:polar_code_N8}
\end{figure}

Polar coding is a coding method that takes advantage of the channel polarization phenomenon \cite{arikan09}. To construct a quantum polar code of length $N=2^n$, $n>0$, we start with $N$ copies of the quantum channel $\mathcal{W}$, pair them in $N/2$ pairs, and apply the channel combining and splitting operation on each pair. The same channel combining Clifford gate is used for each of the $N/2$ pairs, which will be denoted by $C$. By doing so, we generate $N/2$ copies of the channel $\mathcal{W}^{(0)} :=  \mathcal{W} \boxast_{C} \mathcal{W}$ and $N/2$ copies of the channel $\mathcal{W}^{(1)} :=  \mathcal{W} \varoast_{C} \mathcal{W}$. Hence, for each $i_1=0,1$, we group together the $N/2$ copies of the $\mathcal{W}^{(i_1)}$ channel, pair them in $N/4$ pairs, and apply the channel combining and splitting operation on each pair, by using some channel combining Clifford gate denoted by $C_{i_1}$.
By performing $n$ {\em polarization steps}  (that is, applying the above construction recursively $n$ times), we generate quantum channels $\mathcal{W}^{(i_1\cdots i_n)}$, which can be recursively defined for $n > 0$, as follows:
\begin{equation}\label{eq:polar_coding_recursion}
\mathcal{W}^{(i_1\cdots i_n)} := \left\{ \begin{array}{@{}ll@{}}
   \mathcal{W}^{(i_1\cdots i_{n-1})} \boxast_{C_{i_1\cdots i_{n-1}}}  \mathcal{W}^{(i_1\cdots i_{n-1})}, & \text{if } i_n = 0 \\
   \mathcal{W}^{(i_1\cdots i_{n-1})} \varoast_{C_{i_1\cdots i_{n-1}}} \mathcal{W}^{(i_1\cdots i_{n-1})}, & \text{if } i_n = 1
\end{array} \right.
\end{equation}
where, for $n=1$, in the right hand side term of the above equality, we set by convention $\mathcal{W}^{(\varnothing)} := \mathcal{W}$ and $C_\varnothing := C$.
Note that, for the sake of simplicity, we have dropped the channel combining Clifford gate from the $\mathcal{W}^{(i_1\cdots i_n)}$ notation. The construction is illustrated in Fig.~\ref{fig:polar_code_N8}, for $N=8$. Horizontal ``wires'' represent qubits, and for each polarization step, we have indicated on each wire the virtual channel $\mathcal{W}^{(i_1i_2\cdots)}$ ``seen'' by the corresponding qubit state.

\medskip The above construction synthesizes a set of $N$  channels and, for any $i=0,\dots,N-1$, we shall further denote $\mathcal{W}^{(i)} := \mathcal{W}^{(i_1\cdots i_n)}$, where $i_1\cdots i_n$ is the binary decomposition of $i$. Let $\mathcal{I} \subseteq \{0,1,\dots,N-1 \}$ denote the set of good channels ({\em i.e.}, with coherent information close to $1$, or equivalently, Rényi-Bhattacharyya parameter close to  $1/2$), and let $\mathcal{J} := \{0,1,\dots,N-1 \} \setminus \mathcal{I}$. With a slight abuse of notation, we shall also denote by ${\cal I}$ and ${\cal J}$ two qudit systems, of dimension $2^{|{\cal I}|}$ and $2^{|{\cal J}|}$, respectively (it will be clear from the context whether the notation is meant to indicate a set of indices or a quantum system).

\smallskip A quantum state $\rho_{\cal I}$ on system ${\cal I}$ is encoded by supplying it as input to channels $i\in\mathcal{I}$, while supplying each channel $j\in\mathcal{J}$ with half of an EPR pair, shared between the sender and the receiver. Precisely, let $\Phi_{\cal JJ'}$ be a maximally entangled state, defined by 
\begin{equation}\label{eq:max_entangled_state}
\Phi_{\cal JJ'} = \otimes_{j\in \mathcal{J}} \Phi_{jj'},
\end{equation}
where indices $j$ and $j'$ indicate the $j$-th qubits of $\mathcal{J}$ and $\mathcal{J}'$ systems, respectively, and $\Phi_{jj'}$ is an EPR pair. Let also $G_q$ denote the quantum polar transform, that is the unitary operator defined by applying Clifford gates corresponding to the $n$ polarization steps. The encoded state, denoted $\varphi_{{\cal IJJ}'}$, is obtained by applying the $G_q\otimes I_{{\cal J}'}$ unitary on the $\mathcal{I}\mathcal{J}\mathcal{J}'$ system, hence:
\begin{equation}\label{eq:encoded_state}
\varphi_{{\cal IJJ}'} \eqdef (G_q\otimes I_{{\cal J}'}) (\rho_{\cal I} \otimes \Phi_{\cal JJ'}) (G_q^\dagger \otimes I_{{\cal J}'}).
\end{equation}
Since no errors occur on the $\mathcal{J}'$ system, the channel output state is given by:
\begin{equation}
\psi_{{\cal IJJ}'} \eqdef ({\cal W}^{\otimes N} \otimes I_{{\cal J}'}) (\varphi_{{\cal IJJ}'}).
\end{equation}

It is worth noticing that randomness is used only at the code construction stage (since Clifford gates used in the $n$ polarization steps are randomly chosen from some predetermined set of gates), but not at the encoding stage.
The constructed polar code allows communicating quantum information over a quantum channel ${\cal W}$ at a rate $|{\cal I}|/N$, which approaches $(1 + I({\cal W}))/2$ (that is, half the symmetric mutual information of the channel), as $N$ goes to infinity. The net communication rate, which we define as the quantum communication rate less the entanglement consumption rate, is given by $(|{\cal I}| - |{\cal J}|)/N$, and approaches $I(\mathcal{W})$, as $N$ goes to infinity.

\subsection{Quantum Polar Codes as Entanglement-Assisted Stabilizer Codes} 
\label{subsec:EASC}

Including all information qubits (${\cal I}$ system) and both systems of the preshared EPR pairs (${\cal J}$ and ${\cal J}'$ systems), the quantum polar code from the above section can be described as an entanglement assisted stabilizer code, in the sense of~\cite{hsieh-devetak-brun2}. Precisely, using the notation from the previous section, the quantum state $\rho_{\cal I} \otimes \Phi_{\cal JJ'}$ is stabilized by the set of Pauli operators 
\begin{equation}
\label{eq:s_i_j_jp}
{\cal S}_{{\cal IJJ}'} := \{ I_{\cal I}\otimes X_j X_{j'},\ I_{\cal I}\otimes Z_j Z_{j'} \mid j \in {\cal J}\},
\end{equation}
where $I_{\cal I}$ denotes the identity on the ${\cal I}$ system, and $X_j X_{j'}$ (respectively, $Z_j Z_{j'}$) denotes the tensor product of the Pauli-$X$ (respectively, Pauli-$Z$) operators of the $j$-th qubits of systems ${\cal J}$ and ${\cal J}'$\,\footnote{Note that the definition of ${\cal S}_{{\cal IJJ}'}$ in (\ref{eq:s_i_j_jp}) depends only on index $j\in{\cal J}$, since $j'\in{\cal J}'$ is the counterpart of $j$ (thus, uniquely determined by the latter).}. Conversely, any quantum state on the tripartite ${\cal IJJ}'$ system, which is stabilized by Pauli operators in ${\cal S}_{{\cal IJJ}'}$, is necessarily of the form $\rho_{\cal I} \otimes \Phi_{\cal JJ'}$ (by a dimension argument). Hence, encoded states ($\varphi_{{\cal IJJ}'}$ defined in (\ref{eq:encoded_state})) are stabilized by the set of Pauli operators obtained by {\em passing the elements of ${\cal S}_{{\cal IJJ}'}$ through the polar transform $G_q$}, that is,
\begin{equation}
\label{eq:stabilizer_set}
\bar{\cal S}_{{\cal IJJ}'} := (G_q \otimes I_{{\cal J}'}) {\cal S}_{{\cal IJJ}'} (G_q^\dagger \otimes I_{{\cal J}'}).
\end{equation} 

For stabilizer codes, the decoding problem for general quantum channels reduces to decoding Pauli errors only, after performing syndrome measurement, {\em i.e.}, measuring all the generators of the stabilizer group (in our case, the elements of  $\bar{\cal S}_{{\cal IJJ}'}$). Here, the implicit assumption is that syndrome measurement induces appropriate projections, such that it results in the standard Pauli error model.

\medskip For Pauli channels, we provide an efficient decoding algorithm (Section~\ref{sec:decod_pauli_channel}), achieving the symmetric coherent information of the channel. For general quantum channels, syndrome measurement coupled with the above decoding on the induced Pauli error model may yield a practical solution to the decoding problem. However, such a solution is not optimal, due to the loss of information incurred during syndrome measurement. Besides, the polar code should be fitted to (and thus exploit the polarization of) the induced Pauli error model, rather than the quantum channel itself. Devising an efficient  decoding algorithm capable of achieving the symmetric coherent information of general quantum channels is an open problem.

\section{Polarization of Pauli Channels}
\label{sec:polarization_of_pauli_channels}
This section further investigates the quantum polarization of Pauli channels. First, to a Pauli channel ${\cal N}$ we associate a classical symmetric channel ${\cal N}^\#$, with both input and output alphabets given by the quotient of the $1$-qubit Pauli group by its centralizer. We then show that the former polarizes quantumly if and only if the latter polarizes classically. We use this equivalence to provide an alternative proof of the quantum polarization for a Pauli channel, as well as fast polarization properties. We then devise an effective way to decode a quantum polar code on a Pauli channel, by decoding its classical counterpart.

\smallskip Let $P_n$ denote the Pauli group on $n$ qubits, and $\bar{P}_n = P_n/\{\pm 1,  \pm i\}$ the Abelian group obtained by taking the quotient of $P_n$ by its centralizer. We write $\bar{P}_1 = \{\sigma_i \mid i=0,\dots,3\}$, with $\sigma_0 = I$, $\sigma_1 = X$, $\sigma_2 = Y$, $\sigma_3 = Z$, and $\bar{P}_2 = \{\sigma_{i,j}\eqdef \sigma_i \otimes \sigma_j \mid i,j=0,\dots,3\} \simeq \bar{P}_1\times\bar{P}_1$.
For any two-qubit Clifford unitary $C$, we denote by $\Gamma(C)$, or simply $\Gamma$ when no confusion is possible, the conjugate action of $C$ on $\bar{P}_2$. Hence, $\Gamma$ is the automorphism of $\bar{P}_2$ (or equivalently $\bar{P}_1\times\bar{P}_1$), defined by $\Gamma(\sigma_{i,j}) = C \sigma_{i,j} C^\dagger$.

\smallskip Let ${\cal N}$ be a Pauli channel defined by\footnote{We use $\sigma_i^\dagger$ in the definition of the Pauli channel, to explicitly indicate that the definition does not depend on the representative of the equivalence class.} ${\cal N}(\rho) = \sum_{i=0}^3 p_i \sigma_i \rho \sigma_i^{\dagger}$, with $\sum_{i=0}^3 p_i = 1$. Its coherent information for a uniformly distributed input is given by $I({\cal N}) = 1 - h(\mathbf{p})$, where $h(\mathbf{p}) = -\sum_{i=0}^3 p_i\log(p_i)$ denotes the entropy of the probability vector $\mathbf{p} = (p_0,p_1,p_2,p_3)$. 

\begin{definition}[Classical counterpart of a Pauli channel] \label{def:classical_channel} Let ${\cal N}$ be a Pauli channel. The classical counterpart of ${\cal N}$, denoted by ${\cal N}^\#$, is the classical channel with input and output alphabet $\bar{P}_1$, and transition probabilities ${\cal N}^\#(\sigma_i \mid \sigma_j) = p_k$, where $k$ is such that $\sigma_i \sigma_j = \sigma_k$~\footnote{Here, equality is understood as equivalence classes in $\bar{P}_1$.}.
\end{definition}

\noindent Hence, ${\cal N}^\#$ is a memoryless symmetric channel, whose capacity is given by the mutual information for uniformly distributed input $\mathtt{I}({\cal N}^\#) = \frac{1}{2}(2 - h(\mathbf{p})) \in [0,1]$. It follows that
\begin{equation*}
\mathtt{I}({\cal N}^\#) = \displaystyle \frac{1 + I({\cal N})}{2}.
\end{equation*}
Note that the right hand side term in the above equation is half the mutual information of the Pauli channel ${\cal N}$, for a uniformly distributed input.

\smallskip It is worth noticing that the quantum channels synthesized during the quantum polarization of a Pauli channel are {\em identifiable} (see below) to classical mixtures of Pauli channels (this will be proved in Proposition~\ref{prop:cq_equiv}). A Classical Mixture of Pauli (CMP) channels is a quantum channel ${\cal N}(\rho) = \sum_{x\in X} \lambda_x \ket{x}\bra{x} \otimes {\cal N}_x(\rho)$, where $\{\ket{x}\}_{x\in X}$ is some orthonormal basis of an auxiliary system,  ${\cal N}_x$ are Pauli channels, and $\sum_{x\in X} \lambda_x = 1$. We further extend Definition~\ref{def:classical_channel} to the case of CMP channels, by defining the classical channel ${\cal N}^\#$ as the mixture of the channels ${\cal N}_x^\#$, where channel ${\cal N}_x^\#$ is used with probability $\lambda_x$. Hence, input and output alphabets of ${\cal N}^\#$  are $\bar{P}_1$ and $X \times \bar{P}_1$, respectively, with channel transition probabilities defined by ${\cal N}^\#(x, \sigma_i \mid \sigma_j) = \lambda_x\,{\cal N}_x(\sigma_i \mid \sigma_j)$.
It also follows that:
\begin{equation*}
\displaystyle \mathtt{I}({\cal N}^\#) = \sum_x \lambda_x \mathtt{I}({\cal N}_x^\#) =\sum_x \lambda_x \frac{1 + I({\cal N}_x)}{2} = \frac{1 + I({\cal N})}{2}.
\end{equation*}

Given two classical channels ${\cal U}$ and ${\cal V}$, we say they are equivalent, and denote it by ${\cal U} \equiv {\cal V}$, if they are defined by the same transition probability matrix, modulo a permutation of rows and columns. The following lemma states that the classical channel associated with a CMP channel does not depend on the basis. 

\begin{lemma}\label{lem:nclassical-indep-of-basis}
Let ${\cal N}(\rho) = \sum_{x\in X} \lambda_x \ket{x}\bra{x} \otimes {\cal N}_x(\rho)$ and ${\cal M}(\rho) = \sum_{y\in Y} \tau_y \ket{y}\bra{y} \otimes {\cal M}_y(\rho)$ be two CMP channels, where $\{\ket{x}\}_{x\in X}$ and $\{\ket{y}\}_{y\in Y}$ are orthonormal bases of the same auxiliary system. If ${\cal N} = {\cal M}$, then there exists a bijective mapping $\pi:X\rightarrow Y$, such that $\lambda_x = \tau_{\pi(x)}$ and ${\cal N}_x = {\cal M}_{\pi(x)}$. In particular, ${\cal N}^\# \equiv {\cal M}^\#$.
\end{lemma}

Finally, we say that a quantum channel ${\cal N}_{U\rightarrow AX}$ is {\em identifiable} to a channel ${\cal N}'_{U\rightarrow A}$ if, for some unitary operator $C$ on the $AX$ system, we have that ${\cal N}(\rho) = C\left({\cal N}'(\rho) \otimes \frac{I_X}{|X|}\right)C^\dagger$, where $|X|$ denotes the dimension of the $X$ system.  In other words, ${\cal N}$ and ${\cal N}'$ are equal modulo the conjugate action of an unitary operator $C$, and possibly after discarding a ``useless'' output system $X$. If ${\cal N}_{U\rightarrow AX}$ is identifiable to a CMP channel  ${\cal N}'_{U\rightarrow A}$, we shall define ${\cal N}^\# \eqdef \left({\cal N}'\right)^\#$. 
It can be seen that ${\cal N}^\#$ is well defined up to equivalence of classical channels, that is, if  ${\cal N}_{U\rightarrow AX}$ is identifiable to another CMP channel  ${\cal N}''_{U\rightarrow A}$, then $\left({\cal N}'\right)^\# \equiv \left({\cal N}''\right)^\#$. This follows from the following lemma, proven in Appendix~\ref{sec:proof-nclassical-well-defined}.

\begin{lemma}\label{lem:nclassical-well-defined}
Let ${\cal N}'$ and ${\cal N}''$ be two CMP channels, such that  ${\cal N}'(\rho) \otimes \frac{I_X}{|X|} = C\left({\cal N}''(\rho) \otimes \frac{I_X}{|X|}\right)C^\dagger$, for some unitary $C$. Then $\left({\cal N}'\right)^\# \equiv \left({\cal N}''\right)^\#$.
\end{lemma}

\subsection{Classical Channel Combining and Splitting Operations}
\label{subsec:classical_combining_splitting}

{\bf Simplified notation:} To simplify notation, we shall identify $\left(\bar{P}_1, \times\right) \cong \left(\{0,1,2,3\}, \oplus\right)$, by identifying $\sigma_u \cong u$, $\forall u=0,\dots,3$, where the additive group operation $u\oplus v$ is given by the bitwise exclusive {\sc or} ({\sc xor}) between the binary representations of integers $u,v$. The classical counterpart ${\cal N}^\#$ of a Pauli channel ${\cal N}(\rho) = \sum_{u=0}^3 p_u \sigma_u \rho \sigma_u^{\dagger}$ (Definition~\ref{def:classical_channel}), is therefore identified  to a channel with input and output alphabet $\bar{P}_1 \cong \{0,1,2,3\}$, and transition probabilities ${\cal N}^\#(u \mid v) = p_{u\oplus v}$.

\medskip Let $N$ and $M$ be two classical channels, both with input alphabet  $\bar{P}_1 \cong \{0,1,2,3\}$, and output alphabets $A$ and $B$, respectively. Channel transition probabilities are denoted by $N(a\mid u)$ and $M(b \mid v)$, for $u,v\in \bar{P}_1$, $a\in A$, and $b \in B$. Let $\Gamma : \bar{P}_1\times\bar{P}_1 \rightarrow \bar{P}_1\times\bar{P}_1$ be any permutation, and  write $\Gamma = (\Gamma_1, \Gamma_2)$, with $\Gamma_i: \bar{P}_1\times\bar{P}_1 \rightarrow \bar{P}_1$, $i=1,2$. The {\em combined channel} $N\bowtie_\Gamma M$ is defined by:
\begin{equation}
(N\!\bowtie_\Gamma\! M) (a,b \!\mid\! u,v) = N(a \!\mid\! \Gamma_1(u,v))\,M(b \!\mid\! \Gamma_2(u,v))
\end{equation}
It is further {\em split} into two channels  $N \boxast_\Gamma M$ and $N \varoast_\Gamma M$, defined by:
\begin{align}
(N \boxast_\Gamma  M)(a, b \mid u)    &= \frac{1}{4} \sum_v (N \bowtie_\Gamma M) (a,b \mid u,v), \\
(N \varoast_\Gamma M)(a, b, u \mid v) &= \frac{1}{4} (N \bowtie_\Gamma M) (a,b \mid u,v).
\end{align}

Applying the above construction to classical counterparts of CMP channels, we have the following proposition, proven in Appendix~\ref{sec:proof_cq_equiv}. 

\begin{proposition} \label{prop:cq_equiv}
Let ${\cal N}_{U\rightarrow A}$ and ${\cal M}_{V\rightarrow B}$ be two CMP channels, and $C$ be any two-qubit Clifford unitary, acting on the two qubit system $UV$. Let  ${\cal N}^\#$ and ${\cal M}^\#$ denote the two classical counterparts of the above CMP channels, and $\Gamma \eqdef \Gamma(C)$ be the permutation induced by the conjugate action of $C$ on $\bar{P}_1 \times \bar{P}_1$.
Then ${\cal N}\boxast_C{\cal M}$ and ${\cal N}\varoast_C{\cal M}$ are  identifiable to CMP channels, thus $({\cal N}\boxast_C{\cal M})^\#$ and $({\cal N}\varoast_C{\cal M})^\#$ are well defined, and the following properties  hold:
\begin{itemize}
\item[$(i)$] $({\cal N}\boxast_C{\cal M})^\# \equiv {\cal N}^\#\boxast_\Gamma{\cal M}^\#$.
\item[$(ii)$] $({\cal N}\varoast_C{\cal M})^\# \equiv {\cal N}^\#\varoast_\Gamma{\cal M}^\#$.
\end{itemize}
\end{proposition}

A consequence of the above proposition is that a CMP channel polarizes under the recursive application of the channel combining and splitting rules, if and only if its classical counterpart does so. Moreover,  processes of both quantum and classical polarization yield the same set of indices for the good/bad channels.  More precisely, we have the following:

\begin{corollary}
Let  ${\cal W}$ be a CMP channel,  and  ${\cal W}^{(i_1\cdots i_{n})}$ be defined recursively as in (\ref{eq:polar_coding_recursion}), $\forall n > 0$, $\forall i_1\cdots i_n\in\{0,1\}^n$. Let ${\cal W}^\#$ be the classical counterpart of ${\cal W}$, and $({\cal W}^\#)^{(i_1\cdots i_{n})}$ be defined recursively, similar to (\ref{eq:polar_coding_recursion}), while replacing ${\cal W}$ by ${\cal W}^\#$, and Clifford unitaries $C_{i_1\cdots i_n}$ by the corresponding permutations $\Gamma_{i_1\cdots i_n} := \Gamma(C_{i_1\cdots i_n})$.
Then $ \left({\cal W}^{(i_1\cdots i_{n})}\right)^\# \equiv ({\cal W}^\#)^{(i_1\cdots i_{n})}$, $\forall n, \forall i_1\cdots i_{n} \in \{0, 1\}^n$. In particular:
\begin{equation*}
\displaystyle \mathtt{I}\left(({\cal W}^\#)^{(i_1\cdots i_{n})}\right) = \frac{1 + I\left({\cal W}^{(i_1\cdots i_{n})}\right)}{2}.
\end{equation*}
\end{corollary}

As we already know that the quantum transform polarizes, it follows that the classical transform does also polarize. Moreover, a direct proof of the classical polarization can be derived by verifying the conditions from Lemma~\ref{lemma:stochastic_proc_convergence}, with stochastic process $\{T_n : n\geq 0\}$  given by Bhattacharyya parameter $Z$ of the  classical channels synthesized during the recursive construction. We recall below the definition of the Bhattacharyya parameter for a classical channel $W$, as defined in~\cite{sta09}. We shall restrict our attention to classical channels with input alphabet $\bar{P}_1$.  

\begin{definition}[\cite{sta09}]\label{def:bhattacharyya}
Let $W$ be a classical channel, with input alphabet $\bar{P}_1 \cong (\{0,1,2,3\}, \oplus)$ and output alphabet $Y$. For $u,u',d\in \bar{P}_1$, we define
\begin{align*}
Z(W_{u, u^\prime}) &:= \sum_{y\in Y} \sqrt{W(y|u) W(y|u^\prime)}.  \\ 
Z_d(W) &:= \frac{1}{4}\sum_{u \in \bar{P}_1} Z(W_{u, u\oplus d}). 
\end{align*}
In particular, note that  $Z(W_{u, u}) = 1, \forall u\in \bar{P}_1$, and $Z_0(W) = 1$. The Bhattacharyya parameter of $W$, denoted $Z(W)$, is then defined as
\begin{equation*} 
 Z(W) := \frac{1}{3} \sum_{d\neq 0} Z_d(W) = \frac{1}{12} \sum_{u \neq u^\prime}  Z(W_{u, u^\prime}). 
 \end{equation*}
\end{definition}

Polarization of the classical channel ${\cal W}^\#$  follows then from the lemma below, whose proof is provided in Appendix~\ref{apnd:pol9}.

\begin{lemma}\label{lem:z-good-classical-channel}
Let  ${\cal W}$ be a CMP channel and ${\cal W}^\#$ its classical counterpart. Given two instances of the channel ${\cal W}^\#$, we have that
\begin{equation}
\mbE_{\Gamma \in \Gamma(\mathcal{L})} Z\left({\cal W}^\# \varoast_{\Gamma} {\cal W}^\# \right) = \mbE_{\Gamma \in \Gamma(\mathcal{R})} Z\left({\cal W}^\# \varoast_{\Gamma} {\cal W}^\#\right) 
=  \frac{1}{3}Z({\cal W}^\#) + \frac{2}{3}Z({\cal W}^\#)^2, \label{eq:pauli_pol_9_clif}
\end{equation}
where $\Gamma(\mathcal{L})$ and $\Gamma(\mathcal{R})$ denote the set of permutations generated on $\bar{P}_2$ by the conjugate action of Cliffords in $\mathcal{L}$ and $\mathcal{R}$, respectively. 
\end{lemma}

\subsection{Polarization Using Only Three Clifford Gates}
\label{subsec:polarization_3_cliffords}
In this section, we show that for Pauli channels the set of channel combining Clifford gates can be reduced to three gates only, while still ensuring polarization.  Let ${\cal S}$ denote the set containing the Clifford gates $L_{1,3}$, $L_{2,2}$, and $L_{3,1}$ from Fig.~\ref{fig:cliff_generators}, and  $\Gamma(\mathcal{S})$ denote the corresponding set of permutations, namely $\Gamma(L_{1,3})$, $\Gamma(L_{2,2})$ and $\Gamma(L_{3,1})$, generated by the conjugate actions of $L_{1,3}$, $L_{2,2}$, and $L_{3,1}$ on $\bar{P}_1 \times \bar{P}_1$.

\begin{lemma} \label{lem:pol_3_clif}
Let  ${\cal W}$ be a CMP channel and ${\cal W}^\#$ its classical counterpart. Given two instances of the channel ${\cal W}^\#$, then
\begin{equation}\label{eq:pauli_pol_3_clif}
\mbE_{\Gamma \in  \Gamma(\mathcal{S})} Z\left(\mathcal{W}^\# \varoast_{\Gamma} \mathcal{W}^\# \right)  \leq  \frac{1}{3}Z(\mathcal{W}^\#) + \frac{2}{3}Z(\mathcal{W}^\#)^2.
\end{equation}
\end{lemma}

\noindent The proof is given in Appendix~\ref{apnd:pol3}.

\subsection{Speed of Polarization}

 Before discussing decoding of quantum polar codes over Pauli channels (Section~\ref{sec:decod_pauli_channel}), it is worth noticing that classical polar codes come equipped with a decoding algorithm, known as successive cancellation (SC) \cite{arikan09}.  However, the effectiveness of the classical SC decoding, {\em i.e.}, its capability of successfully decoding at rates close to the capacity, depends on the speed of polarization. The Bhattacharyya parameter of the synthesized channels plays an important role in determining the speed at which polarization takes place. First, we note that for a classical channel $W$,  the Bhattacharyya parameter upper bounds the error probability of uncoded transmission. Precisely, given a classical channel $W$ with input alphabet $X$, the error probability of the maximum-likelihood decoder for a single channel use, denoted $P_e$, is upper-bounded as follows (\cite[Proposition 2]{sta09}):
\begin{equation*}
P_e \leq (|X|-1)Z(W).
\end{equation*}
Now, consider a polar code defined by the recursive application of $n$ polarization steps to the classical channel $W:=\mathcal{W}^\#$ (the input alphabet is $X:=\bar{P}_1$, of size $|\bar{P}_1|=4$). The construction is the same as the one in Section~\ref{sec:quantum_polar_coding}, while replacing the quantum channel $\mathcal{W}$ by its classical counterpart $W$, and channel combining Clifford gates $C_{i_1 i_2\cdots}$ by the corresponding permutations $\Gamma_{i_1 i_2\cdots} := \Gamma(C_{i_1 i_2\cdots})$.  For any $i=0,\dots,N-1$, let $W^{(i)} := (\mathcal{W}^\#)^{(i_1\cdots i_n)}$, where $i_1\cdots i_n$ is the binary decomposition of $i$. For the sake of simplicity, we drop the channel combining permutations $\Gamma$'s from the above notation. Let $\mathcal{I} \subset\{0,1,\dots,N-1 \}$ denote the set of good channels ({\em i.e.}, channels used to transmit information symbols, as opposed to bad channels, which are frozen to symbol values known to both the encoder and  decoder). Since  the SC decoding proceeds by decoding successively the synthesized good channels\footnote{Each good channel is decoded by taking a maximum-likelihood decision, according to the observed channel output and the previously decoded channels.}, it can be easily seen that the block error probability of the SC decoder, denoted by $P_e(N, \mathcal{I})$,  is upper-bounded by (see also \cite[Proposition 2]{arikan09}):
\begin{equation}\label{eq:pe_sc_upperbound}
P_e(N, \mathcal{I}) \leq 3\sum_{i \in\mathcal{I}}Z(W^{(i)}).
\end{equation}
If the Bhattacharyya parameters of the $W^{(i)}$ channels, with $i \in\mathcal{I}$, converge sufficiently fast to zero, one can use (\ref{eq:pe_sc_upperbound}) to ensure that 
 $P_e(N, \mathcal{I})$ goes to zero.  Since the number of terms in the right hand side of (\ref{eq:pe_sc_upperbound}) is linear in $N$, it is actually enough to prove that $Z(W^{(i)}) \leq O(N^{-(1+\theta)}), \forall i\in \mathcal{I}$, for some $\theta > 0$.

\medskip The proof of fast polarization properties in \cite[Lemma 3]{sta09}, for channels with non-binary input alphabets, exploits two main ingredients:
\begin{itemize}
\item[(1)] The quadratic improvement of the Bhattacharyya parameter, when taking the good channel, {\em i.e.},  $Z\left(W^{(i_1\cdots i_{n-1}i_n)}\right) \leq Z\left(W^{(i_1\cdots i_{n-1})}\right)^2$, $\forall i_1\cdots i_{n-1}i_n\in\{0,1\}^n$, such that $i_n=1$.
\item[(2)] The linearly upper-bounded degradation of the Bhattacharyya parameter, when taking the bad channel, {\em i.e.},  $Z\left(W^{(i_1\cdots i_{n-1}i_n)}\right) \leq \kappa Z\left(W^{(i_1\cdots i_{n-1})}\right)$, $\forall i_1\cdots i_{n-1}i_n\in\{0,1\}^n$, such that $i_n=0$, for some constant $\kappa>0$.
\end{itemize}

Regarding the second condition, in our case we have the following lemma, where for a classical channel $W$ with input alphabet $\bar{P}_1 \cong  \{0,1,2,3\}$, we define
\begin{equation}
\bar{Z}(W) := \max_{d=1,2,3} Z_d(W).
\end{equation}

\begin{lemma}\label{lem:Fast_Polarization} For any classical channel $W$ with input alphabet $\bar{P}_1$, and any linear permutation $\Gamma : \bar{P}_1\times\bar{P}_1 \rightarrow \bar{P}_1\times\bar{P}_1$, the following inequalities hold:
\begin{align*}
\bar{Z}(W \boxast_\Gamma W)  &\leq 4 \bar{Z}(W). \\
Z(W \boxast_\Gamma W)  &\leq 12 Z(W).
\end{align*}
\end{lemma}
\noindent The proof is given in Appendix~\ref{sec:proof_lemma_fast_polarization}.

\medskip Condition (1) above -- quadratic improvement of the  Bhattacharyya parameter, when taking the good channel -- is more problematic, due to the linear term in the right hand side of (\ref{eq:pauli_pol_9_clif}) and~(\ref{eq:pauli_pol_3_clif}). In particular, we can not apply \cite[Lemma 3]{sta09} to derive fast polarization properties in our case. Instead, we will prove fast polarization properties by drawing upon arguments similar to those in the proof of \cite[Theorem 2]{arikan09}. 
First, we need the following definition. 

\begin{definition}
 Let  $W$ be a classical channel with input alphabet $\bar{P}_1$, and $\bm{\Gamma} = \{ \Gamma, \Gamma_{i_1\cdots i_n} \mid n > 0, \break i_1\cdots i_n \in\{0,1\}^n \}$ be an infinite sequence of permutations. For $n > 0$, let 
\begin{equation}\label{eq:c_polar_recursion}
W^{(i_1\cdots i_n)} := \left\{ \begin{array}{@{}ll@{}}
   W^{(i_1\cdots i_{n-1})} \boxast_{\Gamma_{i_1\cdots i_{n-1}}}  W^{(i_1\cdots i_{n-1})}, & \text{if } i_n = 0 \\
   W^{(i_1\cdots i_{n-1})} \varoast_{\Gamma_{i_1\cdots i_{n-1}}} W^{(i_1\cdots i_{n-1})}, & \text{if } i_n = 1
\end{array} \right.
\end{equation}
where, for $n=1$, in the right hand side term of the above equality, we set by convention $W^{(\varnothing)} := W$ and $\Gamma_\varnothing := \Gamma$. We say that $\bm{\Gamma}$ is a {\em polarizing sequence} (or that polarization happens for $\bm{\Gamma}$), if for any $\delta > 0$,
\begin{equation*}
\lim_{n\rightarrow\infty} \frac{\#\{(i_1\cdots i_{n}) \in \{0,1\}^{n} : \mathtt{I}\left( {\cal W}^{(i_1\cdots i_{n})} \right) \in (\delta, 1-\delta)  \}}{2^n} = 0.
\end{equation*}
\end{definition}

\medskip Note that different from (the classical counterpart of) Theorem~\ref{thm:quantum_polarization}, we consider here a given  sequence of permutations, instead of averaging over some set of  sequences. If $W=\mathcal{W}^\#$ is the classical counterpart of a CMP channel $\mathcal{W}$, by Lemma~\ref{lem:pol_3_clif}, we know that polarization happens when averaging over all the sequences $\bm{\Gamma} \in \Gamma(\mathcal{S})^\infty$.   As a consequence, there exists a subset $\Gamma(\mathcal{S})^\infty_{\text{pol}} \subset \Gamma(\mathcal{S})^\infty$ of positive probability\footnote{Note that $\Gamma(\mathcal{S})^{\infty}$ is the infinite product space of countable many copies of $\Gamma(\mathcal{S})$, and it is endowed with the infinite product probability measure, taking the uniform probability measure on each copy of $\Gamma(\mathcal{S})$. See \cite{hewitt2013real} for infinite product probability measures.}, such that polarization happens for any $\bm{\Gamma} \in \Gamma(\mathcal{S})^\infty_{\text{pol}}$. We are now ready to state the following fast polarization result, whose proof is given in Appendix~\ref{sec:proof_prop_fast_polarization}.

\begin{proposition}\label{prop:fast_polarization}
Let $\mathcal{W}$ be a CMP channel, $W:=\mathcal{W}^\#$ its classical counterpart, and $\mathcal{S}$ the set of three Clifford gates from Section~\ref{subsec:polarization_3_cliffords}. Then  the following {\em fast polarization property} holds for almost all  $\bm{\Gamma}$ sequences in $\Gamma(\mathcal{S})^\infty_{\text{pol}}$: 

\smallskip For any $\theta > 0$ and $R < \mathtt{I}(W)$, there exists a sequence of sets $\mathcal{I}_N \subset \{0,\dots,N-1\}$, $N\in\{1, 2, \dots,2^n,\dots\}$, such that $|\mathcal{I}_N| \geq NR$ and $Z\left( W^{(i)} \right) \leq O\left(N^{-(1+\theta)}\right)$, $\forall i\in \mathcal{I}_N$. In particular, the block error probability of  polar coding under SC decoding satisfies
\begin{equation*}
P_e(N, \mathcal{I}_N) \leq O\left(N^{-\theta}\right).
\end{equation*}
\end{proposition}

\subsection{Decoding the Quantum Polar Code by Using its Classical Counterpart} \label{sec:decod_pauli_channel}

Let ${\cal W}$ be a CMP channel and ${\cal W}^\#$ its classical counterpart. Let $G_q$ denote the unitary operator corresponding to the quantum polar code (defined by the recursive application of $n$ polarization steps, see Section~\ref{sec:quantum_polar_coding}), and $G_c$ denote the linear transformation corresponding to the classical polar code. Let ${\cal I}$ and ${\cal J}$ be the set of indices corresponding to the good and bad channels, respectively, with  $|{\cal I}| + |{\cal J}| = N \eqdef 2^n$.  We shall use the following notation from Section~\ref{sec:quantum_polar_coding}:

\begin{itemize}
\item $\rho_{\cal I}$ denotes the original state of system ${\cal I}$, 

\item $\varphi_{{\cal IJJ}'} \eqdef (G_q\otimes I_{{\cal J}'}) (\rho_{\cal I} \otimes \Phi_{\cal JJ'}) (G_q^\dagger \otimes I_{{\cal J}'})$ denotes the {\em encoded state}, where $\Phi_{\cal JJ'}$ is a maximally entangled state,  as defined in (\ref{eq:max_entangled_state}).

\item $\psi_{{\cal IJJ}'} \eqdef ({\cal W}^{\otimes N} \otimes I_{{\cal J}'}) (\varphi_{{\cal IJJ}'})$ denotes the {\em channel output state}.
\end{itemize} 
Since ${\cal W}$ is a CMP channel, it follows that:
\begin{equation*}
\psi_{{\cal IJJ}'}  = (E_{{\cal I}{\cal J}}G_q\otimes I_{{\cal J}'}) (\rho_{\cal I} \otimes \Phi_{\mathcal{J}\mathcal{J'}}) (G_q^\dagger E_{{\cal I}{\cal J}}^\dagger \otimes I_{{\cal J}'}).
\end{equation*}
for some {\em error}  $E_{{\cal I}{\cal J}} \in P_N$. Hence, quantum polar code decoding can be performed in the 4 steps described below.

\medskip \noindent {\bf Step 1: Apply the inverse quantum polar transform on the channel output state.}  Applying $G_q^\dagger$ on the output state $\psi_{{\cal IJJ}'}$, leaves the ${\cal IJJ}'$ system in  the following state:
\begin{align*}
\psi'_{{\cal IJJ}'} &=  (G_q^\dagger E_{{\cal I}{\cal J}}G_q\otimes I_{{\cal J}'}) (\rho_{\cal I} \otimes \Phi_{\mathcal{J}\mathcal{J'}}) (G_q^\dagger E_{{\cal I}{\cal J}}^\dagger G_q \otimes I_{{\cal J}'}) \nonumber \\
&=  (E'_{{\cal I}{\cal J}}\otimes I_{{\cal J}'}) (\rho_{\cal I} \otimes \Phi_{\mathcal{J}\mathcal{J'}}) (E{'}_{{\cal I}{\cal J}}^{\,\dagger} \otimes I_{{\cal J}'}).
\end{align*}
where $E'_{{\cal I}{\cal J}} \eqdef G_q^\dagger E_{{\cal I}{\cal J}}G_q$. Since we only need to correct up to a global phase, we may assume that $E'_{{\cal I}{\cal J}}, E_{{\cal I}{\cal J}} \in {P}_N / \{\pm 1, \pm i\} \simeq {\bar P}_1^{N}$, and thus write $E'_{{\cal I}{\cal J}} = G_c^{-1} E_{{\cal I}{\cal J}}$, or equivalently: 
\begin{equation*}
E_{{\cal I}{\cal J}} = G_c E'_{{\cal I}{\cal J}}.
\end{equation*}
Put differently,  $E_{{\cal I}{\cal J}}$ is the classical polar encoded version of $E'_{{\cal I}{\cal J}}$. 

\medskip \noindent {\bf Step 2: Quantum measurement.}\footnote{ Steps (1) and (2) together  perform a set of measurements that are equivalent to measuring the elements of the stabilizer set $\bar{\cal S}_\mathcal{IJJ'}$ defined in~(\ref{eq:stabilizer_set}).}
 Let $E'_{{\cal I}{\cal J}} = \displaystyle \mathop{\otimes}_{i\in {\cal I}} E'_i \mathop{\otimes}_{j\in {\cal J}} E'_j$, with $E'_i, E'_j \in \bar{P}_1$.  Measuring $X_jX_{j'}$ and $Z_jZ_{j'}$ observables\footnote{Here, indices $j$ and $j'$ indicate the $j$-th qubits of ${\cal J}$ and ${\cal J}'$ systems.}, allows determining the value of $E'_j$, for any $j \in {\cal J}$, since no errors occurred on the ${\cal J}'$ system.

\medskip \noindent {\bf Step 3: Decode the classical polar code counterpart.} We note that the error $E_{{\cal I}{\cal J}}$ can be seen as the output of the  classical  vector channel $({\cal W}^\#)^{N}$, when the {\em ``all-identity vector''}  $\sigma_0^N \in \bar{P}_1^N$ is applied at the channel input.  However, by the definition of the classical  channel ${\cal W}^\#$, we have $({\cal W}^\#)^{N}(E_{{\cal I}{\cal J}} \mid \sigma_0^N) = ({\cal W}^\#)^{N}(\sigma_0^N \mid E_{{\cal I}{\cal J}})$, meaning that we can equivalently consider $\sigma_0^N$ as being the observed channel output, and  $E_{{\cal I}{\cal J}}$ the (unknown) channel input. 
Hence, we have given $(i)$ the value of $E'_{\cal J} \eqdef  \mathop{\otimes}_{j\in {\cal J}} E'_j$, and $(ii)$ a noisy observation (namely $\sigma_0^N$) of $E_{{\cal I}{\cal J}} = G_c E'_{{\cal I}{\cal J}}$. We can then use classical polar code decoding to recover the value of $E'_{\cal I} \eqdef  \mathop{\otimes}_{i\in {\cal I}} E'_i$.

\medskip \noindent {\bf Step 4: Error correction.} Once we have recovered the $E'_{\cal J}$ (step 2) and $E'_{\cal I}$ (step 3) values, we can apply the $E'_{{\cal I}{\cal J}} \otimes I_{{\cal J}'}$ operator on $\psi'_{{\cal IJJ}'}$, thus leaving the ${\cal IJJ}'$ system in the state $\rho_{\cal I} \otimes \Phi_{\mathcal{J}\mathcal{J'}}$.

\section{Polarization with Vanishing Rate of Preshared Entanglement}
\label{sec:vanishing-rate-preshared-entanglement}
In this section we present a code construction using an asymptotically vanishing rate or preshared entanglement, while achieving a transmission rate equal
to the coherent information of the channel. In particular,  we shall assume that the coherent information of the channel is positive, $I({\cal W}) > 0$. The proposed construction bears similarities to the universal polar code construction in \cite[Section V]{hassani2014universal}, capable of achieving the compound capacity of a finite set of classical channels. 

\medskip 
Let $P_q(N, \mathcal{J}, \mathcal{I})$ denote a quantum polar code of length $N=2^n$, for some $n > 0$, where  $\mathcal{I}$ and  $\mathcal{J}$  denote the sets of good and bad channels respectively.  By Theorem~\ref{thm:quantum_polarization}, as $n$ goes to infinity,  $|{\cal I}|$ approaches $\frac{1+I(\mathcal{W})}{2}N$, and thus $|{\cal J}|$ approaches $\frac{1-I(\mathcal{W})}{2}N$. Since $I(\mathcal{W}) > 0$, it follows that $|{\cal J}| < |{\cal I}|$, provided that $n$ is large enough. Therefore, we may find a subset of  good channels $\mathcal{I}' \subset \mathcal{I}$, such that $|\mathcal{I}'| = |\mathcal{J}|$. In the sequel, we shall extend the definition of a polar code to include such a subset ${\cal I}'$, and denote it by $P_q(N, \mathcal{J}, \mathcal{I}, \mathcal{I}')$.

\medskip
Let us now consider $k$ copies of a quantum polar code $P_q(N, \mathcal{J}, \mathcal{I}, \mathcal{I}')$, denoted by\break $P_q^l(N, \mathcal{J}_l, \mathcal{I}_l, \mathcal{I}'_l)$ or simply by $P_q^l$, for any $ l \in \{0, 1, \dots, k-1 \}$. We define a quantum code $C_q^k$ of codelength $|C_q^k| = kN$, by {\em chaining} them in the following way (see also Fig.~\ref{fig:three_polar_Catal}):

\begin{figure}[!t]
\centering
\begin{tikzpicture}
\draw
(0, 0) to node[above](c){} ++(1,0) to++(0,-0.1)
(0,0)++(1, 0) to ++(0,0.1) 
(0,0)++(1, 0) to node[above](c){} ++(.75, 0) to ++(0,0.1) 
(0,0)++(1, 0) to ++(.75, 0) to ++(0,-0.1)
(0,0)++(1, 0) to ++(.75, 0) to ++(1, 0)
;
\draw [decorate,decoration={brace,amplitude=10pt},xshift=0pt,yshift=4pt]
(0,0) -- (1,0) node [black,midway,yshift=.5cm]
{\footnotesize $\mathcal{J}_0$}
;
\draw [decorate,decoration={brace,amplitude=10pt, mirror, raise= 4pt},xshift=0pt,yshift=-3pt]
(1,0) -- (2.75,0) node [black,midway,yshift=-.65cm]
{\footnotesize $\mathcal{I}_0$}
;
\draw [decorate,decoration={brace,amplitude=10pt},xshift=0pt,yshift=4pt]
(1.75,0) -- (2.75,0) node [black,midway,yshift=.5cm]
{\footnotesize $\mathcal{I}'_0$}
;
\draw 
(0, 0) ++(1.37, 0) ++ (0, -1.1) node[]{$P_q^0$}
;
\draw
(4,0) to node[above](c){} ++(1,0) to++(0,-0.1)
(4,0)++(1, 0) to ++(0,0.1) 
(4,0)++(1, 0) to node[above](c){} ++(.75, 0) to ++(0,0.1) 
(4,0)++(1, 0) to ++(.75, 0) to ++(0,-0.1)
(4,0)++(1, 0) to ++(.75, 0) to ++(1, 0)
;
\draw [decorate,decoration={brace,amplitude=10pt},xshift=0pt,yshift=4pt]
(4,0) -- (5,0) node [black,midway,yshift=.5cm]
{\footnotesize $\mathcal{J}_1$}
;
\draw [decorate,decoration={brace,amplitude=10pt, mirror, raise= 4pt},xshift=0pt,yshift=-3pt]
(5,0) -- (6.75,0) node [black,midway,yshift=-.65cm]
{\footnotesize $\mathcal{I}_1$}
;
\draw [decorate,decoration={brace,amplitude=10pt},xshift=0pt,yshift=4pt]
(5.75,0) -- (6.75,0) node [black,midway,yshift=.5cm]
{\footnotesize $\mathcal{I}'_1$}
;
\draw 
(4, 0) ++(1.37, 0) ++ (0, -1.1) node[]{$P_q^1$}
;
\draw
(8, 0) to node[above](c){} ++(1,0) to++(0,-0.1)
(8,0)++(1, 0) to ++(0,0.1) 
(8,0)++(1, 0) to node[above](c){} ++(.75, 0) to ++(0,0.1) 
(8,0)++(1, 0) to ++(.75, 0) to ++(0,-0.1)
(8,0)++(1, 0) to ++(.75, 0) to ++(1, 0)
;
\draw [decorate,decoration={brace,amplitude=10pt},xshift=0pt,yshift=4pt]
(8,0) -- (9,0) node [black,midway,yshift=.5cm]
{\footnotesize $\mathcal{J}_2$}
;
\draw [decorate,decoration={brace,amplitude=10pt, mirror, raise= 4pt},xshift=0pt,yshift=-3pt]
(9,0) -- (10.75,0) node [black,midway,yshift=-.65cm]
{\footnotesize $\mathcal{I}_2$}
;
\draw [decorate,decoration={brace,amplitude=10pt},xshift=0pt,yshift=4pt]
(9.75,0) -- (10.75,0) node [black,midway,yshift=.5cm]
{\footnotesize $\mathcal{I}'_2$}
;
\draw(2.25, 0) ++(0, 0.8) to ++(0, .3) to ++(1.1, .75) to ++(1.1, -.75) to ++(0, -.3)
(6.25, 0) ++(0, 0.8) to ++(0, .3) to ++(1.1, .75) to ++(1.1, -.75) to ++(0, -.3)
;
\draw
 (3.35,2.1) node{\footnotesize  $\Phi_{\mathcal{I}'_{0}\mathcal{J}_{1}}$}
 (7.35,2.1) node{\footnotesize  $\Phi_{\mathcal{I}'_{1}\mathcal{J}_{2}}$}
;
\draw 
(8, 0) ++(1.37, 0) ++ (0, -1.1) node[]{$P_q^2$}
;
\draw [decorate,decoration={brace,amplitude=10pt},xshift=0pt,yshift=4pt]
(-2.25,0) -- (-1.25,0) node [black,midway,yshift=.5cm]
{\footnotesize $\mathcal{J}'_0$}
;
\draw [dashed]
(-2.25,0) to (-1.25,0)
;
\draw
(-2.25,-0.1) to (-2.25,0.1)
(-1.25,-0.1) to (-1.25,0.1)
;
\draw [dashed] (-1.75, 0) ++(0, 0.8) to ++(0, .3) to ++(1.1, .75) 
;
\draw
(-0.65, 1.85) to ++(1.1, -.75) to ++(0, -.3)
;
\draw
 (-0.65,2.1) node{\footnotesize  $\Phi_{\mathcal{J}'_{0}\mathcal{J}_{0}}$}
;
\draw [dashed]
 (-0.65,1) to (-0.65,-1.5)
;
\draw
 (-0.65,-1.1) node[left]{\footnotesize Receiver}
 (-0.65,-1.1) node[right]{\footnotesize Sender}
;
\end{tikzpicture}
\caption{$C_q^3$: Chaining construction with $k=3$ copies of a quantum polar codes $P_q$}
\label{fig:three_polar_Catal}
\end{figure}
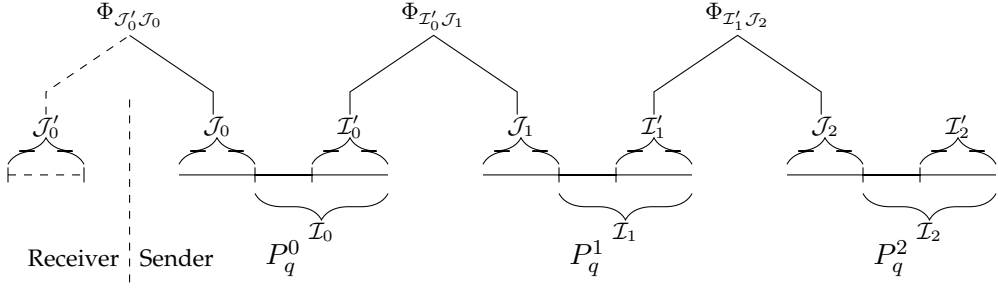

\begin{itemize}
\item[$(i)$] For system $\mathcal{J}_0$, the input quantum state before encoding is half of a maximally entangled state $\Phi_{\mathcal{J}_0\mathcal{J}'_0}$, where system $\mathcal{J}'_0$ is part of channel output. This is the only preshared entanglement between the sender and the receiver. 

\item[$(ii)$] For systems $\mathcal{I}'_{l-1}$ and $\mathcal{J}_{l}$, with $l \neq 0$, the input quantum state before encoding is a maximally entangled state $\Phi_{\mathcal{I}'_{l-1}\mathcal{J}_{l}}$. 

\item[$(iii)$] Systems $ \mathcal{I}_l \setminus \mathcal{I}'_l$, for $l \neq k - 1$, and $\mathcal{I}_{k-1}$ are information systems, meaning that the corresponding quantum state is the one that needs to be transmitted from the sender to the receiver.
\end{itemize}

It can be easily seen that the transmission (coding) rate of the proposed scheme is given by
\begin{equation*}
R := \frac{\sum_{l=0}^{k-2} |\mathcal{I}_l \setminus \mathcal{I}'_l| + |\mathcal{I}_{k-1}|}{kN}  \xrightarrow[n\rightarrow \infty]{} 
     \frac{(k-1)I(\mathcal{W}) + \frac{1+I(\mathcal{W})}{2}}{k}  \xrightarrow[k\rightarrow \infty]{} I(\mathcal{W}),
\end{equation*}
while the rate of preshared entanglement is given by
\begin{equation*}
E := \frac{|\mathcal{J}_0|}{kN} \xrightarrow[n\rightarrow \infty]{} \frac{1-I(\mathcal{W})}{2k} \xrightarrow[k\rightarrow \infty]{} 0.
\end{equation*}

\noindent \textbf{Decoding $C_q^k$:}  We shall assume that we are given an effective decoding algorithm of the quantum polar code $P_q$, capable of achieving the symmetric coherent information of the channel. We note that  this is indeed the case for Pauli channels (Section~\ref{sec:decod_pauli_channel}), but it is an open problem for general quantum channels.  In this case,  $C_q^k$ can be decoded sequentially, by decoding first $P_q^0$, then  $P_q^1$, $P_q^2$, and so on. Indeed, after decoding $P_q^0$, thus in particular correcting the  state of the $\mathcal{I}'_0$ system,  the EPR pairs $\Phi_{\mathcal{I}'_0\mathcal{J}_1}$ will play the role of the preshared entanglement required to decode $P_q^1$. Therefore, $P_q^1$ can be decoded  once $P_q^0$ has been decoded, and similarly, $P_q^l$ can be decoded after $P_q^{l - 1}$ has been decoded, for any $ l \in \{2, \dots, k - 1 \}$. 

\medskip \noindent {\bf Entanglement as a catalyst:}  Finally, the above coding scheme can be slightly modified, such that preshared entanglement  between the sender and the receiver is not consumed. In the above construction, we have considered that for the last $P_q^{k-1}$ polar code, the $\mathcal{I}'_{k-1}$ system is an information system, {\em i.e.}, used to transmit quantum information from the sender to the receiver (system $\mathcal{I}'_{2}$ in Fig.~\ref{fig:three_polar_Catal}). Let us now assume that the input quantum state to the $\mathcal{I}'_{k-1}$ system is half of a maximally entangled state $\Phi_{\mathcal{I}'_{k-1} \mathcal{J}_{k}}$, where quantum system $\mathcal{J}_{k}$ is held by the sender. When the receiver completes decoding of the $C_q^k$ code, it restores the initial state of the $\mathcal{I}'_{k-1}$, thus resulting in a maximally entangled state $\Phi_{\mathcal{I}'_{k-1} \mathcal{J}_{k}}$ shared between the sender ($\mathcal{J}_{k}$ system) and the receiver ($\mathcal{I}'_{k-1}$ system). Hence, the initial preshared entanglement $\Phi_{\mathcal{J}_0\mathcal{J}'_0}$ acts as a catalyst, in that it produces a new state $\Phi_{\mathcal{I}'_{k-1} \mathcal{J}_{k}}$ shared between the sender and the receiver, which can be used for the next transmission. 

\section{Conclusion and Perspectives}
\label{sec:conclusion}
In this paper, we demonstrated an entanglement assisted polarization phenomenon that happens at the quantum level, and which relies on a channel combining construction using randomized two-qubit Clifford gates instead of the CNOT gate. In the case of Pauli channels, we have proven that  the  quantum polarization is equivalent to a classical  polarization for an associated non-binary channel which allows us to have an efficient decoding scheme. We also proved a fast polarization property in this case. Finally, we presented a quantum polar code chaining construction, for which the required entanglement assistance is negligible with respect to the code length.

\smallskip The quantum polar code construction proposed in this paper makes an efficient use of preshared entanglement, achieving a quantum communication rate equal to half the symmetric mutual information of the quantum channel. This differentiates our construction from the CSS-based one.  In terms of net communication rate, both constructions exhibit the same asymptotic performance, achieving the symmetric coherent information of the quantum channel.  However, for finite code-lengths, we expect that our scheme may have some benefits, due to more degrees of freedom in the construction of the polar code. We expect that an appropriate choice of the Clifford unitaries in the channel combining step may increase not only the polarization speed, but also the minimum distance of the code. Increasing the minimum distance is relevant in case that list decoding \cite{tal2015list} is used, which may significantly improve the error correction capability for finite code-lengths. Finite-length specific issues ({\em e.g.}, devising optimization algorithms for code construction, or considering decoding algorithms other than successive-cancellation) may be the subject of future research.

\smallskip Besides, the quantum polarization phenomenon presented here opens new perspectives that complement or extend the classical CSS-based viewpoint, and which are discussed below. 

\smallskip 1) One of the main characteristics of our approach is to yield a family of quantum codes, whose construction does not rely exclusively on the stabilizer formalism. This invites the study of quantum decoding algorithms that might achieve the symmetric coherent information for non Pauli channels (as explained in Section~\ref{subsec:EASC}, stabilizer decoding is not optimal for non-Pauli channels). Explicit capacity-achieving decoders for non Pauli channels have been recently proposed in \cite{renes2017belief}, based on a belief-propagation algorithm that passes quantum messages and is capable of decoding the classical–quantum channel with pure state outputs. Since the successive-cancellation decoding of polar codes is essentially a belief-propagation algorithm, it would be  interesting to devise  similar approaches for the family of quantum polar codes proposed here.

\smallskip 2) Another question with potentially wide-reaching implications is related to the use of quantum polar codes in the context of fault-tolerant computing. The preshared entanglement requirement amounts to having a reliable quantum memory to store one half of each EPR pair. In this context, it would be interesting to study the impact of a ``noisy'' channel combining step on the quantum polarization.

\smallskip 3)  Finally, while this paper focused on the polarization of qubit-input channels, our construction admits a simple and natural  generalization to the case of qudit-input channels.  Indeed, the definitions of symmetric coherent information and Rényi-Bhattacharyya parameter are valid for qudit channels.  Therefore, to prove quantum polarization for qudit channels, we again need to show that all the constraints in Lemma~\ref{lemma:stochastic_proc_convergence} are satisfied. The most challenging constraint in Lemma~\ref{lemma:stochastic_proc_convergence} is point (t.2),  which in the qubit case follows from Lemma~\ref{lem:d-good-channel}. It is not too difficult to see that Lemma~\ref{lem:d-good-channel} can be generalized to the qudit case with a few adjustments depending on the system dimension, reusing again the computations in \cite{fred-these}, and given that the channel combining operation is randomly chosen from a unitary 2-design.


\clearpage
\appendix

\section{Proof of Lemma~\ref{lemma:equiv_conditions}}
\label{sec:proof-quantum-red-clifford}

Throughout this section, $\mathcal{N} := \mathcal{N}_{A_1' \to B_1}$ and $\mathcal{M} := \mathcal{M}_{A_2' \to B_2}$ denote two quantum channels with input qubit systems $A_1'$ and $A_2'$. We denote by $\mathcal{N}_{A_1' \to E_1}^c$ and $\mathcal{M}_{A_2' \to E_2}^c$ their complementary channels, and by $\Phi_{A_1 A_1'}$ and $\Phi_{A_2' A_2}$ two EPR pairs on $2$-qubit systems $A_1A_1'$ and $A_2' A_2$. First, we will need the following lemma, which basically states that taking the complementary of the bad/good channels inverts the bad and good channel constructions, as well as the input system on which we input the half of the EPR pair. 
\begin{lemma} \label{lem:complementary_gd_bad}
Let $C := C_{A'_1 A'_2}$ be a $2$-qubit Clifford unitary acting on the $A'_1 A'_2$ system.  Then the bad channel's complementary  $(\mathcal{N} \boxast_C \mathcal{M})^c$ is a channel from $A_1'$ to $E_1E_2A_2$, the good channel's complementary  $(\mathcal{N} \varoast_C \mathcal{M})^c$ is a channel from $A_2'$ to $E_1E_2$, and the following equalities hold:
%
\begin{itemize}
\item[$(a)$] $(\mathcal{N} \boxast_C \mathcal{M})_{A_1' \to E_1E_2A_2}^c(\rho_{A_1'}) = \mathcal{N}_{A_1' \to E_1}^c \otimes \mathcal{M}_{A_2' \to E_2}^c \left(C \left(\rho_{ A_1'} \otimes \Phi_{A_2' A_2}\right)C^\dagger \right)$,
\item[$(b)$] $(\mathcal{N} \varoast_C \mathcal{M})_{A_2' \to E_1E_2}^c(\rho_{A_2'})= \mathcal{N}_{A_1' \to E_1}^c \otimes \mathcal{M}_{A_2'\to E_2}^c \left(C \left(\tfrac{\ident_{A_1'}}{2} \otimes \rho_{A_2'} \right)C^\dagger \right)$.
\end{itemize}
\end{lemma}

\begin{proof}
Let $U_{A_1'\to B_1E_1}$ and $V_{A_2'\to B_2E_2}$ be Stinespring dilations of $\mathcal{N}_{A_1'\to B_1}$ and $\mathcal{M}_{A_2'\to B_2}$, respectively. Define isometries $W_{A_1'\to B_1B_2E_1E_2A_2}$ and $W_{A_2'\to B_1B_2E_1E_2 A_1}'$, such that
\begin{align*} 
W (\rho_{A_1'}) W^\dagger  &\eqdef U \otimes V \left(C \left( \rho_{A_1'} \otimes \Phi_{A_2'A_2} \right)C^\dagger \right) U^\dagger \otimes V^\dagger, \\
W' (\rho_{A_2'}) W'^\dagger  &\eqdef U \otimes V \left(C \left( \Phi_{A_1A_1'} \otimes \rho_{A_2'} \right)C^\dagger \right) U^\dagger \otimes V^\dagger.
\end{align*}
It is easy to see that 
\begin{align*}
\tr_{E_1E_2A_2} (W (\rho_{A_1'}) W^\dagger) &= \mathcal{N} \otimes \mathcal{M} \left(C \left(\rho_{A_1'} \otimes \tfrac{\ident_{A_2'}}{2} \right)C^\dagger \right)  = (\mathcal{N} \boxast_C \mathcal{M})(\rho_{A_1'}),\\
\tr_{E_1E_2} (W' (\rho_{A_2'}) W'^\dagger) &= \mathcal{N} \otimes \mathcal{M} \left(C \left(\Phi_{A_1A_1'} \otimes \rho_{A_2'} \right)C^\dagger \right) = (\mathcal{N} \varoast_C \mathcal{M})(\rho_{A_2'}).
\end{align*} 
 Thus, isometries $W_{A_1'\to B_1B_2E_1E_2A_2}$ and $W_{A_2'\to B_1B_2E_1E_2A_1}'$ are  Stinespring dilations of $(\mathcal{N} \boxast_C \mathcal{M})$ and $(\mathcal{N} \varoast_C \mathcal{M})$, respectively (see also~\cite[Theorem 1]{holevo07}). Therefore, by tracing out the channel outputs of  $(\mathcal{N} \boxast_C \mathcal{M})$ and  $(\mathcal{N} \varoast_C \mathcal{M})$, we get their respective complementary channels:
\begin{align*}
(\mathcal{N} \boxast_C \mathcal{M})^c(\rho_{A_1'}) &= \tr_{B_1B_2} (W (\rho_{A_1'}) W^\dagger) = \mathcal{N}^c \otimes \mathcal{M}^c \left(C \left(\rho_{A_1'} \otimes \Phi_{ A_2'A_2}\right)C^\dagger \right), \\
(\mathcal{N} \varoast_C \mathcal{M})^c(\rho_{A_2'}) &= \tr_{A_1B_1B_2} (W' (\rho_{A_2'}) W'^\dagger) = \mathcal{N}^c \otimes \mathcal{M}^c \left(C \left(\tfrac{\ident_{A_1'}}{2} \otimes \rho_{A_2'}\right)C^\dagger\right). 
\end{align*}
This concludes the proof of the lemma.
\end{proof}

To prove Lemma~\ref{lemma:equiv_conditions}, we need to show that two equivalent Clifford unitaries $C'  \sim C'' \in \mathcal{C}_2$ yield bad/good quantum channels with identical Rényi-Bhattacharyya parameter, when used for combining (then splitting) two quantum channels $\mathcal{N}$ and $\mathcal{M}$. Here, we first show this for the bad channel, and then for the good channel.

\subsection{Bad Channel} \label{subsec:red_cliff_bad}

\noindent We prove first the following lemma.
\begin{lemma}\label{prop:equi_cliff_bad}
Given $C_{A_1'A_2'}'' = C_{A_1'A_2'}' (C_{A_1'}^1 \otimes C_{A_2'}^2)$, let 
\begin{equation*}
\rho' \eqdef (\mathcal{N} \boxast_{C'} \mathcal{M})^c_{A_1' \to E_1E_2A_2}(\Phi_{A_1 A_1'}) 
\ \ \ \text{ and } \ \ \ 
\rho'' \eqdef  (\mathcal{N} \boxast_{C''} \mathcal{M})^c_{A_1' \to E_1E_2A_2}(\Phi_{A_1 A_1'}).
\end{equation*}
  Then,
 \begin{equation}
\rho'' = C_{A_1}^{1^{\top}} \otimes C_{A_2}^{2^{\top}} (\rho' ) \bar{C}^1_{A_1} \otimes \bar{C}^2_{A_2}. \nonumber
\end{equation}
\end{lemma}

\begin{proof}
We have:
\begin{align*}
\rho'' &=  \mathcal{N}_{A_1' \rightarrow E_1}^c \otimes \mathcal{M}_{A_2' \rightarrow E_2}^c \left( C_{A_1'A_2'}' (C_{A_1'}^1 \otimes C_{A_2'}^2) \left(\Phi_{A_1A_1'} \otimes \Phi_{A_2' A_2} \right) ({C_{A_1'}^1}^\dagger \otimes {C_{A_2'}^2}^\dagger ) C_{A_1'A_2'}'^\dagger \right) \\
&=  \mathcal{N}_{A_1' \rightarrow E_1}^c \otimes \mathcal{M}_{A_2' \rightarrow E_2}^c \left( C_{A_1'A_2'}' (C_{A_1}^{1^\top}\otimes C_{A_2}^{2^\top}) \left(\Phi_{A_1A_1'} \otimes \Phi_{A_2 'A_2}\right) (\bar{C}_{A_1}^1 \otimes \bar{C}_{A_2}^{2}) C_{A_1'A_2'}'^\dagger \right)   \\
&=  C_{A_1}^{1^\top}\otimes C_{A_2}^{2^\top} \left(\mathcal{N}_{A_1' \rightarrow E_1}^c \otimes \mathcal{M}_{A_2' \rightarrow E_2}^c \left( C_{A_1'A_2'}'  \left(\Phi_{A_1A_1'} \otimes  \Phi_{A_2 'A_2}\right) C_{A_1'A_2'}'^\dagger \right) \right)\bar{C}_{A_1}^1 \otimes \bar{C}_{A_2}^{2}  \\
&=  C_{A_1}^{1^\top} \otimes C_{A_2}^{2^\top} (\rho') \bar{C}_{A_1}^1 \otimes \bar{C}_{A_2}^{2}, 
\end{align*}
where the first equality follows from part (a) of Lemma~\ref{lem:complementary_gd_bad} and second equality follows from the relation $(\ident \otimes Z) \ket{\Phi} = (Z^\top \otimes \ident) \ket{\Phi}$, for any matrix $Z$.
\end{proof}

\medskip \noindent 
\textbf{Proof of $R( \mathcal{N} \boxast_{C''} \mathcal{M} )  =  R( \mathcal{N} \boxast_{C'} \mathcal{M} )$.} By the definition of the Rényi-Bhattacharyya parameter, see~(\ref{def:renyi_bhat}), we have that
\begin{equation*}
R( \mathcal{N} \boxast_{C''} \mathcal{M} ) = 2^{-\tilde{H}_2^{\downarrow}(A_1|E_1E_2A_2)_{\rho''}},
\end{equation*}
where $ \htwo(A|B)_\rho = -\Tilde{D}_2(\rho_{AB} || I \otimes \rho_B ),$ with $ \Tilde{D}_2(\rho || \sigma)  = \log\tr\left[(\sigma^{-\frac{1}{4}} \rho \sigma^{-\frac{1}{4}})^2\right]$~\cite{mdsft13, tbh14}. We have the following  unitary equivalence for $\Tilde{D}_2(\rho || \sigma)$,
\begin{equation} \label{eq:uni_equiv}
  \Tilde{D}_2(\rho || \sigma) = \Tilde{D}_2(U \rho U^\dagger || U \sigma U^\dagger ).
 \end{equation}
Hence,
\begin{align*}
\htwo(A_1|E_1E_2A_2)_{\rho''}  &= - \Tilde{D}_2(\rho''|| \ident \otimes \text{Tr}_{A_1}(\rho''))\\
&= - \Tilde{D}_2(C_{A_1}^{1^\top} \otimes C_{A_2}^{2^\top} (\rho') \bar{C}_{A_1}^1 \otimes \bar{C}_{A_2}^{2} || \ident \otimes (C_{A_2}^{2^\top} \text{Tr}_{A_1}(\rho') \bar{C}_{A_2}^{2}) )\\
&= -  \Tilde{D}_2(\rho' || \ident \otimes \text{Tr}_{A_1}(\rho') ) \\
&=  \htwo(A_1|E_1E_2A_2)_{\rho'}, 
\end{align*}
where the second equality follows from Lemma~\ref{prop:equi_cliff_bad} and $\tr_{A_1}(\rho'') = C_{A_2}^{2^\top} \tr_{A_1}(\rho') \bar{C}_{A_2}^{2}$, and the third equality follows from the unitary equivalence above. Therefore,  we get $R( \mathcal{N} \boxast_{C''} \mathcal{M} ) =  R( \mathcal{N} \boxast_{C'} \mathcal{M} )$, as desired.
\ \hfill $\blacksquare$

\subsection{Good Channel} \label{subsec:red_cliff_good}

\begin{lemma} \label{prop:equi_cliff_good}   
Given $C_{A_1'A_2'}'' = C_{A_1'A_2'}' (C_{A_1'}^1 \otimes C_{A_2'}^2)$, let  
\begin{equation*}  
\rho'  \eqdef (\mathcal{N} \varoast_{C'} \mathcal{M})^c_{A_2' \to E_1E_2} (\Phi_{A_2 A_2'})
\ \ \ \text{ and } \ \ \  
\rho'' \eqdef (\mathcal{N} \varoast_{C''} \mathcal{M})^c_{A_2' \to E_1E_2}(\Phi_{A_2 A_2'}). 
\end{equation*} 
Then
 $$ \rho'' = C_{A_2}^{2^\top} \rho' \bar{C}_{A_2}^{2}.$$  
\end{lemma}

\begin{proof} 
Similar to the proof of Lemma~\ref{prop:equi_cliff_bad}.
\end{proof}

\medskip \noindent 
\textbf{Proof of $R( \mathcal{N} \varoast_{C''} \mathcal{M} )  =  R( \mathcal{N} \varoast_{C'} \mathcal{M} )$.} 
Using Lemma~\ref{prop:equi_cliff_good}, it can be proved similar to the proof for bad channel in Subsection~\ref{subsec:red_cliff_bad} that $R( \mathcal{N} \varoast_{C''} \mathcal{M} )  =  R( \mathcal{N} \varoast_{C'} \mathcal{M} )$.
\hfill $\blacksquare$

\section{Proof of Lemma~\ref{lem:nclassical-well-defined}}\label{sec:proof-nclassical-well-defined}

We have to prove that if ${\cal N}'$ and ${\cal N}''$ are CMP channels, such that 
\begin{equation}\label{eq:identifiable_channels}
{\cal N}'(\rho) \otimes \frac{I_X}{|X|} = C \left({\cal N}''(\rho) \otimes \frac{I_X}{|X|}\right)C^\dagger,
\end{equation} for some unitary $C$, then $\left({\cal N}'\right)^\# \equiv \left({\cal N}''\right)^\#$. We restrict ourselves to the case when ${\cal N}'$ and ${\cal N}''$ are Pauli channels, since the case of CMP channels follows in a similar manner, by introducing an auxiliary system providing a classical description of the Pauli channel being used. Hence, we may write ${\cal N}'(\rho) = \sum_{i=0}^3 p'_i \sigma_i \rho \sigma_i^{\dagger}$ and ${\cal N}''(\rho) = \sum_{i=0}^3 p''_i \sigma_i \rho \sigma_i^{\dagger}$, with $\sum_{i=0}^{3}p'_i = \sum_{i=0}^{3}p''_i = 1$. It follows that ${\cal N}'(\sigma_k) = \alpha'_k \sigma_k$ and ${\cal N}''(\sigma_k) = \alpha''_k \sigma_k$,  where $\alpha'_0 = \alpha''_0 = 1$,  and for $k=1,2,3$, $\alpha'_k = p'_0+p'_k - p'_{k_1} - p'_{k_2}$, $\alpha''_k = p''_0+p''_k - p''_{k_1} - p''_{k_2}$, with $\{k_1, k_2\} = \{1,2,3\}\setminus \{k\}$. Using bold notation for vectors $\bm{p}' \eqdef (p'_0,p'_1, p'_2,p'_3)$, and similarly  $\bm{p}'', \bm{\alpha}', \bm{\alpha}''$, the above equalities rewrite as
\begin{equation}\label{eq:a-matrix}
\bm{\alpha}' = A \bm{p}' \mbox{ and } \bm{\alpha}'' = A \bm{p}'', 
\end{equation}
$$\text{ where } A \eqdef \left(\begin{array}{rrrr}
  1  &  1  &  1  &  1 \\
  1  &  1  & -1  & -1 \\
  1  & -1  &  1  & -1 \\
  1  & -1  & -1  &  1 
\end{array}\right)$$
Now, replacing $\rho$ by $\sigma_k$ in (\ref{eq:identifiable_channels}), we have that 
\begin{equation}
\alpha'_k \sigma_k \otimes I_X = C \left( \alpha''_k \sigma_k \otimes I_X \right) C^\dagger.
\end{equation}
Since the conjugate action of the unitary $C$ preserves the Hilbert–Schmidt norm of an operator, it follows that 
$\left\| \alpha'_k \sigma_k \otimes I_X \right\|_{\text{HS}} = \left\|  \alpha''_k \sigma_k \otimes I_X \right\|_{\text{HS}}$, and therefore $|\alpha'_k| = |\alpha''_k|$.

\medskip \noindent {\bf Case 1:} We first assume that $\alpha'_k = \alpha''_k, \forall k=1,2,3$. In this case, using (\ref{eq:a-matrix}), it follows that $\bm{p}' = \bm{p}''$, and therefore $\left({\cal N}'\right)^\# = \left({\cal N}''\right)^\#$.

\medskip \noindent {\bf Case 2:} We consider now the case when $\alpha'_k \neq \alpha''_k$, for some $k=1,2,3$. To address this case, we start by writing $C = \sum_{i=0}^3 \sigma_i \otimes C_i$, where $C_i$ are linear operators on the system $X$. Hence, (\ref{eq:identifiable_channels}) rewrites as 
\begin{equation}
{\cal N}'(\rho) \otimes \frac{I_X}{|X|} = \sum_{i,j}\left(\sigma_i {\cal N}''(\rho)\sigma_j^\dagger\right) \otimes \frac{C_i C_j^\dagger}{|X|}. 
\end{equation}
Tracing out the $X$ system, we have 
\begin{equation}\label{eq:nprime}
{\cal N}'(\rho) = \sum_{i,j} \gamma_{i,j} \sigma_i {\cal N}''(\rho)\sigma_j^\dagger, \ \ \mbox{ where } \gamma_{i,j} = \frac{\tr(C_i C_j^\dagger)}{|X|}. 
\end{equation}
We define $\gamma_i \eqdef \gamma_{i,i}$, and from~(\ref{eq:nprime}) it follows that $\gamma_i \eqdef \gamma_{i,i} \in \mathbb{R}_+$. Replacing $\rho = \sigma_k$ in (\ref{eq:nprime}), we have that for all $k=0,\dots,3$,
\begin{equation}
\alpha'_k \sigma_k  = \alpha''_k\sum_{i }\gamma_{i}\sigma_i \sigma_k \sigma_i^\dagger + \alpha''_k\sum_{i,j, i \neq j}\gamma_{i,j}\sigma_i \sigma_k \sigma_j^\dagger.
\end{equation}
The left hand side of the above equation has only $\sigma_k$ term, so only $\sigma_k$ on the right hand side should survive as Pauli matrices form an orthogonal basis. It follows that either $\alpha'_k = \alpha''_k = 0$, or the terms of the second sum in the right hand side of the above equation necessarily cancel each other. In both cases, we have that
 \begin{align}
 \alpha'_k \sigma_k   = &\ \alpha''_k\sum_{i }\gamma_{i}\sigma_i \sigma_k \sigma_i^\dagger  = \alpha''_k \lambda_k \sigma_k, \\
 \mbox{and thus, }  \qquad \alpha'_k  = &\ \lambda_k\alpha''_k, \label{eq:lambda}\\
 \mbox{where, }  \qquad \lambda_0 \eqdef &\ \gamma_0 + \gamma_1 + \gamma_2 + \gamma_3 \label{eq:lambda_0}\\
 \lambda_1 \eqdef &\ \gamma_0 + \gamma_1 - \gamma_2 - \gamma_3 \label{eq:lambda_1}\\
 \lambda_2 \eqdef &\ \gamma_0 - \gamma_1 + \gamma_2 - \gamma_3 \label{eq:lambda_2}\\
 \lambda_3 \eqdef &\ \gamma_0 - \gamma_1 - \gamma_2 + \gamma_3 \label{eq:lambda_3}
 \end{align}
We also note that $\lambda_0 = 1$, since $\alpha'_0 = \alpha''_0 = 1$. We further rewrite (\ref{eq:lambda}) as
\begin{equation}\label{eq:lambda-rewrite}
\bm{\alpha}' = \Lambda \bm{\alpha}''.
\end{equation}
where $\Lambda = \mbox{diag}(\lambda_0, \lambda_1, \lambda_2, \lambda_3)$ is the square diagonal matrix with $\lambda_i$'s on the main diagonal. 
Plugging (\ref{eq:a-matrix}) into (\ref{eq:lambda-rewrite}), and using $A^2 = 4I$, we get
\begin{equation}
\bm{p}' = \frac{1}{4}A \Lambda A \bm{p}''  = \Gamma \bm{p}'', \label{eq:Gamma_matrix} 
\end{equation}
$$\text{ where } \Gamma \eqdef \frac{1}{4}A \Lambda A = 
\left( \begin{array}{rrrr}
\gamma_0  & \gamma_1  & \gamma_2  & \gamma_3  \\
\gamma_1  & \gamma_0  & \gamma_3  & \gamma_2  \\
\gamma_2  & \gamma_3  & \gamma_0  & \gamma_1  \\
\gamma_3  & \gamma_2  & \gamma_1  & \gamma_0  
\end{array}\right) $$

 We now come back to our assumption, namely $\alpha'_k \neq \alpha''_k$, for some $k=1,2,3$. Without loss of generality, we may assume that $\alpha'_1 \neq \alpha''_1$. Since $|\alpha'_1| = |\alpha''_1|$ and $\alpha'_1 = \lambda_1 \alpha''_1$, it follows that $\lambda_1 = -1$. Then, using~(\ref{eq:lambda_0}) and (\ref{eq:lambda_1}), we have that $2(\gamma_0 + \gamma_1) = \lambda_0 + \lambda_1 = 0$, which implies 
\begin{equation}\label{eq:zero-gammas}
\gamma_0 = \gamma_1 = 0,
\end{equation} since they are non-negative.  
We proceed now with several sub-cases:

\smallskip
\hangindent=\parindent
\noindent {\em Case 2.1:} either $\alpha'_2 \neq \alpha''_2$ or $\alpha'_3 \neq \alpha''_3$. Similarly to the derivation of (\ref{eq:zero-gammas}), we get either $\gamma_2 = 0$ (in which case  $\gamma_3 = 1$) or $\gamma_3 = 0$ (in which case  $\gamma_2 = 1$). In either case $\Lambda$ is a permutation matrix, which implies that $\left({\cal N}'\right)^\# \equiv \left({\cal N}''\right)^\#$, as desired.

\smallskip
\hangindent=\parindent
\noindent {\em Case 2.2:} $\alpha'_2 = \alpha''_2$ and $\alpha'_3 = \alpha''_3$, and either $\alpha'_2 = \alpha''_2 \neq 0$ or $\alpha'_3 = \alpha''_3 \neq 0$. Let us assume that $\alpha'_2 = \alpha''_2 \neq 0$. In this case, using~(\ref{eq:lambda}), we have that $\lambda_2 = 1$, and from~(\ref{eq:lambda_2}) it follows that $\gamma_2 - \gamma_3 = 1$. This implies  $\gamma_2 = 1$ and  $\gamma_3 = 0$, therefore $\Lambda$ is a permutation matrix, and thus $\left({\cal N}'\right)^\# \equiv \left({\cal N}''\right)^\#$, as desired.

\smallskip
\hangindent=\parindent
\noindent {\em Case 2.3:} $\alpha'_2 = \alpha''_2 = 0$ and $\alpha'_3 = \alpha''_3 = 0$. Using $\alpha'_k = 2(p'_0 + p'_k) - 1, \forall k\neq 0$, we get $p'_2 = p'_3 = \frac{1}{2} - p'_0$, and similarly $p''_2 = p''_3 = \frac{1}{2} - p''_0$. Moreover, using (\ref{eq:Gamma_matrix}) and the fact that $\gamma_2 + \gamma_3 = 1$, we get $p'_0 = p'_1 =  p''_2 = p''_3$ and $p'_2 = p'_3 =  p''_0 = p''_1$. This implies that $\left({\cal N}'\right)^\# \equiv \left({\cal N}''\right)^\#$, as desired.

\smallskip \noindent This concludes the second case, and finishes the proof.
\hfill $\qed$

\section{Proof of Proposition~\ref{prop:cq_equiv}}\label{sec:proof_cq_equiv}

Using the notation from Section~\ref{subsec:classical_combining_splitting}, we shall identify $\left(\bar{P}_1, \times\right) \cong \left(\{0,1,2,3\}, \oplus\right)$, where $\sigma_i~\cong~i$,  $\forall i =0,\dots,3$, and thus assume that the classical channel ${\cal N}^\#$ -- associated with a Pauli channel ${\cal N}(\rho) = \sum_{i=0}^3 p_i \sigma_i \rho \sigma_i^{\dagger}$ -- has input and output alphabet $\{0,1,2,3\}$, with transition probabilities defined by ${\cal N}^\#(i \mid j) = p_{i\oplus j}$. Moreover, the automorphism $\Gamma = \Gamma(C)$ induced by the conjugate action of a two-qubit Clifford  unitary $C$ on $\bar{P}_1 \times \bar{P}_1$, is identified to a linear permutation $\Gamma: \{0,1,2,3\}^2 \rightarrow  \{0,1,2,3\}^2$, such that $C\sigma_{i,j}C^\dagger = \sigma_{\Gamma(i,j)}$.  We shall also write $\Gamma = (\Gamma_1, \Gamma_2)$, with $\Gamma_i : \{0,1,2,3\}^2 \rightarrow \{0,1,2,3\}$,  $i=1,2$.

\medskip It can be easily seen that it is enough to prove the statement of Proposition~\ref{prop:cq_equiv} for the case when ${\cal N}$ and ${\cal M}$ are Pauli channels.  Let 
${\cal N}(\rho) = \sum_{i=0}^3 p_i \sigma_i \rho \sigma_i^{\dagger}$ and ${\cal M}(\rho) = \sum_{j=0}^3 q_j \sigma_j \rho \sigma_j^{\dagger}$.

\medskip We start by proving $(i)$.
\begin{align*}
({\cal N}\boxast{\cal M})(\rho_U) &=  ({\cal N}\otimes{\cal M})\left( C \left( \rho_U \otimes \frac{I_V}{2} \right) C^\dagger \right) \\
&= \sum_{i,j} p_i q_j \sigma_{i,j} C \left( \rho_U \otimes \frac{I_V}{2} \right) C^\dagger \sigma_{i,j}^\dagger \\
&= \sum_{i,j}r_{i,j}  C \sigma_{\Gamma^{-1}(i,j)} \left( \rho_U \otimes \frac{I_V}{2} \right) \sigma_{\Gamma^{-1}(i,j)}^\dagger C^\dagger, \mbox{ where } \ r_{i,j} \eqdef p_i q_j\\
&=  C \left( \sum_{i,j}r_{\Gamma(i,j)}  \sigma_{i,j} \left( \rho_U \otimes \frac{I_V}{2} \right) \sigma_{i,j}^\dagger \right) C^\dagger \\ 
&=  C \left( \sum_{i,j}r_{\Gamma(i,j)}  \sigma_{i,j} \left( \rho_U \otimes \frac{I_V}{2} \right) \sigma_{i,j}^\dagger \right) C^\dagger\\
&=  C \left( \sum_{i,j}r_{\Gamma(i,j)}  \sigma_{i} \rho_U \sigma_{i}^\dagger \otimes \frac{I_V}{2} \right)   C^\dagger \\
&=  C \left( \sum_{i} s_i\,  \sigma_{i} \rho_U \sigma_{i}^\dagger \otimes \frac{I_V}{2} \right)   C^\dagger, \mbox{ where } s_i \eqdef \sum_j r_{\Gamma(i,j)} 
\end{align*}
where the fourth equality follows from the variable change $(i,j) \mapsto \Gamma(i,j)$. 
Omitting the conjugate action of the unitary $C$ and discarding the $V$ system, we may further identify:
\begin{equation*}
({\cal N}\boxast{\cal M})(\rho_U) = \sum_{i} s_i  \sigma_{i} \rho_U \sigma_{i}^\dagger.
\end{equation*}
Hence, the associated classical channel $({\cal N}\boxast{\cal M})^\#$ is defined by the probability vector $\mathbf{s} = (s_0,s_1,s_2,s_3)$, meaning that 
\begin{equation}\label{eq:quantum_bad_channel}
({\cal N}\boxast{\cal M})^\#(i \mid j) = s_{i\oplus j}. 
\end{equation}

On the other hand, we have:
\begin{align*}
({\cal N}^\#\boxast{\cal M}^\#)(a, b \mid u) 
&=  \frac{1}{4} \sum_v {\cal N}^\#(a \mid {\Gamma_1(u,v)}) {\cal M}^\#(b \mid {\Gamma_2(u,v)}) \\
&=  \frac{1}{4} \sum_v p_{a \oplus \Gamma_1(u,v)} q_{b \oplus \Gamma_2(u,v)}.
\end{align*}
Applying $\Gamma^{-1}$ on the channel output, we may identify ${\cal N}^\#\boxast{\cal M}^\#$ to a channel with output
$(a',b') = \Gamma^{-1}(a,b)$,  and transition probabilities given by:
\begin{align*}
({\cal N}^\#\boxast{\cal M}^\#)(a', b' \mid u) 
&= \frac{1}{4} \sum_v p_{\Gamma_1(a',b') \oplus \Gamma_1(u,v)} q_{\Gamma_2(a',b') \oplus \Gamma_2(u,v)} \\
&= \frac{1}{4} \sum_v p_{\Gamma_1((a',b') \oplus (u,v))} q_{\Gamma_2((a',b') \oplus (u,v))} \\
&= \frac{1}{4} \sum_v p_{\Gamma_1(a'\oplus u, b'\oplus v)} q_{\Gamma_2(a' \oplus u, b' \oplus v)} \\
&= \frac{1}{4}  \sum_v p_{\Gamma_1(a'\oplus u, v)} q_{\Gamma_2(a' \oplus u, v)} \\
&= \frac{1}{4}  \sum_v r_{\Gamma(a'\oplus u, v)} \\
&= \frac{1}{4}  s_{a'\oplus u}.
\end{align*}
We can then discard the $b'$ output, since the channel transition probabilities do not depend on it, which gives a channel defined by transition probabilities:
\begin{equation}\label{eq:classical_bad_channel}
({\cal N}^\#\boxast{\cal M}^\#)(a' \mid u) = s_{a'\oplus u}.
\end{equation} 
Finally, using (\ref{eq:quantum_bad_channel}) and~(\ref{eq:classical_bad_channel}), and noticing that omitting the conjugate action of the unitary $C$ and discarding the $V$ system in the derivation of (\ref{eq:quantum_bad_channel}) is equivalent to applying $\Gamma^{-1}$ on the channel output and discarding the $b'$  output in the derivation of (\ref{eq:classical_bad_channel}), we conclude that $({\cal N}\boxast{\cal M})^\# \equiv {\cal N}^\#\boxast{\cal M}^\#$

\medskip We prove now the $(ii)$ statement. Similar to the derivations used for $(i)$, we get:
\begin{align}
({\cal N}\varoast{\cal M})(\rho_V) &=   C \left( \sum_{i,j}r_{\Gamma(i,j)}  \sigma_{i,j} \left( \Phi_{U'U} \otimes \rho_V \right) \sigma_{i,j}^\dagger \right) C^\dagger \nonumber\\
&=  C \left( \sum_{i,j}r_{\Gamma(i,j)}  \left( (I_{U'}\otimes \sigma_i)(\Phi_{U'U})(I_{U'}\otimes \sigma_i^\dagger)\right) \otimes (\sigma_j \rho_V \sigma_j^\dagger) \right) C^\dagger \label{eq:xxx}
\end{align}
Omitting the conjugate action of the unitary $C$, and expressing $(I_{U'}\otimes \sigma_i)(\Phi_{U'U})(I_{U'}\otimes \sigma_i^\dagger)$ in the Bell basis, 
$\{\ket{i}\}_{i=0,\dots,3} \eqdef \{ \frac{\ket{00}+\ket{11}}{\sqrt{2}}, \frac{\ket{01}+\ket{10}}{\sqrt{2}}, \frac{\ket{01}-\ket{10}}{\sqrt{2}}, \frac{\ket{00}-\ket{11}}{\sqrt{2}} \}$, we get:
\begin{equation*}
({\cal N}\varoast{\cal M})(\rho_V) =  \sum_{i,j}r_{\Gamma(i,j)}  \ket{i}\bra{i} \otimes (\sigma_j \rho_V \sigma_j^\dagger).
\end{equation*}
Let $\lambda_i \eqdef \sum_j r_{\Gamma(i,j)}$ and $s_{i,j} \eqdef r_{\Gamma(i,j)} / \lambda_i$ (with $s_{i,j} \eqdef 0$ if $\lambda_i=0$). Denoting by ${\cal S}_i$ the Pauli channel defined by ${\cal S}(\rho)_i = \sum_j s_{i,j} \sigma_j \rho_V \sigma_j^\dagger$, we may rewrite:
\begin{equation*}
({\cal N}\varoast{\cal M})(\rho_V) =  \sum_{i} \lambda_{i}  \ket{i}\bra{i} \otimes {\cal S}_i(\rho_V). 
\end{equation*}
Hence, $({\cal N}\varoast{\cal M})^\#$ is the mixture of the channels ${\cal S}_i^\#$, with ${\cal S}_i^\#$ being used with probability $\lambda_i$, whose transition probabilities are given by:
\begin{equation}\label{eq:quantum_good_channel}
({\cal N}\varoast{\cal M})^\#(i, j\mid k) = \lambda_i s_{i, j \oplus k} = r_{\Gamma(i, j\oplus k)}.
\end{equation}

On the other hand, we have:
\begin{align*}
({\cal N}^\#\varoast{\cal M}^\#)(a, b, u \mid v) 
&=  \frac{1}{4}  {\cal N}^\#(a \mid {\Gamma_1(u,v)}) {\cal M}^\#(b \mid {\Gamma_2(u,v)})\\
&=  \frac{1}{4}  p_{a \oplus \Gamma_1(u,v)} q_{b \oplus \Gamma_2(u,v)}.
\end{align*}
We apply $\Gamma^{-1}$ on the $(a,b)$ output of the channel, which is equivalent to omitting the conjugate action of the unitary $C$ in (\ref{eq:xxx}), and then identify ${\cal N}^\#\varoast{\cal M}^\#$ to a channel with output
$(a',b', u)$, where $(a',b') = \Gamma^{-1}(a,b)$,  and transition probabilities:
\begin{align*}
({\cal N}^\#\varoast{\cal M}^\#)(a', b', u \mid v) 
&= \frac{1}{4}  p_{\Gamma_1(a',b') \oplus \Gamma_1(u,v)} q_{\Gamma_2(a',b') \oplus \Gamma_2(u,v)} \\
&= \frac{1}{4}  p_{\Gamma_1(a'\oplus u, b'\oplus v)} q_{\Gamma_2(a' \oplus u, b' \oplus v)} \\
&=  \frac{1}{4}  r_{\Gamma(a'\oplus u, b'\oplus v)}.
\end{align*}
We further perform a change of variable, replacing $(a', u)$ by $(a'\oplus u, u)$, which makes the above transition probability independent of $u$. We may then discard the $u$ output, and thus identify  ${\cal N}^\#\varoast{\cal M}^\#$ to a channel with output
$(a',b')$  and transition probabilities:
\begin{equation}\label{eq:classical_good_channel}
({\cal N}^\#\varoast{\cal M}^\#)(a', b' \mid v) =   r_{\Gamma(a', b'\oplus v)} .
\end{equation}
Finally, using (\ref{eq:quantum_good_channel}) and~(\ref{eq:classical_good_channel}), we conclude that $({\cal N}\varoast{\cal M})^\# \equiv {\cal N}^\#\varoast{\cal M}^\#$.
\ \hfill $\qed$

\section{Proof of Lemma~\ref{lem:z-good-classical-channel}} \label{apnd:pol9}

We prove first the following lemma.
\begin{lemma}\label{lemma:Zd_good_channels}
For any classical channels $N, M$, with input alphabet $\bar{P}_1 \cong \left( \{0,1,2,3\}, \oplus \right)$, and any linear permutation $\Gamma  = (A, B) : \bar{P}_1\times\bar{P}_1 \rightarrow \bar{P}_1\times\bar{P}_1$, the following equality holds for any $d \in \bar{P}_1$:
\begin{align*}
Z_d(N \varoast_\Gamma M) &= Z_{A(0,d)}(N) Z_{B(0,d)}(M).
\end{align*}
\end{lemma}

\begin{proof}
According to Definition~\ref{def:bhattacharyya}, for the channel $ N \varoast_{\Gamma} M$, we have:
\begin{multline*}
Z\left((N \varoast_{\Gamma} M)_{v,\, v'}\right) = \sum_{u, y_1, y_2}  \sqrt{(N \varoast_{\Gamma} M)(y_1, y_2, u \mid v) \; (N \varoast_{\Gamma} M)(y_1, y_2, u \mid v')}  \\
\begin{aligned}
&= \frac{1}{4} \sum_{u, y_1, y_2}  \sqrt{N(y_1 \mid A(u,v)) \; M(y_2 \mid B(u,v)) \; N(y_1 \mid A(u,v'))\; M(y_2 \mid B(u,v'))}  \\
&= \frac{1}{4} \sum_{u, y_1, y_2}  \sqrt{N(y_1 \mid A(u,v))\; N(y_1 \mid A(u,v')) \; M(y_2 \mid B(u,v)) \; M(y_2 \mid B(u,v'))}  \\
&= \frac{1}{4} \sum_{u} \, Z\left(N_{A(u,v),\, A(u,v')}\right) \, Z\left(M_{B(u,v),\, B(u,v')}\right).
\end{aligned}
\end{multline*}
Therefore,
\begin{align*}
Z_d\left(N \varoast_{\Gamma} M\right) &= \frac{1}{4}\sum_{v} Z\left((N \varoast_{\Gamma} M)_{v,\, v\oplus d}\right) \\
&= \frac{1}{16} \sum_{u,v} Z\left(N_{A(u,v),\, A(u,v\oplus d)}\right) Z\left(M_{B(u,v),\, B(u,v\oplus d)}\right) \\
&= \frac{1}{16} \sum_{u,v} Z\left(N_{A(u,v),\, A(u,v)\oplus A(0, d)}\right) Z\left(M_{B(u,v),\, B(u,v)\oplus B(0, d)}\right) \\ 
&=  \frac{1}{16} \sum_{a} Z\left(N_{a,\, a\oplus A(0, d)}\right)\sum_{b}  Z\left(M_{b,\, b\oplus B(0, d)}\right)  \\ 
&= Z_{A(0,d)}(N) Z_{B(0,d)}(M),
\end{align*}
where the third equality follows from the linearity of the permutation $\Gamma = (A,B)$, and the fourth equality follows from the change of basis for the summation from $(u,v)$ to $(a,b) := (A(u,v), B(u,v))$.
\end{proof}

\begin{figure*}[!t]
    \centering
	\resizebox{\linewidth}{!}{%
    \begin{tikzpicture}
    \def\nodewidth{9mm}     
    \def\nodeheight{6.5mm}  
    \begin{scope}
    \draw
    (0, 0) node[draw, minimum width=\nodewidth, minimum height=\nodeheight](C1){$I$}
    (C1) ++(0, -1) node[draw, minimum width=\nodewidth, minimum height=\nodeheight](C2){$I$}
    (C1) ++(-1.25, 0) node[not](C3){}
    (C2) ++(-1.25, 0) node[phase](C4){}
    ;
    \draw
    (C3) to  (C1)
    (C4) to  (C2)
    (C3) to  (C4)
    (C3) to  ++(-0.8, 0) node[left]{$u_1, u_2$}
    (C4) to  ++(-0.8, 0) node[left]{$v_1, v_2$}
    (C2) to  ++(+1, 0) node[right]{$u_1 \oplus v_1, u_1 \oplus v_2$}
    (C1) to  ++(+1, 0) node[right]{$u_1, u_2\oplus v_1\oplus v_2$}
    (C1) ++(-3.7, -.5) node[](){$\Gamma_{1,1}:$}
    ;
    \draw
    (0, -2.5) node[draw, minimum width=\nodewidth, minimum height=\nodeheight](C1){$\sqrt{Z}$}
    (C1) ++(0, -1) node[draw, minimum width=\nodewidth, minimum height=\nodeheight](C2){$I$}
    (C1) ++(-1.25, 0) node[not](C3){}
    (C2) ++(-1.25, 0) node[phase](C4){}
    ;
    \draw
    (C3) to  (C1)
    (C4) to  (C2)
    (C3) to  (C4)
    (C3) to  ++(-0.8, 0) node[left]{$u_1, u_2$}
    (C4) to  ++(-0.8, 0) node[left]{$v_1, v_2$}
    (C2) to ++(+1, 0) node[right]{$u_1 \oplus v_1, u_1 \oplus v_2$}
    (C1) to ++(+1, 0) node[right]{$u_2\oplus v_1\oplus v_2, u_1$}
    (C1) ++(-3.7, -.5) node[](){$\Gamma_{2,1}:$}
    ;
    \draw
    (0, -5.0) node[draw, minimum width=\nodewidth, minimum height=\nodeheight](C1){$\sqrt{Y}$}
    (C1) ++(0, -1) node[draw, minimum width=\nodewidth, minimum height=\nodeheight](C2){$I$}
    (C1) ++(-1.25, 0) node[not](C3){}
    (C2) ++(-1.25, 0) node[phase](C4){}
    ;
    \draw
    (C3) to  (C1)
    (C4) to  (C2)
    (C3) to  (C4)
    (C3) to  ++(-0.8, 0) node[left]{$u_1, u_2$}
    (C4) to  ++(-0.8, 0) node[left]{$v_1, v_2$}
    (C2) to ++(+1, 0)  node[right]{$u_1 \oplus v_1, u_1 \oplus v_2$}
    (C1) to ++(+1, 0)  node[right]{$\begin{array}{r}u_1 \oplus u_2\oplus v_1\oplus v_2, \hspace*{5mm}\, \\u_2\oplus v_1\oplus v_2\end{array}$}
    (C1) ++(-3.7, -0.5) node[](){$\Gamma_{3,1}:$}
    ;
    \end{scope}
    \begin{scope}[xshift =10cm]
    \draw
    (0, 0) node[draw, minimum width=\nodewidth, minimum height=\nodeheight](C1){$I$}
    (C1) ++(0, -1) node[draw, minimum width=\nodewidth, minimum height=\nodeheight](C2){$\sqrt{X}$}
    (C1) ++(-1.25, 0) node[not](C3){}
    (C2) ++(-1.25, 0) node[phase](C4){}
    ;
    \draw
    (C3) to  (C1)
    (C4) to  (C2)
    (C3) to  (C4)
    (C3) to  ++(-0.8, 0) node[left]{$u_1, u_2$}
    (C4) to  ++(-0.8, 0) node[left]{$v_1, v_2$}
    (C2) to ++(+1, 0) node[right]{$u_1 \oplus v_1, v_1 \oplus v_2$}
    (C1) to ++(+1, 0) node[right]{$u_1, u_2\oplus v_1\oplus v_2$}
    (C1) ++(-3.7, -0.5) node[](){$\Gamma_{1,2}:$}
    ;
    \draw
    (0, -2.5) node[draw, minimum width=\nodewidth, minimum height=\nodeheight](C1){$\sqrt{Z}$}
    (C1) ++(0, -1) node[draw, minimum width=\nodewidth, minimum height=\nodeheight](C2){$\sqrt{X}$}
    (C1) ++(-1.25, 0) node[not](C3){}
    (C2) ++(-1.25, 0) node[phase](C4){}
    ;
    \draw
    (C3) to  (C1)
    (C4) to  (C2)
    (C3) to  (C4)
    (C3) to  ++(-0.8, 0) node[left]{$u_1, u_2$}
    (C4) to  ++(-0.8, 0) node[left]{$v_1, v_2$}
    (C2) to ++(+1, 0) node[right]{$u_1 \oplus v_1, v_1 \oplus v_2$}
    (C1) to ++(+1, 0)  node[right]{$u_2\oplus v_1\oplus v_2, u_1$}
    (C1) ++(-3.7, -0.5) node[](){$\Gamma_{2,2}:$}
    ;
     \draw
    (0, -5) node[draw, minimum width=\nodewidth, minimum height=\nodeheight](C1){$\sqrt{Y}$}
    (C1) ++(0, -1) node[draw, minimum width=\nodewidth, minimum height=\nodeheight](C2){$\sqrt{X}$}
    (C1) ++(-1.25, 0) node[not](C3){}
    (C2) ++(-1.25, 0) node[phase](C4){}
    ;
    \draw
    (C3) to  (C1)
    (C4) to  (C2)
    (C3) to  (C4)
    (C3) to  ++(-0.8, 0) node[left]{$u_1, u_2$}
    (C4) to  ++(-0.8, 0) node[left]{$v_1, v_2$}
    (C2) to ++(+1, 0)  node[right]{$u_1 \oplus v_1, v_1 \oplus v_2$}
    (C1) to ++(+1, 0)  node[right]{$\begin{array}{r}u_1 \oplus u_2\oplus v_1\oplus v_2, \hspace*{5mm}\, \\u_2\oplus v_1\oplus v_2\end{array}$}
    (C1) ++(-3.7, -.5) node[](){$\Gamma_{3,2}:$}
    ;
    \end{scope}
    \begin{scope}[yshift =-7.5cm]
     \draw
    (0, 0) node[draw, minimum width=\nodewidth, minimum height=\nodeheight](C1){$I$}
    (C1) ++(0, -1) node[draw, minimum width=\nodewidth, minimum height=\nodeheight](C2){$\sqrt{Y}$}
    (C1) ++(-1.25, 0) node[not](C3){}
    (C2) ++(-1.25, 0) node[phase](C4){}
    ;
    \draw
    (C3) to  (C1)
    (C4) to  (C2)
    (C3) to  (C4)
    (C3) to  ++(-0.8, 0) node[left]{$u_1, u_2$}
    (C4) to  ++(-0.8, 0) node[left]{$v_1, v_2$}
    (C2) to ++(+1, 0) node[right]{$v_1\oplus v_2, u_1 \oplus v_2$}
    (C1) to ++(+1, 0) node[right]{$u_1, u_2\oplus v_1\oplus v_2$}
    (C1) ++(-3.7, -0.5) node[](){$\Gamma_{1,3}:$}
    ;
    \end{scope}
    \begin{scope}[xshift = 10cm, yshift =-5cm]
     \draw
    (0, -2.5) node[draw, minimum width=\nodewidth, minimum height=\nodeheight](C1){$\sqrt{Z}$}
    (C1) ++(0, -1) node[draw, minimum width=\nodewidth, minimum height=\nodeheight](C2){$\sqrt{Y}$}
    (C1) ++(-1.25, 0) node[not](C3){}
    (C2) ++(-1.25, 0) node[phase](C4){}
    ;
    \draw
    (C3) to  (C1)
    (C4) to  (C2)
    (C3) to  (C4)
    (C3) to  ++(-0.8, 0) node[left]{$u_1, u_2$}
    (C4) to  ++(-0.8, 0) node[left]{$v_1, v_2$}
    (C2) to ++(+1, 0)  node[right]{$v_1\oplus v_2, u_1 \oplus v_2$}
    (C1) to ++(+1, 0)  node[right]{$u_2\oplus v_1\oplus v_2, u_1$}
    (C1) ++(-3.7, -.5) node[](){$\Gamma_{2,3}:$}
    ;
    \end{scope}
    \begin{scope}[xshift = 5cm, yshift =-5cm]
    \draw
    (0, -5) node[draw, minimum width=\nodewidth, minimum height=\nodeheight](C1){$\sqrt{Y}$}
    (C1) ++(0, -1) node[draw, minimum width=\nodewidth, minimum height=\nodeheight](C2){$\sqrt{Y}$}
    (C1) ++(-1.25, 0) node[not](C3){}
    (C2) ++(-1.25, 0) node[phase](C4){}
    ;
    \draw
    (C3) to  (C1)
    (C4) to  (C2)
    (C3) to  (C4)
    (C3) to  ++(-0.8, 0) node[left]{$u_1, u_2$}
    (C4) to  ++(-0.8, 0) node[left]{$v_1, v_2$}
    (C2) to ++(+1, 0)  node[right]{$v_1\oplus v_2, u_1 \oplus v_2$}
    (C1) to ++(+1, 0)  node[right]{$u_1 \oplus u_2\oplus v_1\oplus v_2, u_2\oplus v_1\oplus v_2$}
    (C1) ++(-3.7, -.5) node[](){$\Gamma_{3,3}:$}
    ;
    \end{scope}
    \end{tikzpicture}
  	}
    \caption{Elements of the set $\Gamma(\mathcal{L})$}
    \label{fig:perm_left_set}
\end{figure*}
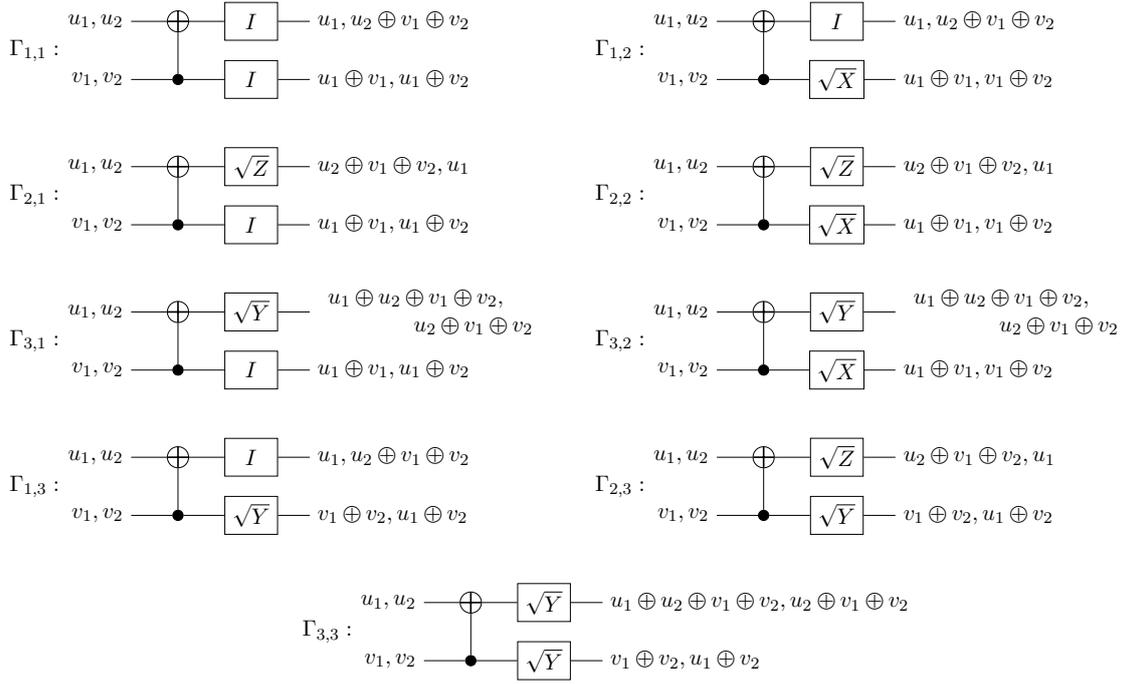

Throughout the remaining of this section, we shall denote by $u:=[u_1, u_2]$ the binary representation of a given $u\in \bar{P}_1 \cong \{0,1,2,3\}$, where $u_1, u_2 \in \{0,1\}$ and $u_2$ is the least significant bit.

\begin{lemma}\label{lemma:Ai_Bj_permutations}
Let $\Gamma_{i,j} := \Gamma(L_{i,j}) : \bar{P}_1\times\bar{P}_1 \rightarrow \bar{P}_1\times\bar{P}_1$ be the permutation defined by the conjugate action of $L_{i,j} \in {\cal L}$, where ${\cal L}$ is the set of two-qubit Clifford gates defined in Section~\ref{sec:polarization_9_cliffords} (Fig.~\ref{fig:cliff_generators}). Then $\Gamma_{i,j} = (A_i, B_j), \forall 1 \leq i,j \leq 3$, with  $A_i, B_j : \bar{P}_1\times\bar{P}_1 \rightarrow \bar{P}_1$ given by:
\begin{align*} 
 A_1(u,v) &= [u_1, u_2\oplus v_1\oplus v_2],   & B_1(u,v) &= [u_1 \oplus v_1, u_1 \oplus v_2] \nonumber \\
 A_2(u,v) &= [u_2\oplus v_1\oplus v_2, u_1 ],  & B_2(u,v) &= [u_1 \oplus v_1, v_1 \oplus v_2] \nonumber \\
 A_3(u,v) &= [u_1 \oplus u_2\oplus v_1\oplus v_2, u_2\oplus v_1\oplus v_2] & B_3(u,v) &= [v_1\oplus v_2, u_1 \oplus v_2] \nonumber
\end{align*}
where $u$ and $v$ inputs are represented in binary form, $u:=[u_1, u_2]$ and $v := [v_1, v_2]$, with $u_1, u_2, v_1, v_2 \in\{0,1\}$ ($\Gamma_{i,j}$ permutations are also depicted in Fig.~\ref{fig:perm_left_set}).
\end{lemma} 
\begin{proof}
Recall from Section~\ref{sec:polarization_9_cliffords}, that $L_{i,j} = (C' \otimes C'')\text{\sc cnot}_{21}$, where $C' \in  \{I, \sqrt{Z}, \sqrt{Y}\}$, and $C'' \in \{I, \sqrt{X}, \sqrt{Y}\}$. Recall also that by identifying  $\bar{P}_1 \cong \{0,1,2,3\}$, we have $I = \sigma_0 \cong 0$,  $X = \sigma_1 \cong 1$, $Y = \sigma_2 \cong 2$, $Z = \sigma_3 \cong 3$. The conjugate action of $\sqrt{X}$ on $\bar{P}_1$, fixes $I$ and $X$, and permutes $Y$ and $Z$. Hence, the corresponding permutation on $\bar{P}_1 \cong \{0,1,2,3\}$, can be written as $(0, 1, 3, 2)$. Similarly, the conjugate action of $\sqrt{Y}$ and $\sqrt{Z}$ induces the permutations $(0, 3, 2, 1)$ and $(0, 2, 1, 3)$, respectively. Replacing $u\in \{0,1,2,3\}$ by its binary representation $[u_1, u_2]$, we may write:
\begin{equation}
\begin{aligned}
\sqrt{X}\, &: [u_1,u_2] \mapsto [u_1, u_1\oplus u_2] \\
\sqrt{Y}\, &: [u_1,u_2] \mapsto [u_1\oplus u_2, u_2] \\
\sqrt{Z}\, &: [u_1,u_2] \mapsto [u_2, u_1]
\end{aligned} \label{eq:sqrt_sigma_perm}
\end{equation}
Moreover, the permutation induced by the conjugate action of the $\text{\sc cnot}_{21}$ gate is the linear permutation on $\bar{P}_1 \times \bar{P}_1$ such that:
\begin{align}
\text{\sc cnot}_{21}\, : &\ (X, I) \mapsto (X, I),\ \     (I, X) \mapsto (X, X) \nonumber \\
  &\ (Z, I)\mapsto (Z, Z),\ \  (I, Z) \mapsto (I, Z) \nonumber\\
 \begin{split}
\Rightarrow\, \, \text{\sc cnot}_{21}\, : &\  \left([u_1,u_2], [v_1,v_2]\right) \mapsto \left( [u_1, u_2 \oplus v_1 \oplus v_2],  [u_1 \oplus v_1, u_1\oplus v_2] \right).
\end{split} \label{eq:cnot21_perm}
\end{align}
Finally, using (\ref{eq:sqrt_sigma_perm}) and (\ref{eq:cnot21_perm}), it can be easily verified that $\Gamma_{i,j} = (A_i, B_j), \forall 1 \leq i,j \leq 3$, with $A_i$ and $B_j$ as given in the lemma.
\end{proof}

\medskip \noindent{\bf Proof of Lemma~\ref{lem:z-good-classical-channel}.} 
To simplify notation, let $W \eqdef {\cal W}^\#$ be the classical counterpart of the CMP channel $\mathcal{W}$ from Lemma~\ref{lem:z-good-classical-channel}. Applying Lemma~\ref{lemma:Zd_good_channels} and Lemma~\ref{lemma:Ai_Bj_permutations}, we may express $Z_d(W\varoast_{\Gamma_{i,j}} W)$ as a function of $(Z_1(W), Z_2(W), Z_3(W))$, for any $\Gamma_{i,j} \in\Gamma(\mathcal{L})$ and any $d=1,2,3$ (recall that $Z_0(W) = 1$). The corresponding expressions are given in Table~\ref{table:Zd_good_channel}.

\begin{table}[!t]
\centering
\caption{$Z_d(W\varoast_{\Gamma_{i,j}} W)$ as a function of $(Z_1(W), Z_2(W), Z_3(W))$}
\label{table:Zd_good_channel}
$\begin{array}{@{\,}c|c|c|c@{\,}}
\hline
(i,j)  & Z_1(W\varoast_{\Gamma_{i,j}} W) & Z_2(W\varoast_{\Gamma_{i,j}} W) & Z_3(W\varoast_{\Gamma_{i,j}} W) \\
\hline \hline
 (1, 1)   &   Z_1(W)^2     & Z_1(W)Z_2(W)   &    Z_3(W)     \\
 (1, 2)   &   Z_1(W)^2     & Z_1(W)Z_3(W)   &    Z_2(W)     \\
 (1, 3)   & Z_1(W)Z_3(W)   & Z_1(W)Z_2(W)   &    Z_1(W)     \\
 (2, 1)   & Z_1(W)Z_2(W)   &   Z_2(W)^2     &    Z_3(W)     \\
 (2, 2)   & Z_1(W)Z_2(W)   & Z_2(W)Z_3(W)   &    Z_2(W)     \\
 (2, 3)   & Z_2(W)Z_3(W)   &   Z_2(W)^2     &    Z_1(W)     \\
 (3, 1)   & Z_1(W)Z_3(W)   & Z_2(W)Z_3(W)   &    Z_3(W)     \\
 (3, 2)   & Z_1(W)Z_3(W)   &   Z_3(W)^2     &    Z_2(W)     \\
 (3, 3)   &   Z_3(W)^2     & Z_2(W)Z_3(W)   &    Z_1(W)     \\
\hline
\end{array}$
\end{table}

\pagebreak \noindent Hence,
\begin{align*}
\sum_{\Gamma \in \Gamma(\cal L)} Z\left(W \varoast_{\Gamma} W\right) &= \frac{1}{3} \sum_{\Gamma \in \Gamma(\cal L)} \sum_{d=1}^3 Z_d\left(W \varoast_{\Gamma} W\right) \\
      &= \sum_{d=1}^3 Z_1(W) + \frac{2}{3} \left( \sum_{d=1}^3 Z_1(W) \right)^2 \\
      & = 3Z(W) + 6Z(W)^2,
\end{align*}
and therefore,
\begin{align*}
\mbE_{\Gamma \in \Gamma(\cal L)} Z\left(W \varoast_{\Gamma} W\right) &= \frac{1}{9} \sum_{\Gamma \in \Gamma({\cal L})} Z\left(W \varoast_{\Gamma} W\right) = \frac{1}{3}Z(W) + \frac{2}{3}Z(W)^2.
\end{align*}
The case $\Gamma\in\mathcal{R}$ can be derived in a similar way. 
Alternatively, one can directly verify that $\mbE_{\Gamma \in \Gamma(\cal L)} Z\left(W \varoast_{\Gamma} W\right) = \mbE_{\Gamma \in \Gamma(\cal R)} Z\left(W \varoast_{\Gamma} W\right)$, similarly to the proof of Lemma~\ref{lemma:swap_eq} in the quantum case.
\hfill $\qed$

\section{Proof of Lemma~\ref{lem:pol_3_clif}} \label{apnd:pol3}

Using Table~\ref{table:Zd_good_channel} from Appendix~\ref{apnd:pol9}, for $\Gamma \in \Gamma({\cal S}) = \{\Gamma_{1,3}, \Gamma_{2,2}, \Gamma_{3,1}\}$, we get
\begin{align*}
\mbE_{\Gamma \in \Gamma({\cal S})} Z\left(W \varoast_{\Gamma} W\right)  &=  \frac{1}{3} \sum_{\Gamma \in \Gamma({\cal S})} Z\left(W \varoast_{\Gamma} W\right) \\
&= \frac{1}{9} \sum_{\Gamma \in \Gamma({\cal S})} \sum_{1\leq d \leq 3} Z_d\left(W \varoast_{\Gamma} W\right) \\
&= \frac{1}{9} \sum_{1\leq d \leq 3} Z_d(W) + \frac{2}{9}  \sum_{1\leq d'\neq d'' \leq 3}  Z_{d'}(W)Z_{d''}(W)  \\  
&\leq  \frac{1}{3}Z(W) + \frac{2}{3} Z(W)^2,
\end{align*}
where, using $Z(W) = (Z_1(W) + Z_2(W) + Z_3(W))/3$, it is easily seen that the last inequality is equivalent to $Z_1(W)Z_2(W) + Z_1(W)Z_3(W) +  Z_2(W) Z_3(W) \leq Z_1(W)^2 + Z_2(W)^2 + Z_3(W)^2$, which follows from $Z_i(W)Z_j(W) \leq (Z_i(W)^2 + Z_j(W)^2)/2$.
\ \hfill $\qed$

 \section{Proof of Lemma~\ref{lem:Fast_Polarization}}\label{sec:proof_lemma_fast_polarization}
 
\noindent We prove first the following lemma.
\begin{lemma}\label{lemma:Zd_bad_channels}
For any classical channels $N, M$, with input alphabet $\bar{P}_1 \cong \left( \{0,1,2,3\}, \oplus \right)$, and any linear permutation $\Gamma  = (A, B) : \bar{P}_1\times\bar{P}_1 \rightarrow \bar{P}_1\times\bar{P}_1$, the following inequality holds for any $d \in \bar{P}_1$:
\begin{equation*}
Z_d(N \boxast_\Gamma M)  \leq \sum_{d' \in \bar{P}_1} Z_{A(d,d')}(N) Z_{B(d,d')}(M).
\end{equation*}
\end{lemma}

\begin{proof}
According to Definition~\ref{def:bhattacharyya}, for the channel $ N \boxast_{\Gamma} M$, we have:
%
\begin{multline*}
Z\left((N \boxast_{\Gamma} M)_{u,\, u'}\right) = \sum_{y_1, y_2} \sqrt{(N \boxast_{\Gamma} M)(y_1, y_2 \mid u) \; (N \boxast_{\Gamma} M)(y_1, y_2 \mid u')}  \\
\begin{aligned}
&=  \frac{1}{4} \sum_{y_1, y_2}  \sqrt{  \sum_v N(y_1 \mid A(u,v)) \,  M(y_2 \mid B(u,v)) \; \sum_{v'}  N(y_1 \mid A(u',v')) \, M(y_2 \mid B(u',v'))}  \\
&\leq \frac{1}{4} \sum_{v,v'} \sum_{y_1,y_2}  \sqrt{ N(y_1 \mid A(u,v)) \;  M(y_2 \mid B(u,v)) \; N(y_1 \mid A(u',v')) \; M(y_2 \mid B(u',v'))}   \\
&= \frac{1}{4} \sum_{v,v'} \sum_{y_1,y_2}  \sqrt{N(y_1 \mid A(u,v)) \; N(y_1 \mid A(u',v')) \; M(y_2 \mid B(u,v)) \; M(y_2 \mid B(u',v'))}  \\
&= \frac{1}{4} \sum_{v,v'}  Z\left(N_{A(u,v),\, A(u',v')}\right) Z\left(M_{B(u,v),\, B(u',v')}\right),
\end{aligned}
\end{multline*} 
where the inequality  above  follows from $\sqrt{\sum_v x_v} \leq \sum_v \sqrt{x_v}$.
Therefore,
\begin{align*}
Z_d\left(N \boxast_{\Gamma} M\right) &= \frac{1}{4}\sum_{u} Z\left((N \boxast_{\Gamma} M)_{u,\, u\oplus d}\right) \\
&\leq \frac{1}{16} \sum_{u,v,v'} Z\left(N_{A(u,v),\, A(u\oplus d,v')}\right) Z\left(M_{B(u,v),\, B(u\oplus d,v')}\right) \\
&= \frac{1}{16} \sum_{u,v,d' \text{\makebox[0mm][l]{\ ($d' := v\oplus v'$)}}} Z\left(N_{A(u,v),\, A(u\oplus d,v\oplus d')}\right) Z\left(M_{B(u,v),\, B(u\oplus d,v\oplus d')}\right) \\
&= \frac{1}{16} \sum_{u,v,d'}\! Z\left(N_{A(u,v),\, A(u,v)\oplus A(d, d')}\right) Z\left(M_{B(u,v),\, B(u,v)\oplus B(d, d')}\right)\\
&=  \frac{1}{16} \sum_{d'} \sum_{a} Z\left(N_{a,\, a\oplus A(d, d')}\right)\sum_{b}  Z\left(M_{b,\, b\oplus B(d, d')}\right) \\
&= \sum_{d'} Z_{A(d,d')}(N) Z_{B(d,d')}(M),
\end{align*}
where the third to last equality   follows from the linearity of the permutation $\Gamma = (A,B)$, and the second to last follows from the change of basis for the summation from $(u,v)$ to $(a,b) := (A(u,v), B(u,v))$.
\end{proof}

\medskip\noindent {\bf Proof of Lemma~\ref{lem:Fast_Polarization}.}
To simplify notation, let $W \eqdef {\cal W}^\#$ be the classical counterpart of the CMP channel $\mathcal{W}$ from Lemma~\ref{lem:Fast_Polarization}. Using Lemma~\ref{lemma:Zd_bad_channels}, we have 
$$Z_d(W \boxast_\Gamma W) \leq \sum_{d' \in \bar{P}_1} Z_{A(d,d')}(W) Z_{B(d,d')}(W).$$
 For $d\neq 0$, $A(d,d')$ and $B(d, d')$ cannot be simultaneously zero (recall that $Z_0(W) = 1$), and therefore we get $Z_{A(d,d')}(W) Z_{B(d,d')}(W) \leq \bar{Z}(W)$. Hence,  $Z_d(W \boxast_\Gamma W) \leq 4\bar{Z}(W), \forall d=1,2,3$, which implies $\bar{Z}(W \boxast_\Gamma W) \leq 4\bar{Z}(W)$, as desired. 
Finally, we have 
\begin{equation*}
Z(W \boxast_\Gamma W) \leq \bar{Z}(W \boxast_\Gamma W) \leq 4\bar{Z}(W) \leq 12 Z(W),
\end{equation*}
which proves the second inequality of the lemma. \hfill $\qed$

\section{Proof of Proposition~\ref{prop:fast_polarization}}
\label{sec:proof_prop_fast_polarization}

We proceed first with several lemmas. In the following, the notation $x = x(\cdot)$ means that the value of $x$ depends only on the list of variables $(\cdot)$ enclosed between parentheses. 

\begin{lemma} \label{lemma:dprime_good_channel}
$(i)$ For any permutation $\Gamma \in \Gamma(\mathcal{S})$, there exist $\delta_1 = \delta_1(\Gamma)$, $\delta_2 = \delta_2(\Gamma)$, $\delta_3 = \delta_3(\Gamma)$, such that $\{\delta_1, \delta_2, \delta_3\} = \{1, 2, 3\}$, and
\begin{align*}
Z_{3} (W \varoast_{\Gamma} W) &= Z_{\delta_3}(W)\\
Z_{2} (W \varoast_{\Gamma} W) &= Z_{\delta_3}(W) Z_{\delta_2}(W) \\
Z_{1} (W \varoast_{\Gamma} W) &= Z_{\delta_3}(W) Z_{\delta_1}(W),
\end{align*}
and the above equalities hold for any channel $W$. 

\smallskip \noindent $(ii)$ For any $d \in \{1,2,3\}$, there exists exactly one permutation $\Gamma \in \Gamma(\mathcal{S})$, such that $\delta_3(\Gamma) = d$.
\end{lemma}
\begin{proof}
Follows from Table~\ref{table:Zd_good_channel} in Appendix~\ref{apnd:pol9}, wherein $\Gamma(\mathcal{S}) = \{\Gamma_{1,3}, \Gamma_{2,2},\Gamma_{3,1}\}$. Precisely, we have 
$\delta_3(\Gamma_{1,3}) = 1, \delta_3(\Gamma_{2,2}) = 2, \delta_3(\Gamma_{3,1}) = 3$.
\end{proof}

\begin{lemma} \label{lemma:dprime_bad_channel}
There exist a constant $\kappa > 1$ and $\bm{\delta} = \bm{\delta}(W) \in\{1, 2, 3\}$, such that for any $\Gamma \in \Gamma(\mathcal{S})$ and any $d\in\{1,2,3\}$, the following equality holds
\begin{equation*}
Z_{d} (W \boxast_{\Gamma} W) \leq \kappa Z_{\bm{\delta}}(W).
\end{equation*}
\end{lemma}
\begin{proof}
Follows from Lemma~\ref{lem:Fast_Polarization}, for $\kappa = 4$ and $\displaystyle \bm{\delta} = \bm{\delta}(W) := \argmax_{d=1,2,3} Z_{d}(W)$.
\end{proof}

We shall also use the following lemma (known as Hoeffding's inequality) providing an upper bound for the probability that the mean of $n$
independent random variables falls below its expected value mean by a positive number.
\begin{lemma}[{\cite[Theorem 1]{hoeffding1963}}] \label{lemma:mean_independent_rv}
 Let $X_1,X_2,\dots,X_n$ be independent random variables such that $0 \leq X_i\leq 1$, for any $i=1,\dots,n$. Let $\bar{X} := \frac{1}{n}\sum_{i=1}^n X_i$, and $\mu = \mbE(\bar{X})$.  Then, for any $0< t < \mu$,
\begin{equation*}
\Pr\left\{ \bar{X} \leq \mu - t \right\} \leq e^{-2nt^2}.
\end{equation*}
\end{lemma}

\medskip Now, let $\Gamma(\mathcal{S})^{\infty}$ be the infinite Cartesian product of countable many copies of $\Gamma(\mathcal{S})$. It is endowed with an infinite product probability measure \cite{hewitt2013real}, denoted by $P$, where the uniform probability measure is taken on each copy of $\Gamma(\mathcal{S})$. For our purposes,  an infinite sequence $\bm{\Gamma} \in\Gamma(\mathcal{S})^{\infty}$ should be written as $\bm{\Gamma} := \left\{ \Gamma, \Gamma_{i_1\cdots i_n} \mid n > 0, i_1 \cdots i_n \in\{0,1\}^n \right\}$ (this is always possible, since the set of indices is countable).
We further define a sequence of independent and identically distributed (i.i.d) Bernoulli random variables on $\Gamma(\mathcal{S})^{\infty}$, denoted $\Delta^{i_1\cdots i_n}$, $n \geq 0$, $i_1 \cdots i_n\in\{0,1\}^n$, 
\begin{equation*}
\Delta^{i_1\cdots i_n}(\bm{\Gamma}) := \mathbf{1}_{\{\delta_3(\Gamma_{i_1\cdots i_n}) \in \{1,2\}\}},
\end{equation*}
that is, $\Delta^{i_1\cdots i_n}(\bm{\Gamma})$ is equal to $1$, if $\delta_3(\Gamma_{i_1\cdots i_n}) \in\{1,2\}$, and equal to $0$, if $\delta_3(\Gamma_{i_1\cdots i_n}) = 3$. 
Note that $\Delta^{i_1\cdots i_n}(\bm{\Gamma})$ does actually only depend  on the $\Gamma_{i_1\cdots i_n}$ element of $\bm{\Gamma}$ (here, $n$ and $i_1\cdots i_n$ are fixed).
From Lemma~\ref{lemma:dprime_good_channel} $(ii)$, it follows that $\mbE(\Delta^{i_1\cdots i_n}) = 2/3$, $\forall n \geq 0$, $\forall i_1 \cdots i_n\in\{0,1\}^n$.

\medskip For $0 < \gamma < 2/3$ and $m > 0$, we define
\begin{align}
\Pi_m(\gamma) &\eqdef \left\{ \bm{\Gamma} \in \Gamma(\mathcal{S})^{\infty} \Bigm| \sum_{i_1\cdots i_{m-1}} \Delta^{i_1\cdots i_{m-1} 1}(\bm{\Gamma}) \geq   \left( \frac{2}{3}-\gamma\right) 2^{m-1} \right\}\label{eq:G:128}\\
\overline{\Pi}_m(\gamma)  &\eqdef \bigcap_{n \geq m} \Pi_n(\gamma) \label{eq:G:129}
\end{align}
 The sum in (\ref{eq:G:128})  comprises all the terms $\Delta^{i_1\cdots i_{m-1} i_m}(\bm{\Gamma})$, with $i_1\cdots i_{m-1} \in \{0,1\}^{m-1}$ and $i_m=1$ (here, $m$ is fixed). Thus, $\Pi_m(\gamma)$ is defined by requiring that at least a fraction of $(2/3-\gamma)$ of $\Delta^{i_1\cdots i_{m-1} i_m}$ variables are equal to $1$, where $i_m=1$.  In (\ref{eq:G:129}), the above condition must hold for any $n\geq m$.

\begin{lemma}
For any $0 < \gamma < 2/3$ and $m > 0$,
\begin{equation}\label{eq:Pi_m_lowerbound}
P\left( \overline{\Pi}_m(\gamma) \right) \geq 2 - \frac{1}{1 - e^{-\gamma^2 2^{m}}}.
\end{equation}
\end{lemma}
\begin{proof}
By Lemma~\ref{lemma:mean_independent_rv}, $P\left( \Pi_m(\gamma) \right) \geq 1 - e^{-\gamma^2 2^{m}}$. Therefore, we have
\begin{align*}
P\left( \overline{\Pi}_m(\gamma) \right) &\geq 1 - \sum_{n\geq m} e^{-\gamma^2 2^{n}} \\
 &= 1 - \sum_{n\geq 0} \left(e^{-\gamma^2 2^{m}}\right)^{2^n} \\
 &\geq 1 - \sum_{n\geq 1} \left(e^{-\gamma^2 2^{m}}\right)^{n} \\
 &= 1 - \left(\frac{1}{1 - e^{-\gamma^2 2^{m}}} - 1 \right) \\
 &= 2 - \frac{1}{1 - e^{-\gamma^2 2^{m}}}. 
\end{align*}
\end{proof}

Note that the right hand side term in (\ref{eq:Pi_m_lowerbound}) converges to $1$ as $m$ goes to infinity. Hence, for  $\varepsilon > 0$, we denote by $m(\gamma,\epsilon)$ the smallest $m$ value, such that $2 - \frac{1}{1 - e^{-\gamma^2 2^{m}}} \geq 1-\epsilon$. It follows that $P\left( \overline{\Pi}_{m(\gamma,\varepsilon)}(\gamma) \right) \geq 1 -\varepsilon$. 

\medskip In the following, we fix once for all some $\gamma$ value, such that $0 < \gamma < 2/3$. The value of $\gamma$ will not matter for any of what we do here, we only need $(2/3-\gamma)$ to be positive. We proceed now with the proof of Proposition~\ref{prop:fast_polarization}.

\medskip \noindent {\bf Proof of Proposition~\ref{prop:fast_polarization}.} 
Let $\Omega := \{0,1\}^\infty$ denote the set of  infinite binary sequences $\omega := (\omega_1\omega_2\cdots) \in \{0,1\}^\infty$. Hence, $\Omega$ can be endowed with an infinite product probability measure, by taking the uniform probability measure on each $\omega_n$ component. We denote this probability measure by $P$ (the notation is the same as for the probability measure on $\Gamma(\mathcal{S})^{\infty}$, but no confusion should arise, since the sample spaces are different). 

\medskip Let $\varepsilon > 0$ and fix any $\bm{\Gamma} \in \Gamma(\mathcal{S})^\infty_{\text{pol}} \cap \overline{\Pi}_{m(\gamma,\varepsilon)}(\gamma) $.  Given $\bm{\Gamma}$, the polarization process can be formally described as a random process on the probability space $\Omega$ \cite{arikan09}. Precisely, for any $\omega = (\omega_1\omega_2\cdots) \in \Omega$ and $n > 0$, we define 
\begin{align*}
Z^{[n]}(\omega) &:= Z\left(W^{(\omega_1\cdots \omega_n)}\right) \\
Z_{d}^{[n]}(\omega) &:= Z_d\left(W^{(\omega_1\cdots \omega_n)}\right), \forall d\in \{1,2,3\}
\end{align*}
Note that $W^{(\omega_1\cdots \omega_n)}$ is recursively defined as in (\ref{eq:c_polar_recursion}), through the implicit assumption of using the channel combining permutations in the given sequence $\bm{\Gamma}$. For $n=0$, we set $Z^{[0]}(\omega) := Z(W)$  and $Z_{d}^{[0]}(\omega) := Z_d(W)$.

\medskip \noindent For $\zeta > 0$ and $m \geq 0$, we define
\begin{equation*}
T_m(\zeta) \eqdef \left\{  \omega \in \Omega \mid Z_{d}^{[n]}(\omega) \leq \zeta, \forall d=1,2,3, \forall n \geq m  \right\}.
\end{equation*}

\noindent Hence, for $\omega \in T_m(\zeta)$, $d\in \{1,2,3\}$, and $n > m$, we may write
\begin{equation}\label{eq:zd_decomp}
Z_{d}^{[n]}(\omega) = \frac{Z_{d_{n}}^{[n]}(\omega)}{Z_{d_{n-1}}^{[n-1]}(\omega)} \frac{Z_{d_{n-1}}^{[n-1]}(\omega)}{Z_{d_{n-2}}^{[n-2]}(\omega)} \cdots
    \frac{Z_{d_{m+1}}^{[m+1]}(\omega)}{Z_{d_{m}}^{[m]}(\omega)}Z_{d_{m}}^{[m]}(\omega),
\end{equation}
where $d_n := d$, and $d_{n-1},\dots, d_m$ are defined as explained below. Recall that $Z_{d}^{[k]}(\omega) :=$\break $Z_d(W^{(\omega_1 \cdots \omega_k)})$, and for $k \in \{n,n-1,\dots,m+1\}$, we have
\begin{equation*}
W^{(\omega_1\cdots \omega_k)} = \left\{ \begin{array}{@{}ll@{}}
   W^{(\omega_1\cdots \omega_{k-1})} \boxast_{\Gamma_{\omega_1\cdots \omega_{k-1}}}\!  W^{(\omega_1\cdots \omega_{k-1})}, & \text{if } \omega_k\!=\!0 \\
   W^{(\omega_1\cdots \omega_{k-1})} \varoast_{\Gamma_{\omega_1\cdots \omega_{k-1}}}\! W^{(\omega_1\cdots \omega_{k-1})}, & \text{if } \omega_k\!=\!1
\end{array} \right.
\end{equation*}
Hence, if $\omega_k = 0$, we set $d_{k-1} := \bm{\delta}\left(W^{(\omega_1\cdots \omega_{k-1})}\right)$ from Lemma~\ref{lemma:dprime_bad_channel}, such that we have
\begin{equation}\label{eq:bad_ubound}
\frac{Z_{d_{k}}^{[k]}(\omega)}{Z_{d_{k-1}}^{[k-1]}(\omega)} \leq \kappa, \ \text{ if } \omega_k = 0.
\end{equation} 
If $\omega_k = 1$, we set $d_{k-1} := \delta_3\left(\Gamma_{\omega_1\cdots \omega_{k-1}}\right)$ from Lemma~\ref{lemma:dprime_good_channel}, such that we have
\begin{align}
\frac{Z_{d_{k}}^{[k]}(\omega)}{Z_{d_{k-1}}^{[k-1]}(\omega)} &= 1, \ \ \text{ if } \omega_k = 1 \text{ and } d_k = 3. \\
\frac{Z_{d_{k}}^{[k]}(\omega)}{Z_{d_{k-1}}^{[k-1]}(\omega)} &\leq \zeta, \ \ \text{ if } \omega_k = 1 \text{ and } d_k \in \{1,2\}. \label{eq:good_ubound}
\end{align}
Let $A_{m,n}(\omega) := \{ k \in \{m+1,\dots,n\} \mid \omega_k = 1 \}$,  and $B_{m,n}(\omega) := \{ k \in \{m+1,\dots,n\} \mid \omega_k = 1\break \text{ and } d_k \in \{1,2\} \}$. 
Using (\ref{eq:zd_decomp}), (\ref{eq:bad_ubound})--(\ref{eq:good_ubound}),  for $\omega \in T_m(\zeta)$ and  $n > m$, we get:
\begin{equation}\label{eq:Zdn_upper_bound_1}
Z_{d}^{[n]}(\omega) \leq \kappa^{(n-m)-|A_{m,n}(\omega)|}\zeta^{|B_{m,n}(\omega)|}\zeta.
\end{equation}
Now, we want to upper-bound the right hand side term of the above inequality, by providing lower-bounds for the $|A_{m,n}(\omega)|$ and $|B_{m,n}(\omega)|$ values.

\medskip\noindent{\em $|A_{m,n}(\omega)|$ lower-bound:} Let $A^{[k]}(\omega) := \omega_k$, hence $|A_{m,n}(\omega)| = \sum_{k=m+1}^n A^{[k]}(\omega)$. Fix any $\alpha \in (0, 1/2)$, and let
\begin{equation*}
\mathcal{A}_{m,n}(\alpha) :=\! \left\{ \omega \in \Omega \Bigm| \sum_{k=m+1}^n \!\!A^{[k]}(\omega) \geq  \left(\frac{1}{2} - \alpha\right)\! (n-m) \right\}.
\end{equation*}
Hence, for any $\omega\in \mathcal{A}_{m,n}(\alpha)$,
\begin{equation}\label{eq:Amn_lower_bound}
|A_{m,n}(\omega)| \geq (1/2 - \alpha)(n-m).
\end{equation}
Moreover, by Lemma~\ref{lemma:mean_independent_rv}, $P\left( \mathcal{A}_{m,n}(\alpha) \right) \geq 1 - e^{-2\alpha^2(n-m)}$.

\medskip\noindent{\em $|B_{m,n}(\omega)|$ lower-bound:} First, note that $d_k$ is defined depending on $\omega_{k+1}$ value. Hence, we may write
\begin{align*}
B_{m,n}(\omega) &=  \left\{ k \in \{m+1,\dots,n\} \mid \omega_k = 1 \text{ and } d_k \in \{1,2\} \right\} \\
  &\supseteq \left\{ k \in \{m+1,\dots,n-1\} \mid \omega_k = 1, \omega_{k+1} = 1, \text{ and } d_k \in \{1,2\} \right\} \\
  &= \left\{ k \in \{ m+1,\dots,n-1 \} \mid \omega_k = 1, \omega_{k+1} = 1, \text{ and } \delta_3\left(\Gamma_{\omega_1\cdots \omega_{k}}\right) \in \{1,2\} \right\}.
\end{align*}  
Let $B^{[k]}$ be the Bernoulli random variable on $\Omega$, defined by
\begin{equation*}
B^{[k]}(\omega) := \mathbf{1}_{\{\omega_{k+1} = 1\}} \mathbf{1}_{\{\omega_{k} = 1\}} \mathbf{1}_{\{\delta_3(\Gamma_{\omega_1\cdots \omega_k}) \in \{1,2\}\}}.
\end{equation*}
The expected value of $B^{[k]}$ is given by
\begin{align*}
\mbE B^{[k]} &= \frac{1}{2^{k+1}} \sum_{i_1\cdots i_k i_{k+1}} \mathbf{1}_{\{i_{k+1} = 1\}} \mathbf{1}_{\{i_{k} = 1\}} \mathbf{1}_{\{\delta_3(\Gamma_{i_1\cdots i_k}) \in \{1,2\}\}} \\
 &= \frac{1}{2^{k+1}} \sum_{i_1\cdots i_{k-1}} \mathbf{1}_{\{\delta_3(\Gamma_{i_1\cdots i_{k-1} 1}) \in \{1,2\}\}} \\
 &= \frac{1}{2^{k+1}} \sum_{i_1\cdots i_{k-1}} \Delta^{i_1\cdots i_{k-1} 1}(\bm{\Gamma}).
\end{align*}
Since $\bm{\Gamma} \in \overline{\Pi}_{m(\gamma,\varepsilon)}(\gamma)$, for $k > m \geq m(\gamma,\epsilon)$, we get  
\begin{equation*}
\mbE B^{[k]} \geq \gamma_0 := \frac{1}{4}\left(\frac{2}{3} - \gamma \right).
\end{equation*}
Let $\mathcal{K}(m,n) := \{ k \in m+1,\dots, n-1 \mid k = m+1 \mod 2\}$, the set of integers $m+1,m+3,\dots$ comprised between $m+1$ and $n-1$. Random variables $B^{[k]}$, $k\in \mathcal{K}(m,n)$, are independent, and the expected value of their mean, denoted $\mbE B_{\mathcal{K}(m,n)} := \frac{1}{|\mathcal{K}(m,n)|}\mbE B^{[k]} $, satisfies $\mbE B_{\mathcal{K}(m,n)} \geq \gamma_0$. Fix any $\beta \in (0, \gamma_0)$, and let
\begin{equation*}
\mathcal{B}_{m,n}(\beta) \!:=\! \left\{ \omega \in \Omega \Bigm|\! \sum_{k\in \mathcal{K}(m,n)}\!\!\!\! B^{[k]}(\omega) \geq  (\gamma_0 - \beta) |\mathcal{K}(m,n)| \right\}.
\end{equation*}
Hence, for $m \geq m(\gamma,\epsilon)$ and $\omega \in \mathcal{B}_{m,n}(\beta)$, we have\footnote{The last inequality could be tighten, but we only need a non-zero fraction of $n-m$. } 
\begin{equation}
|B_{m,n}(\omega)| \geq \sum_{k=m+1}^{n-1} B^{[k]}(\omega) \geq \sum_{k\in \mathcal{K}(m,n)} B^{[k]}(\omega)  
\geq (\gamma_0 - \beta) |\mathcal{K}(m,n)| \geq (\gamma_0 - \beta)\frac{n-m}{3}. \label{eq:Bmn_lower_bound}
\end{equation}
Moreover, by applying Lemma~\ref{lemma:mean_independent_rv}, we have
\begin{align*}
P\left( \mathcal{B}_{m,n}(\beta) \right)  
&\geq  P\left( \sum_{k\in \mathcal{K}(m,n)} B^{[k]}(\omega) \geq  (\mbE B_{\mathcal{K}(m,n)} - \beta) |\mathcal{K}(m,n)| \right) \\
&\geq 1 - e^{-2\beta^2|\mathcal{K}(m,n)|} \\
&\geq 1 - e^{-2\beta^2 \frac{n-m}{3}}.
\end{align*}
We define $\mathcal{U}_{m,n}(\zeta, \alpha, \beta) := T_m(\zeta) \cap \mathcal{A}_{m,n}(\alpha) \cap \mathcal{B}_{m,n}(\beta)$.  Using (\ref{eq:Zdn_upper_bound_1}), (\ref{eq:Amn_lower_bound}), and (\ref{eq:Bmn_lower_bound}), for $n > m \geq m(\gamma,\epsilon)$ and $\omega \in \mathcal{U}_{m,n}(\zeta, \alpha, \beta)$, we have
\begin{equation*}
Z_{d}^{[n]}(\omega) \leq \kappa^{(\alpha+\frac{1}{2})(n-m)}\zeta^{\frac{\gamma_0 - \beta}{3}(n-m)}\zeta =\!\! \left( \kappa^{\alpha + \frac{1}{2}} \zeta^{\frac{\gamma_0 - \beta}{3}} \right)^{n-m} \!\!\!\zeta.
\end{equation*}
Note that $\alpha, \beta$, and $\gamma$ (thus, $\gamma_0$) are some fixed constants. Hence, for any $\theta > 0$ (as in the fast polarization property), we may choose $\zeta > 0$, such that $\kappa^{\alpha + \frac{1}{2}} \zeta^{\frac{\gamma_0 - \beta}{3}} \leq 2^{-(1+\theta)}$.
Using $Z^{[n]}(\omega) \leq \max_{d=1,2,3} Z_{d}^{[n]}(\omega)$, we get the following inequality, that holds for any $n > m \geq m(\gamma,\epsilon)$ and any $\omega \in \mathcal{U}_{m,n}(\zeta, \alpha, \beta)$:
\begin{equation*}
Z^{[n]}(\omega) \leq c 2^{-n(1+\theta)} = cN^{-(1+\theta)}.
\end{equation*}
where $c = c(m,\alpha,\beta,\gamma,\zeta) := \left( \kappa^{\alpha + \frac{1}{2}} \zeta^{\frac{\gamma_0 - \beta}{3}} \right)^{-m} \zeta$, and $N = 2^n$. Note that $\alpha,\beta,\gamma$, and $\zeta$ have been fixed at this point, and only the value of $m$ can still be varied.

\medskip To complete the proof, we need to show that $\mathcal{U}_{m,n}(\zeta, \alpha, \beta)$ is sufficiently large  (for some $m$, and large enough $n>m$), so that we may find information sets $\mathcal{I}_N$ of size $|\mathcal{I}_N| \geq RN$, for $R < \mathtt{I}(W)$. For this, we need the following lemma, which is essentially the same as Lemma 1 in \cite{arikan09}, and the proof follows using exactly the same arguments as in {\em loc. cit.} (and also using the fact that $\bm{\Gamma}$ is a polarizing sequence).

\begin{lemma} \label{lemma:Tm_zeta_large_enough}
For any fixed $\zeta > 0$ and any $0 \leq \delta < \mathtt{I}(W)$, there exists an integer $m_0(\zeta, \delta)$, such that
\begin{equation*}
P\left( T_{m_0}(\zeta) \right) \geq \mathtt{I}(W) - \delta.
\end{equation*}
\end{lemma}

Therefore, $P\left( T_{m}(\zeta) \right)$ can be made arbitrarily  close to $\mathtt{I}(W)$, by taking $m$ large enough, and once we have made $P\left( T_{m}(\zeta) \right)$ as close as desired to $\mathtt{I}(W)$, we can make $P\left( \mathcal{A}_{m,n}(\alpha) \right)$ and $P\left( \mathcal{B}_{m,n}(\beta) \right)$ arbitrarily close to $1$, by taking $n > m$ large enough. Hence, for any $R <  \mathtt{I}(W)$, we may find $m_0 = m_0(\zeta, R)$ and $n_0 = n_0(m_0, \alpha, \beta, \gamma) > m_0$, such that 
\begin{equation*}
P\left( \mathcal{U}_{m_0,n}(\zeta, \alpha, \beta)  \right) > R, \ \ \forall n \geq n_0,
\end{equation*}
and since we may assume that $m_0 \geq m(\gamma,\varepsilon)$, we also have
\begin{equation}\label{eq:znomega_ubound}
Z^{[n]}(\omega) \leq  c_0 N^{-(1+\theta)},  \ \ \forall n \geq n_0, \ \forall \omega \in \mathcal{U}_{m_0,n}(\zeta, \alpha, \beta) 
\end{equation}
where $c_0 := c_0(m_0,\alpha,\beta,\gamma,\zeta)$.

\medskip Now, for $n > 0$, let $\mathcal{V}_n := \{ \omega \in \Omega \mid Z^{[n]}(\omega) \leq c_0 N^{-(1+\theta)} \}$. Using~(\ref{eq:znomega_ubound}), we have that $\mathcal{U}_{m_0,n}(\zeta, \alpha, \beta) \subseteq \mathcal{V}_n$, for any $n \geq n_0$, and therefore $P\left[ \mathcal{V}_n \right] \geq R$. On the other hand, 
\begin{align*}
P\left[ \mathcal{V}_n \right] &= \sum_{i_1\cdots i_n \in \{0,1\}^n} \frac{1}{2^n} \mathbf{1}\left\{ Z(W^{(i_1\cdots i_n)}) \leq c_0 N^{-(1+\theta)} \right\} \\
&= \frac{1}{N}|\mathcal{I}_N|,
\end{align*}
where $\mathcal{I}_N := \left\{ i\in \{0,\dots,N-1\} \mid Z(W^{(i)}) \leq cN^{-(1+\theta)} \right\}$. It follows that $|\mathcal{I}_N| \geq RN$, for $n\geq n_0$.

\medskip We have shown that, given $\varepsilon > 0$, the fast polarization property holds for any $\bm{\Gamma} \in \Gamma(\mathcal{S})^\infty_{\text{pol}} \cap \overline{\Pi}_{m(\gamma,\varepsilon)}(\gamma)$, with $P\left( \overline{\Pi}_{m(\gamma,\varepsilon)}(\gamma) \right) \geq 1 -\varepsilon$. We then conclude that it holds for any\break $\bm{\Gamma} \in \Gamma(\mathcal{S})^\infty_{\text{pol}} \bigcap \left(\bigcup_{\epsilon > 0} \overline{\Pi}_{m(\gamma,\varepsilon)}(\gamma)\right)$, which is a measurable subset of $\Gamma(\mathcal{S})^\infty_\text{pol}$, of same probability. 
\ \hfill $\qed$

\printbibliography


\end{document}